\documentclass[a4paper,
11pt,openright]{article} 
\usepackage[english]{babel} 
\usepackage[latin1]{inputenc} 
\usepackage{amsmath,amssymb,amsthm,amsfonts,amscd,authblk,dsfont,color,cancel,diagbox,enumerate,epsfig,eucal,fancyhdr,latexsym,lmodern,mathabx,mathrsfs,mathtools,multicol,multirow,musicography,phaistos,skull,simpler-wick,subcaption,tikz-cd,ulem,xfrac}
\usepackage[multiuser]{fixme}
\FXRegisterAuthor{nd}{annd}{Nico} 
\FXRegisterAuthor{cvdv}{ancvdv}{Chris}
\FXRegisterAuthor{lp}{anlp}{Lo}
\usepackage[
final=true,                    
plainpages=false,              
pdfpagelabels=true,            
pdfencoding=auto,              
unicode=true,                  
hypertexnames=true,            
naturalnames=true              
]{hyperref}
\usepackage[all,cmtip]{xy}
\usetikzlibrary{matrix,calc}
\usepackage[autostyle, english = american]{csquotes}
\MakeOuterQuote{"} 

\definecolor{NiColor}{RGB}{77,77,255}
\definecolor{NiColoRed}{RGB}{255,77,77}
\definecolor{NiCitation}{RGB}{0,181,26}

\newtheoremstyle{TheoremStyle}
{\topsep}
{\topsep}
{}
{}
{\sc}
{:}
{.5em}
{}
\makeatletter
\def\@endtheorem{\begin{flushright}$\diamond$\end{flushright}} 
\makeatother

\theoremstyle{TheoremStyle}
\newtheorem{theorem}{Theorem}[section]

\newtheorem{proposition}[theorem]{Proposition}
\newtheorem{lemma}[theorem]{Lemma}

\newtheorem{remark}[theorem]{Remark}
\newtheorem{example}[theorem]{Example} 

\title{Classical and Quantum KMS States on Spin Lattice Systems}
\author[a,b]{N Drago\thanks{\href{mailto:nicolo.drago@unige.it}{nicolo.drago@unige.it}}}
\author[a]{L. Pettinari\thanks{\href{mailto:lorenzo.pettinari@unitn.it}{lorenzo.pettinari@unitn.it}}}
\author[c,d]{C. J. F. van de Ven\thanks{\href{mailto:christiaan.van-de-ven@uni-tuebingen.de}{christiaan.van-de-ven@uni-tuebingen.de}}}
\affil[a]{Dipartimento di Matematica, Universit\`a di Trento and INFN-TIFPA and INdAM, Via Sommarive 14, I-38123 Povo, Italy}
\affil[b]{Dipartimento di Matematica, Universit\`{a} di Genova and INdAM, Via Dodecaneso 35, I-16146 Genova, Italy}
\affil[c]{T\"{u}bingen University, Fachbereich Mathematik, AB Mathematische Physik, Auf der Morgenstelle 10,
72076 T\"{u}bingen, Germany}
\affil[d]{Friedrich-Alexander-Universit\"at Erlangen-N\"urnberg, Department of Mathematics,
Cauerstra\ss e 11 91058 Erlangen, Germany}

\begin{document}
\maketitle

\begin{abstract}
We study the classical and quantum KMS conditions within the context of spin lattice systems.
Specifically, we define a strict deformation quantization (SDQ) for a $\mathbb{S}^2$-valued spin lattice system over $\mathbb{Z}^d$ generalizing the renown Berezin SDQ for a single sphere.
This allows to promote a classical dynamics on the algebra of classical observables to a quantum dynamics on the algebra of quantum observables.
We then compare the notion of classical and quantum thermal equilibrium by showing that any weak*-limit point of a sequence of quantum KMS states fulfils the classical KMS condition.
In short, this proves that the semiclassical limit of quantum thermal states describes classical thermal equilibrium, strenghtening the physical interpretation of the classical KMS condition.
Finally we provide two sufficient conditions ensuring uniqueness of classical and quantum KMS states: The latter are based on an version of the Kirkwood-Salzburg equations adapted to the system of interest.
As a consequence we identify a mild condition which ensures uniqueness of classical KMS states and of quantum KMS states for the quantized dynamics for a common sufficiently high temperature.
\end{abstract}

\tableofcontents

\section{Introduction}
\label{Sec: introduction}

The description of thermal equilibrium is a well-established and extensively studied topic in classical and quantum statistical mechanics \cite{Bratteli_Robinson_97,Friedli_Velenik_2017,Israel_1979,Ruelle_1967}.
Adopting the algebraic approach, where observables of the physical system of interest are modelled by a $C^*$-algebra $\mathfrak{A}$, classical and quantum thermal equilibrium are characterized by two slightly different yet related conditions, called classical and quantum Kubo-Martin-Schwinger (KMS) conditions, \textit{cf.} \cite{Gallavotti_Verboven_1975,Haag_Hugenholtz_Winnik_1967}.

Specifically, a quantum system is described in term of an abstract non-commutative $C^*$-algebra $\mathfrak{A}$ ---thorough all paper we will only be interested in algebras with a unit.
Time evolution is modelled by a strongly continuous one-parameter group $t\mapsto\tau_t$ of $*$-automorphisms on $\mathfrak{A}$ with infinitesimal generator $\delta$.
Within this setting a state $\omega\in S(\mathfrak{A})$ ---that is, a linear, positive and normalized functional $\omega\colon\mathfrak{A}\to\mathbb{C}$--- is called \textbf{$(\beta,\delta)$-KMS quantum state}, $\beta\in[0,\infty)$, if
\begin{align}\label{Eq: KMS condition - quantum on abstract C* algebra}
	\omega\big(\mathfrak{a}\tau_{i\beta}(\mathfrak{b})\big)
	=\omega(\mathfrak{b}\mathfrak{a})\,,
\end{align}
for all pairs $\mathfrak{a},\mathfrak{b}\in\mathfrak{A}$ of analytic elements for $\tau$, \textit{cf.} \cite{Bratteli_Robinson_97}.
The quantum $(\beta,\delta)$-KMS condition \eqref{Eq: KMS condition - quantum on abstract C* algebra} selects those states on $\mathfrak{A}$ which are interpreted as describing thermal equilibrium with respect to $\tau$ at a fixed inverse temperature $\beta$, \textit{cf.} \cite{Haag_Hugenholtz_Winnik_1967} ---here $\beta=0$ corresponds to infinite temperature.

The description of thermal equilibrium for a classical system is slightly different, \textit{cf.} \cite{Aizenman_Gallavotti_Goldstein_Lebowitz_1976,Aizenman_Goldstein_Gruber_Lebowitz_Martin_1977,Drago_Van_de_Ven_2023}.
In this scenario the observables of a classical physical system are described by a commutative Poisson $C^*$-algebra $\mathfrak{A}$.
We recall that a \textbf{Poisson structure} over a commutative $C^*$-algebra $\mathfrak{A}$ is given by a bilinear map $\{\;,\;\}\colon \dot{\mathfrak{A}}\times \dot{\mathfrak{A}}\to\dot{\mathfrak{A}}$ defined on a dense $*$-subalgebra $\dot{\mathfrak{A}}\subset\mathfrak{A}$ which fulfils:
\begin{multline*}
	\{\mathfrak{a},\mathfrak{b}\}
	=-\{\mathfrak{b},\mathfrak{a}\}\,,
	\qquad
	\{\mathfrak{a},\mathfrak{b}\}^*=\{\mathfrak{a}^*,\mathfrak{b}^*\}\,,
	\qquad
	\{\mathfrak{a},\mathfrak{b}\mathfrak{c}\}
	=\{\mathfrak{a},\mathfrak{b}\}\mathfrak{c}+\mathfrak{b}\{\mathfrak{a},\mathfrak{c}\}\,,
	\\
	\{\mathfrak{a},\{\mathfrak{b},\mathfrak{c}\}\}
	=\{\{\mathfrak{a},\mathfrak{b}\},\mathfrak{c}\}
	+\{\mathfrak{b},\{\mathfrak{a},\mathfrak{c}\}\}\,,
\end{multline*}
for all $\mathfrak{a},\mathfrak{b},\mathfrak{c}\in \dot{\mathfrak{A}}$.
Given $\beta\in[0,+\infty)$ and a $*$-derivation $\delta\colon\dot{\mathfrak{A}}\to\mathfrak{A}$, a state $\omega\in S(\mathfrak{A})$ is called \textbf{$(\beta,\delta)$-KMS classical state} if
\begin{align}\label{Eq: KMS condition - classical on abstract C* algebra}
	\omega(\{\mathfrak{a},\mathfrak{b}\})
	=\beta\omega(\mathfrak{b}\delta(\mathfrak{a}))\,,
\end{align}
for all $\mathfrak{a},\mathfrak{b}\in\dot{\mathfrak{A}}$.
Once again $\delta$ is regarded as the infinitesimal generator of the time evolution on $\mathfrak{A}$ and $\beta$ is interpreted as an inverse temperature.

The quantum KMS condition has received a lot of attention and has been investigated in several scenarios \cite{Bratteli_Robinson_97,Israel_1979,Ruelle_1967}.
In particular the physical content of Equation \eqref{Eq: KMS condition - quantum on abstract C* algebra} has been investigated, \textit{cf.} \cite{Haag_Hugenholtz_Winnik_1967,Pusz_Woronowicz_1978}, providing concrete justifications for its interpretation.
Conversely, there are fewer investigations on the classical KMS condition \eqref{Eq: KMS condition - classical on abstract C* algebra}.
The latter has been introduced in \cite{Gallavotti_Verboven_1975} and further developed in \cite{Aizenman_Gallavotti_Goldstein_Lebowitz_1976,Aizenman_Goldstein_Gruber_Lebowitz_Martin_1977,Fannes_Pule_Verbeure_1977,Gallavotti_Pulvirenti_1976,Pulvirenti_Riela_1977} in the context classical system of infinitely many particles.
The classical KMS condition has been investigated also in the context of pure Poisson geometry in \cite{Bayen_Flato_Fronsdal_Lichnerowicz_78I, Bayen_Flato_Fronsdal_Lichnerowicz_78II,
Basart_Flato_Lichnerowicz_Sternheimer_84,Basart_Lichnewicz_85,Bordermann_Hartmann_Romer_Waldmann_98,Drago_Waldmann_2024}, moreover, its relation with the Dobrushin-Landford-Ruelle (DLR) \cite{Dobrushin_1968_1,Dobrushin_1968_2,Dobrushin_1968_3,Lanford_Ruelle_1969} probabilistic approach to classical thermal equilibrium was investigated in \cite{Drago_vandeVen_2023}
while its connection to the notion of Gibbs measures for non-linear Hamiltonian systems was studied in \cite{Ammari_Sohinger_2023,Arsernev_1984,Chueshov_1987,Peskov_1985}.

The physical justification of the classical KMS condition \eqref{Eq: KMS condition - classical on abstract C* algebra} seems to be less studied.
In particular, in \cite{Gallavotti_Verboven_1975} condition \eqref{Eq: KMS condition - classical on abstract C* algebra} has been formally derived by considering a suitable semiclassical limit of the quantum KMS condition \eqref{Eq: KMS condition - quantum Gibbs state}.
The first goal of this paper is to provide a mathematically rigorous version of this derivation within the setting of strict deformation quantization, \textit{cf.} Theorem \ref{Thm: limit points of quantum KMS in oo volume are classical KMS in oo volume}.

Strict (or $C^*$-algebraic) deformation quantization  (SDQ) provides a mathematically rigorous setting to study the quantization of a classical system \cite{Landsman_1998,Rieffel_94}.
This framework is not only suitable to investigate the semiclassical limit of states of a quantum system with a fixed, but arbitrary number of degrees of freedom \textit{cf.} \cite{Landsman_2017,Moretti_vandeVen_2022}, but it can also be applied to describe the macroscopic properties of quantum systems over an infinitely extended lattice \cite{Drago_Van_de_Ven_2023,Landsman_Moretti_vandeVen_2020,Moretti_vandeVen_2020,vandeVen_2020,vandeVen_2022} or the semi-classical properties of condensates \cite{Pettinari_2024}.

From a mathematical point of view, a SDQ requires the notion of bundle of $C^*$-algebra which we briefly recall following \cite[App. C.19]{Landsman_2017}, \cite[\S IV.1.6]{Blackadar_2006}.
Setting $\overline{\mathbb{Z}_+/2}:=\mathbb{Z}_+/2\cup\{\infty\}$ and given a collection $\{\mathfrak{A}_j\}_{j\in\overline{\mathbb{Z}_+/2}}$ of $C^*$-algebras, we denote by $\prod_{j\in\overline{\mathbb{Z}_+/2}}\mathfrak{A}_j$ the associated \textbf{full $C^*$-direct product}, which is the $C^*$-algebra made by sequences $(\mathfrak{a}_j)_{j\in\overline{\mathbb{Z}_+/2}}$, $\mathfrak{a}_j\in\mathfrak{A}_j$, such that $\|(\mathfrak{a}_j)_{j\in\overline{\mathbb{Z}_+/2}}\|_{\prod_{j\in\overline{\mathbb{Z}_+/2}}\mathfrak{A}_j}:=\sup_j\|\mathfrak{a}_j\|_{\mathfrak{A}_j}<\infty$.
Within this setting a \textbf{continuous bundle of $C^*$-algebras} over $\overline{\mathbb{Z}_+/2}$ (with fibers $\{\mathfrak{A}_j\}_{j\in\overline{\mathbb{Z}_+/2}}$) is a $C^*$-subalgebra $\mathfrak{A}\subset\prod_{j\in\overline{\mathbb{Z}_+/2}}\mathfrak{A}_j$ such that:
(i) $(\|\mathfrak{a}_j\|_{j\in\overline{\mathbb{Z}_+/2}})\in C(\overline{\mathbb{Z}_+/2})$ for all $(\mathfrak{a}_j)_{j\in\overline{\mathbb{Z}_+/2}}\in\mathfrak{A}$, where $C(\overline{\mathbb{Z}_+/2})$ denotes the space of sequences $(\alpha_j)_{j\in\overline{\mathbb{Z}_+/2}}$, $\alpha_j\in\mathbb{C}$, such that $\alpha_\infty=\lim_{j\to\infty}\alpha_j$;
(ii) $\alpha \mathfrak{a}\in\mathfrak{A}$ for all $\mathfrak{a}\in\mathfrak{A}$ and $\alpha\in C(\overline{\mathbb{Z}_+/2})$.
A \textbf{strict deformation quantization} (SDQ) is then defined by the following data:
\begin{enumerate}\label{data on SDQ}
	\item
	A commutative Poisson $C^*$-algebra $\mathfrak{A}_\infty$, with Poisson structure $\{\;,\;\}\colon\dot{\mathfrak{A}}_\infty\times\dot{\mathfrak{A}}_\infty\to\dot{\mathfrak{A}}_\infty$;
	
	\item
	A continuous bundle of $C^*$-algebras \cite{Dixmier_1977} $\widetilde{\mathfrak{A}}\subset\prod_{j\in\overline{\mathbb{Z}_+/2}}\mathfrak{A}_j$;
	
	\item
	A family of linear maps, called \textbf{quantization maps}, $Q_j\colon\dot{\mathfrak{A}}_\infty\to\mathfrak{A}_j$, $j\in\overline{\mathbb{Z}_+/2}$, such that:
	\begin{enumerate}
		\item\label{Item: quantization maps are Hermitian and define a continuous section}
		$Q_\infty(\mathfrak{a})=\mathfrak{a}$ for all $\mathfrak{a}\in\dot{\mathfrak{A}}_\infty$, moreover, $Q_j(\mathfrak{a})^*=Q_j(\mathfrak{a}^*)$ and $(Q_j(\mathfrak{a}))_{j\in\overline{\mathbb{Z}_+/2}}\in\widetilde{\mathfrak{A}}$.
		
		\item\label{Item: quantization maps fulfils the DGR condition}
		For all $\mathfrak{a},\mathfrak{b}\in\dot{\mathfrak{A}}_\infty$ the \textbf{Dirac-Groenewold-Rieffel (DGR) condition} holds:
		\begin{align}\label{Eq: Dirac-Groenewold-Rieffel condition}
			\lim_{j\to\infty}\big\|Q_j(\{\mathfrak{a},\mathfrak{b}\})-i (2j+1)[Q_j(\mathfrak{a}),Q_j(\mathfrak{b})]\big\|_{\mathfrak{A}_j}=0\,.
		\end{align}
		
		\item\label{Item: quantization maps are strict}
		For all $j\in\overline{\mathbb{Z}_+/2}$, $Q_j(\dot{\mathfrak{A}})$ is a dense $*$-subalgebra of $\mathfrak{A}_j$.
	\end{enumerate}
\end{enumerate}
In this framework $j\in\mathbb{Z}_+/2$ is interpreted as a semiclassical parameter ---in fact, $h_j:=1/(2j+1)$ is the proper semiclassical parameter--- and $j\to\infty$ corresponds to the semiclassical limit.
Given a sequence $(\omega_j)_{j\in\mathbb{Z}_+/2}$ of states such that $\omega_j\in S(\mathfrak{A}_j)$, the semiclassical limit is obtained considering the weak*-limit points of the sequence $(\omega_j\circ Q_j)_{j\in\mathbb{Z}_+/2}$ of functionals over $\mathfrak{A}_\infty$ ---each such weak*-limit point defines a state on $\mathfrak{A}_\infty$.

In this paper we provide a concrete, yet general SDQ model where the relation between classical and quantum KMS conditions \eqref{Eq: KMS condition - quantum on abstract C* algebra}-\eqref{Eq: KMS condition - classical on abstract C* algebra} can be studied rigorously.
In particular we will focus on $\mathbb{S}^2$-valued spin lattice systems on $\Gamma:=\mathbb{Z}^d$, $d\in\mathbb{N}$, \textit{cf.} \cite{Bratteli_Robinson_97, Friedli_Velenik_2017}, which are described by the renown quasi-local algebras $B_\infty^\Gamma$, $B_j^\Gamma$.
The latter are $C^*$-inductive limits of corresponding $C^*$-inductive systems $\{B_\infty^\Lambda\}_{\Lambda\Subset\Gamma}$, $\{B_j^\Gamma\}_{\Lambda\Subset\Gamma}$ where, for any finite region $\Lambda\Subset\Gamma$, $B_\infty^\Gamma$ (\textit{resp.} $B_j^\Lambda$) denotes the algebra of observables localized in $\Lambda$ for the classical (\textit{resp.} quantum) system, \textit{cf.} Section \ref{Subsec: classical and quantum lattice systems on Gamma}.
Within this setting classical and quantum KMS states have been investigated in detail \cite{Bratteli_Robinson_97,Friedli_Velenik_2017}.
Moreover, this setting fits within the framework of the Berezin quantization, which identifies a SDQ for the physical system associated with finite regions, \textit{cf.} \cite{Berezin_1975,Murro_vandeVen_2022}.

Within this framework we may summarize our results as follows:
\begin{enumerate}[(I)]
	\item\label{Item: our result - SDQ}
	In Theorem \ref{Thm: Berezin SDQ on Gamma} we construct a SDQ for the spin lattice system associated to the infinite region $\Gamma$, extending the results of \cite{Berezin_1975,Murro_vandeVen_2022}.
        This completes the framework in which we will subsequently investigate the properties of the semiclassical limits of KMS quantum states.
        It is worth to mention that the study of thermal equilibrium leads to physically relevant results only for infinitely extended systems: Thus, the construction of our SDQ is well-suited for the purposes of this study, \textit{cf.} \eqref{Item: our result - weak*-limit of KMS quantum states}-\eqref{Item: our result - absence of CPT implies absence of QPT} below.
	
	\item\label{Item: our result - weak*-limit of KMS quantum states}
	We study the properties of weak*-limit points of KMS quantum states, in particular, in Theorem \ref{Thm: limit points of quantum KMS in oo volume are classical KMS in oo volume} we prove that they all fulfil the KMS classical condition \eqref{Eq: KMS condition - classical on abstract C* algebra}.
	This provides a rigorous derivation of the classical KMS condition from the quantum KMS condition along the line of \cite{Gallavotti_Verboven_1975}.
	
	\item\label{Item: our result - absence of CPT implies absence of QPT}
	We investigate further the relationship between classical and quantum thermal equilibrium with a specific focus on phase transitions.
	The latter describe the uniqueness/non-uniqueness of KMS states and are of utmost relevance for describing when a physical system undergoes an abrupt change in its macroscopic behaviour, \textit{e.g.} gas-to-liquid condensation.
	It is common folklore that classical and quantum phase transition at non-vanishing temperature should be in bijection.
	Yet, to the best of our knowledge, no mathematically rigorous proof of this claim has been given.
	In Section \ref{Sec: Common absence of CPTs and of QPTs} we prove that, within the model $\mathbb{S}^2$-valued lattice spin system considered in Section \ref{Sec: Berezin SDQ on a lattice system}, under a sufficiently mild assumption, \textit{cf.} \eqref{Eq: uniqueness result for classical KMS states - assumption on potential}, classical and quantum KMS states are unique for temperatures higher than a common sufficiently high threshold temperature.
	This results is a consequence of Theorems \ref{Thm: uniqueness result for classical KMS states}, \ref{Thm: uniqueness result for quantum KMS states} which provides two new sufficient conditions for the uniqueness of KMS classical and quantum states.
	Thus, our result is in line with the claimed equivalence between classical and quantum phase transitions.
\end{enumerate}

It is worth to point out that our results are companions of other existing works in this area.
In particular, \cite{Murro_vandeVen_2022} already described an abstract framework which covers the Berezin SDQ for a lattice system in a finite region.
Our result \eqref{Item: our result - SDQ} generalizes this setting for a specific model but allowing to deal with a spin lattice system on an infinitely extended region: This is important because physically interesting results on thermal equilibrium, \textit{e.g.} phase transitions, can only be described on infinitely extended system.
Concerning \eqref{Item: our result - weak*-limit of KMS quantum states}, in \cite{Falconi_2018} the semiclassical limit of KMS quantum states has been investigated under an assumption which is an abstract version of our Lemma \ref{Lem: classical limit of local dynamics in oo-volume}.
In \cite{Ammari_Ratsimanetrimanana_2019} the classical and quantum KMS conditions were related for the case of the Bose-Hubbard system on a finite graph.
Similarly, \cite{vandeVen_2024} deals with the semiclassical limit of Gibbs quantum and classical states, \textit{i.e.} KMS states on a spin lattice system associated with a finite region: From this point of view, our result \eqref{Item: our result - weak*-limit of KMS quantum states} can be seen as a generalization of \cite{vandeVen_2024} to a physically more interesting scenario.
Finally, concerning \eqref{Item: our result - absence of CPT implies absence of QPT}, Theorems \ref{Thm: uniqueness result for classical KMS states}, \ref{Thm: uniqueness result for quantum KMS states} are inspired by \cite[Prop. 6.2.45]{Bratteli_Robinson_97}, which provides a sufficient condition for uniqueness of KMS quantum states with an argument based on a quantum version of the Kirkwood-Salzburg equations.
Our result provides a classical analogous of \cite[Prop. 6.2.45]{Bratteli_Robinson_97}, moreover, it strengths some of its conclusion ---specifically the $j$-dependence of the inverse critical temperature, \textit{cf.} Remark \ref{Rmk: NO CPT implies NO QPT; comparison with BR result}.
It is also worth to point out that our results cover the regime of high temperatures.
For low temperatures it was shown in \cite{Biskup_Chayes_Starr_2007} that, whenever chessboard estimates can be used to prove a phase transition in the classical model, the corresponding quantum model will have a similar phase transition provided $\beta^2\ll j$.

\bigskip

The paper is organized as follows.
Section \ref{Sec: Berezin SDQ on a lattice system} describes the model of interest and construct a SDQ suitable for our purposes.
Section \ref{Sec: The semiclassical limit of the quantum KMS condition} deals with the semiclassical limit of KMS quantum states, proving that each weak*-limit point fulfils the classical KMS condition \eqref{Eq: KMS condition - classical on abstract C* algebra}.
This requires a few technical results, \textit{cf.} Lemmata \ref{Lem: classical auto-correlation lower bound}-\ref{Lem: classical limit of local dynamics in oo-volume}.
Finally Section \ref{Sec: Common absence of CPTs and of QPTs} deals with the topic of classical and quantum phase transitions.
In particular, Section \ref{Subsec: uniqueness result for classical KMS state} is devoted to the proof of a uniqueness result for KMS classical states, while Section \ref{Subsec: uniqueness result for quantum KMS state} deals with an analogous result for KMS quantum states.
Eventually the relation between these results is discussed, \textit{cf.} Remark \ref{Rmk: NO CPT implies NO QPT; comparison with BR result}, leading to the proof that, under reasonably mild assumptions, classical and quantum phase transitions are absent for temperatures higher than a common threshold temperature.

\paragraph{Acknowledgements.}
We are indebted with V. Moretti for many helpful discussions on this project.
N.D. and L. P. acknowledge the support of the GNFM group of INdAM.
C. J. F. van de Ven is supported by a postdoctoral fellowship granted by the Alexander von Humboldt Foundation (Germany).

\paragraph{Data availability statement.}
Data sharing is not applicable to this article as no new data were created or analysed in this study.

\paragraph{Conflict of interest statement.}
The authors certify that they have no affiliations with or involvement in any
organization or entity with any financial interest or non-financial interest in
the subject matter discussed in this manuscript.

\section{Berezin SDQ on a lattice system}
\label{Sec: Berezin SDQ on a lattice system}

The goal of this section is to provide a SDQ in the sense of Berezin for a classical lattice system.
For definiteness we will consider the lattice $\mathbb{Z}^d$, $d\in\mathbb{N}$, where to each site $x\in\mathbb{Z}^d$ one associates the spin space $\mathbb{S}^2$, in the classical case, or $\mathbb{C}^{2j+1}$, $j\in\mathbb{Z}_+/2$, in the quantum case.
Here $j$ will play the role of a semiclassical parameter, the semiclassical limit being $j\to\infty$.
Within this setting the Berezin SDQ $Q_j\colon C(\mathbb{S}^2)\to M_{2j+1}(\mathbb{C})$ identifies a classical-to-quantum map between the algebras of observables associated to the system \cite{Berezin_1975}.
We will prove that such SDQ can be lifted to a quantization on the associated quasi local-algebra on the whole $\Gamma$, \textit{cf.} Proposition \ref{Thm: Berezin SDQ on Gamma}, \textit{i.e.} for the infinitely extended classical system.
This generalizes \cite{Murro_vandeVen_2022} where the Berezin SDQ for a classical system localized in a finite region $\Lambda\Subset\mathbb{Z}^d$ was considered.

In section \ref{Subsec: classical and quantum lattice systems on Gamma} we briefly introduce the data for the classical and quantum system on a single site $x\in\mathbb{Z}^4$.
Section \ref{Subsec: Berezin SDQ for a single site system} recollects the relevant properties of the standard Berezin deformation quantization, which will play a role also for the discussion in Section \ref{Sec: Common absence of CPTs and of QPTs}.
Finally, Section \ref{Subsec: Berezin SDQ on Gamma} extends the result of \cite{Murro_vandeVen_2022} by constructing a Berezin SDQ for the infinitely extended classical system on $\Gamma$.

\subsection{Classical and quantum lattice systems on $\Gamma$}
\label{Subsec: classical and quantum lattice systems on Gamma}

In this section we will briefly summarize the data of the classical and quantum spin system we will consider for the rest of the paper, \textit{cf.} \cite{Bratteli_Robinson_97,Drago_Van_de_Ven_2023,Friedli_Velenik_2017}.

At a classical level, we will consider the lattice $\boxed{\Gamma}:=\mathbb{Z}^d$, $d\in\mathbb{N}$.
The spin configuration space $\boxed{\mathbb{S}_x^2}:=\mathbb{S}^2$ at each $x\in\mathbb{Z}^d$ is a closed symplectic manifold.
By definition, the algebra of classical observables at $x\in\mathbb{Z}^d$ is the $C^*$-algebra $\boxed{B_\infty}:=C(\mathbb{S}_x^2)$, the latter being also a Poisson $C^*$-algebra with Poisson bracket $\{\;,\;\}_{B_\infty}$ defined on $\boxed{\dot{B}_\infty}:=C^\infty(\mathbb{S}_x^2)$.

For any finite region $\Lambda\Subset\mathbb{Z}^d$ the algebra of classical observables $\boxed{B_\infty^\Lambda}$ associated with $\Lambda$ is defined by $B_\infty^\Lambda:=\overline{\bigotimes_{x\in\Lambda}B_\infty}\simeq C(\mathbb{S}^2_\Lambda)$ where the spin configuration space is now $\boxed{\mathbb{S}^2_\Lambda}:=\bigotimes_{x\in\Lambda}\mathbb{S}^2\simeq(\mathbb{S}^2)^{|\Lambda|}$.
Notice that $B_\infty^\Lambda$ is a Poisson $C^*$-algebra with Poisson bracket $\{\;,\;\}_{B_\infty^\Lambda}\colon \dot{B}_\infty^\Lambda\times\dot{B}_\infty^\Lambda\to \dot{B}_\infty^\Lambda$ defined on the dense $*$-sub-algebra $\boxed{\dot{B}_\infty^\Lambda}:=C^\infty(\mathbb{S}^2_\Lambda)$ and associated with the symplectic structure of $\mathbb{S}^2_\Lambda$.

In the thermodynamic limit one identifies the $C^*$-algebra $\boxed{B_\infty^\Gamma}$ of quasi-local classical observables on $\Gamma$ with $B_\infty^\Gamma:=C(\mathbb{S}_\Gamma)$, where $\boxed{\mathbb{S}_\Gamma}:=(\mathbb{S}^2)^\Gamma$ is compact in the product topology.
It is worth to point out that $B_\infty^\Gamma$ is the $C^*$-direct limit of the $C^*$-direct system $\{B_\infty^\Lambda\}_{\Lambda\Subset\mathbb{Z}^d}$.
The latter is characterized by the $C^*$-injective maps
\begin{align}\label{Eq: inclusion map for classical observables}
	\iota^{\Lambda_0}_{\Lambda_1}\colon B_\infty^{\Lambda_0}\to B_\infty^{\Lambda_1}\,,
	\qquad
	\iota^{\Lambda_0}_{\Lambda_1}a_{\Lambda_0}
	:=a_{\Lambda_0}\otimes\bigotimes_{x\in\Lambda_1\setminus\Lambda_0}I_\infty
	\qquad
	\forall a_{\Lambda_0}\in B_\infty^{\Lambda_0}\,,
\end{align}
where $\Lambda_0\subset\Lambda_1\Subset\mathbb{Z}^d$ while $I_\infty\in B_\infty$ denotes the constant function $I_\infty\equiv 1$.
Denoting by $\iota^\Lambda\colon B_\infty^\Lambda\to B_\infty^\Gamma$ the associated $C^*$-inclusion maps we observe that
\begin{align}
	\dot{B}_\infty^\Gamma:=\bigcup_{\Lambda\Subset\mathbb{Z}^d}\iota^\Lambda\dot{B}_\infty^\Lambda\,,
\end{align}
is a dense $*$-subalgebra of $B_\infty^\Gamma$.
With a standard slight abuse of notation in what follows we will identify $B_\infty^\Lambda$ and $\iota^\Lambda B_\infty^\Lambda$, therefore, we will drop the inclusions maps $\iota^\Lambda$, $\iota^{\Lambda_1}_{\Lambda_2}$.

As described in \cite{Drago_vandeVen_2023}, $B_\infty^\Gamma$ is a Poisson $C^*$-algebra with Poisson structure defined on $\dot{B}_\infty^\Gamma$.
In particular one observes that the maps $\iota^{\Lambda_0}_{\Lambda_1}\colon B_\infty^{\Lambda_0}\to B_\infty^{\Lambda_1}$ are Poisson, namely
\begin{align}\label{Eq: Poisson structures - consistency with respect to inclusions}
	\iota^{\Lambda_0}_{\Lambda_1}\{a_{\Lambda_0},\tilde{a}_{\Lambda_0}\}_{B_\infty^{\Lambda_0}}
	:=\{\iota^{\Lambda_0}_{\Lambda_1}a_{\Lambda_0},
	\iota^{\Lambda_0}_{\Lambda_1}\tilde{a}_{\Lambda_0}\}_{B_\infty^{\Lambda_1}}\,.
\end{align}
Out of Equation \eqref{Eq: Poisson structures - consistency with respect to inclusions} the Poisson structure $\{\;,\;\}_{B_\infty^\Gamma}$ is defined by
\begin{align}\label{Eq: Poisson bracket on Cstar algebra of classical observables}
	\{\;,\;\}_{B_\infty^\Gamma}\colon\dot{B}_\infty^\Gamma\times\dot{B}_\infty^\Gamma\to\dot{B}_\infty^\Gamma\,,
	\qquad
	\{a_{\Lambda_1},a_{\Lambda_2}\}_{B_\infty^\Gamma}
	:=\{a_{\Lambda_1},a_{\Lambda_2}\}_{B_\infty^{\Lambda_1\cup\Lambda_2}}\,.
\end{align}
Notice that, in fact,
\begin{align}\label{Eq: Poisson structures - reduction property}
	\{a_{\Lambda_1},a_{\Lambda_2}\}_{B_\infty^{\Lambda_1\cup\Lambda_2}}
	=\{a_{\Lambda_1},a_{\Lambda_2}\}_{B_\infty^{\Lambda_1\cap\Lambda_2}}\,.
\end{align}

\bigskip

On the quantum side, we will consider a spin lattice system over $\Gamma$, where each site $x\in\Gamma$ is associated with a finite dimensional algebra of non-commutative observables.
Specifically, let $j\in\mathbb{Z}_+/2$ and let $\boxed{B_j}:=M_{2j+1}(\mathbb{C})$: The latter will be considered the algebra of quantum observables at each site $x\in\Gamma$ ---as we will see in Section \ref{Subsec: Berezin SDQ for a single site system} $j$ will play the role of a semiclassical parameter.
For any $\Lambda\Subset\Gamma$ we then set $\boxed{B_j^\Lambda}:=\overline{\bigotimes_{x\in\Lambda}B_j}$.
Then the collection $\{B_j^\Lambda\}_{\Lambda\Subset\Gamma}$ form a $C^*$-direct system, \textit{cf.} \cite{Bratteli_Robinson_97,Ruelle_1967}, with injective $C^*$-maps denoted by, with a slight abuse of notation,
\begin{align}\label{Eq: inclusion map for quantum observables}
	\iota^{\Lambda_0}_{\Lambda_1}\colon B_j^{\Lambda_0}\to B_j^{\Lambda_1}\,,
	\qquad
	\iota^{\Lambda_0}_{\Lambda_1}A_{\Lambda_0}
	:=A_{\Lambda_0}\otimes\bigotimes_{x\in\Lambda_1\setminus\Lambda_0}I_j
	\qquad
	\forall A_{\Lambda_0}\in B_j^{\Lambda_0}\,,
\end{align}
where $\Lambda_0\subset\Lambda_1\Subset\Gamma$ while $I_j\in B_j$ is the identity matrix.

The algebra $\boxed{B_j^\Gamma}$ of quantum observables in the thermodynamic limit is the $C^*$-direct limit of the $C^*$-direct system $\{B_j^\Lambda\}_{\Lambda\Subset\Gamma}$.
With a slight abuse of notation we will denote by $\iota^\Lambda\colon B_j^\Lambda\to B_j^\Gamma$ the associated $C^*$-inclusion maps.
In particular $\boxed{\dot{B}_j^\Gamma}:=\bigcup_{\Lambda\Subset\Gamma}\iota^\Lambda B_j^\Lambda$ is a dense $*$-algebra of $B_j^\Gamma$, moreover, $\|\iota^\Lambda a_\Lambda\|_{B_j^\Gamma}=\|a_\Lambda\|_{B_j^\Lambda}$ for all $a_\Lambda\in B_j^\Lambda$.
Similarly to the classical case we will identify $B_j^\Lambda$ and $\iota^\Lambda B_j^\Lambda$ and drop the inclusion maps $\iota^\Lambda$, $\iota^{\Lambda_0}_{\Lambda_1}$ when not strictly necessary.

\subsection{Berezin SDQ for a single site system}
\label{Subsec: Berezin SDQ for a single site system}

This section focuses on the standard Berezin quantization of the sphere $\mathbb{S}^2$ \cite{Berezin_1975,Bordemann_Meinrenken_Schlichenmaier_1994,Combescure_Robert_2012,LeFloch_2018,Perelomov_1972,Rios_Straume_2014}.
We will recall without proof the main results, pointing out useful consequences which we were not able to find in the existing literature, \textit{cf.} Remark \ref{Rmk: Berezin quantization useful remark}.

\bigskip

To begin with we consider the Lie group $SU(2)$ and denote by $\{J_i\}_{i=1}^3$ the generators of the corresponding Lie algebra $\mathfrak{su}(2)$ with commutation relations $[J_1,J_2]=iJ_3$ and extended cyclically.
Adopting the standard physicist's notation we denote by $\boxed{D^{(j)}}\colon SU(2)\to M_{2j+1}(\mathbb{C})$ the irreducible representation of $SU(2)$ of spin $j\in\mathbb{Z}_+/2$, \textit{cf.} \cite[\S 5.4]{Folland_2015}.
We will denote by
\begin{align}\label{Eq: bra-ket notation}
	|j,m\rangle\in\mathbb{C}^{2j+1}
	\qquad m\in
        \{-j,\ldots,j\}\,,
\end{align}
the orthonormal basis of $\mathbb{C}^{2j+1}$ made by the eigenvectors of $D^{(j)}(J_3)$, the latter denoting the infinitesimal generator of $D^{(j)}(e^{-iJ})$ ---we adopt the bra-ket notation, \textit{cf.} \cite{Napolitano_Sakuraki_2021}.
In particular
\begin{align*}
	D^{(j)}(J_3)|j,m\rangle=m|j,m\rangle
	\qquad
	\langle j,m|j,m'\rangle=\delta_{m,m'}\,.
\end{align*}
The \textbf{coherent state} associated with $\sigma\in\mathbb{S}^2\simeq SU(2)/U(1)$ is defined by
\begin{align}\label{Eq: coherent states}
	|j,\sigma\rangle
	:=D^{(j)}(\sigma)|j,j\rangle
	:=D^{(j)}(e^{-i\phi(\sigma) J_z}e^{-i\theta(\sigma) J_y})|j,j\rangle\,,
\end{align}
where $(\phi(\sigma),\theta(\sigma))\in(-\pi,\pi)\times(0,\pi)$ are the spherical coordinates associated with $\sigma$.
By a standard argument, \textit{cf.} \cite{Perelomov_1972}, the family of coherent states $\{|j,\sigma\rangle\}_{\sigma\in\mathbb{S}^2}$ form an over-complete set in $\mathbb{C}^{2j+1}$ in the sense that
\begin{align}\label{Eq: coherent states - overcompleteness}
	\int_{\mathbb{S}^2}
	|j,\sigma\rangle\langle j,\sigma|\,
	\mathrm{d}\mu_j(\sigma)=I\,,
\end{align}
where the integral in the left-hand side is computed in the weak sense while $|j,\sigma\rangle\langle j,\sigma|$ denotes the orthogonal projector along $|j,\sigma\rangle$ ---here $\mu_j$ denotes the standard measure on $\mathbb{S}^2$ normalized so that $\mu_j(\mathbb{S}^2)=2j+1$.

At this stage the \textbf{Berezin quantization map} $Q_j\colon B_\infty\to B_j$ is defined by the weak integral
\begin{align}\label{Eq: Berezin quantization map - single site}
	Q_j(a):=\int_{\mathbb{S}^2}a(\sigma)|j,\sigma\rangle\langle j,\sigma|\,\mathrm{d}\mu_j(\sigma)\,.
\end{align}
It is worth observing that $Q_j(a)\geq 0$ whenever $a\geq 0$, moreover, $\|Q_j(a)\|_{B_j}\leq\|a\|_{B_\infty}$.
Furthermore, setting
\begin{align}\label{Eq: Berezin SDQ - check function}
	\check{a}_j(\sigma)
	:=\langle j,\sigma|Q_j(a)|j,\sigma\rangle
	=\int_{\mathbb{S}^2}a(\sigma')|\langle j,\sigma|j,\sigma'\rangle|^2\mathrm{d}\mu_j(\sigma')\,,
	\qquad a\in C(\mathbb{S}^2)\,,
\end{align}
one finds $\check{a}_j\in C(\mathbb{S}^2)$ and $\check{a}_j\to a$ in the sup-norm.
Within this setting the data
\begin{align*}
	B_j:=\begin{dcases}
		M_{2j+1}(\mathbb{C})
		&j\in\mathbb{Z}_+/2
		\\
		C(\mathbb{S}^2)
		&j=\infty
	\end{dcases}\,,
	\qquad
	Q_j\colon \dot{B}_\infty\to B_j
	\qquad
	Q_j(a_j):=\begin{dcases}
		Q_j(a_j)
		&j\in\mathbb{Z}_+/2
		\\
		a_\infty
		&j=\infty
	\end{dcases}\,,
\end{align*}
identify a SDQ over a suitably defined bundle of $C^*$-algebras $\boxed{B_\ast}\subseteq\prod_{j\in\overline{\mathbb{Z}_+/2}}B_j$, \textit{cf.} \cite[Thm. 8.1]{Landsman_2017}.
From a physical point of view, the semiclassical parameter is identified with $\boxed{h_j}:=(2j+1)^{-1}$.

\begin{example}\label{Ex: Berezin SDQ - quantization of spherical harmonics}
	For later purposes, we report in this example the fairly explicit computation of the Berezin quantization of a generic spherical harmonic, \textit{cf.} \cite{Perelomov_1972}.
	In more details, let $\boxed{Y_{\ell,m}}\in C^\infty(\mathbb{S}^2)$ be the spherical harmonic with parameter $\ell\in\mathbb{Z}_+$, $m\in[-\ell,\ell]\cap\mathbb{Z}$: Explicitly we set
	\begin{align}\label{Eq: spherical harmonics - convention}
		Y_{\ell,m}(\sigma)
		:=\sqrt{\frac{(\ell-m)!}{(\ell+m)!}}P_{\ell,m}[\cos\theta(\sigma)]e^{im\phi(\sigma)}\,,
	\end{align}
	where $P_{\ell,m}$ denotes the Legendre polynomial of order $\ell,m$.
	Notice that this choice of normalization is such that $\|Y_{\ell,m}\|_{L^2(\mathbb{S}^2,\mu_\ell)}=1$, moreover, $\|Y_{\ell,m}\|_{B_\infty}\leq 1$, \textit{cf.} \cite[Cor. 2.9]{Stein_Weiss_1971}.
	The set $\{Y_{\ell,m}\}_{\ell,m}$ is a complete orthogonal system for $L^2(\mathbb{S}^2,\mu_0)$ made by orthogonal but not $L^2$-normalized vectors.
	With this convention we also have \cite[\S 3.6.2]{Napolitano_Sakuraki_2021}
	\begin{align}\label{Eq: coherent states - relation with spherical harmonics}
		\overline{Y_{\ell,m}(\sigma)}
		=\langle \ell,m|D^{(j)}(\sigma)|\ell,0\rangle
		=D^{(\ell)}_{m,0}(\sigma)\,,
	\end{align}
	where  $\boxed{D^{(j)}_{m,k}(\sigma)}:=\langle j,m|D^{(j)}(\sigma)|j,k\rangle$ denotes the \textbf{Wigner D-matrix}.
	By direct inspection we find
	\begin{align*}
		\langle j,m_1|Q_j(Y_{\ell,m})|j,m_2\rangle
		&=\int_{\mathbb{S}^2}Y_{\ell,m}(\sigma)\langle j,m_1|j,\sigma\rangle\langle j,\sigma|j,m_2\rangle\,\mathrm{d}\mu_j(\sigma)
		\\
		&=\int_{\mathbb{S}^2}\overline{D^{(\ell)}_{m,0}(\sigma)}
		D^{(j)}_{m_1,j}(\sigma)
		\overline{D^{(j)}_{m_2,j}(\sigma)}
		\mathrm{d}\mu_j(\sigma)
		\\
		&=\textsc{cg}_{\ell,m;j,m_2}^{j,m_1}
		\textsc{cg}_{\ell,0;j,j}^{j,j}\,,
	\end{align*}
	where we used Schur orthogonality relations, \textit{cf.} \cite[\S 4.10]{Khersonskii_Moskalev_Varshalovich_1988}, while $\boxed{\textsc{cg}_{j_1,m_1;j_2,m_2}^{j,m}}$ denotes the Clebsch-Gordan coefficients, \textit{cf.} \cite[\S 3]{Napolitano_Sakuraki_2021}.
	We recall in particular that $\textsc{cg}_{j_1,m_1;j_2,m_2}^{j,m}\in\mathbb{R}$, moreover, the coefficient vanishes unless $m=m_1+m_2$ and $|j_1-j_2|\leq j\leq j_1+j_2$.
	For later convenience it is also worth recalling that, for all $\sigma\in\mathbb{S}^2$,
	\begin{align*}
		D^{(j_1)}_{m_1,k_1}(\sigma)
		D^{(j_2)}_{m_2,k_2}(\sigma)
		=\sum_{j=|j_1-j_2|}^{j_1+j_2}
		\textsc{cg}_{j_1,m_1;j_2,m_2}^{j,m_1+m_2}
		\textsc{cg}_{j_1,k_1;j_2,k_2}^{j,k_1+k_2}
		D^{(j)}_{m_1+m_2,k_1+k_2}(\sigma)\,,
	\end{align*}
	which for the particular case of $k_1=k_2=0$ and $j\in\mathbb{Z}_+$ reduces to
	\begin{align}\label{Eq: spherical harmonics - expansion of pointwise product}
		Y_{j_1,m_1}Y_{j_2,m_2}
		=\sum_{j=|j_1-j_2|}^{|j_1+j_2|}
		\textsc{cg}_{j_1,0;j_2,0}^{j,0}
		\textsc{cg}_{j_1,m_1;j_2,m_2}^{j,m}
		Y_{j,m_1+m_2}\,.
	\end{align}
	Overall we find
	\begin{align}\label{Eq: Berezin SDQ - quantization of spherical harmonics}
		Q_j(Y_{\ell,m})
		=\textsc{cg}_{\ell,0;j,j}^{j,j}
		\sum_{m'=-j}^j
		\textsc{cg}_{\ell,m;j,m'}^{j,m+m'}
		|j,m+m'\rangle\langle j,m'|\,.
	\end{align}
	In particular $Q_j(Y_{\ell,m})=0$ for $\ell>2j$.
\end{example}

\begin{remark}\label{Rmk: Berezin quantization useful remark}
	In Section \ref{Sec: Common absence of CPTs and of QPTs}, \textit{cf.} Theorem \ref{Thm: uniqueness result for quantum KMS states}, we will profit of the following well-known properties of the Berezin quantization map $Q_j\colon C(\mathbb{S}^2)\to M_{2j+1}(\mathbb{C})$.
	The latter have interesting and crucial consequences that we could not find in the existing literature and which we describe in the present remark.
	To begin with, let
	\begin{align}\label{Eq: Hilbert-Schmidt scalar product}
		\langle A|B\rangle_{\textsc{hs}}
		:=\frac{1}{2j+1}\operatorname{tr}(A^*B)\,,
	\end{align}
	be the normalized Hilbert-Schmidt scalar product between $A,B\in M_{2j+1}(\mathbb{C})$.
	The irreducible representation $D^{(j)}\colon SU(2)\to M_{2j+1}(\mathbb{C})$, induces a new unitary representation $\boxed{\tilde{D}^{(j)}}\colon SU(2)\to\mathcal{B}(M_{2j+1}(\mathbb{C}))$, where $M_{2j+1}(\mathbb{C})$ is regarded as an Hilbert space with scalar product $\langle\,|\,\rangle_{\textsc{hs}}$, defined by
	\begin{align}\label{Eq: Hilbert-Schmidt induced representation}
		\tilde{D}^{(j)}(R)(A)
		:=\operatorname{Ad}_{D^{(j)}(R)}(A)
		=D^{(j)}(R)AD^{(j)}(R)^*\,.
	\end{align}
	By direct inspection $\tilde{D}^{(j)}(R)$ is unitary with respect to the Hilbert-Schmidt scalar product \eqref{Eq: Hilbert-Schmidt scalar product} on $M_{2j+1}(\mathbb{C})$, moreover $R\mapsto\tilde{D}^{(j)}(R)$ is a unitary representation of $SU(2)$.
	Though $D^{(j)}$ is irreducible, $\tilde{D}^{(j)}$ is not irreducible and by Peter-Weyl Theorem, \textit{cf.} \cite[\S 5.2]{Folland_2015}, it decomposes into irreducible representations $\{D^{(\ell)}\}_{\ell\in\mathbb{Z}_+/2}$.
	Notably the Berezin quantization map $Q_j\colon C(\mathbb{S}^2)\to M_{2j+1}(\mathbb{C})$ of Equation \eqref{Eq: Berezin quantization map - single site} can be used to provide the explicit decomposition of $\tilde{D}^{(j)}$ in its irreducible components.
	To this avail, let
	\begin{align}\label{Eq: SU2 left-action representation}
		SU(2)\ni R\mapsto\hat{R}\in\mathcal{B}(L^2(\mathbb{S}^2,\mu_0))
		\qquad
		(\hat{R}a)(\sigma):=a(R^{-1}\sigma)\,,
	\end{align}
	be the usual left-action unitary representation of $SU(2)$ ---with a slight abuse of notation we dropped the isomorphism $SU(2)/\{\pm I\}\simeq SO(3)$.
	When necessary will denote by $\hat{J}_k$ the infinitesimal generator of $\hat{R}$ for $R=e^{-iJ_k}$.
	It is well known that the left-action representation decomposes into the irreducible representations of $SU(2)$ with integer spin by considering the $L^2$-decomposition of $L^2(\mathbb{S}^2,\mu_0)$ in spherical harmonics.
	Actually, restricting the action of $\hat{R}$ to the vector space spanned by $\{Y_{\ell,m}\}_{m\in[-\ell,\ell]\cap\mathbb{Z}}$ leads to a representation which is unitary equivalent to $D^{(\ell)}$.
	
	At this stage we may observe that, by direct inspection,
	\begin{multline}\label{Eq: Berezin SDQ - left-action representation intertwining property}
		Q_j(\hat{R}a)
		=\int_{\mathbb{S}^2}a(R^{-1}\sigma)|j,\sigma\rangle\langle j,\sigma|\,\mathrm{d}\mu_j(\sigma)
		=\int_{\mathbb{S}^2}a(\sigma)|j,R\sigma\rangle\langle j,R\sigma|\,\mathrm{d}\mu_j(\sigma)
		\\
		=D^{(j)}(R)\int_{\mathbb{S}^2}a(\sigma)|j,\sigma\rangle\langle j,\sigma|\,\mathrm{d}\mu_j(\sigma)D^{(j)}(R)^*
		=\tilde{D}^{(j)}(R)Q_j(a)\,,
	\end{multline}
	where $a\in C(\mathbb{S}^2)$ while we used the rotation invariance of $\mu_j$ and the fact that $D^{(j)}(R)|j,\sigma\rangle=e^{i\alpha(\sigma,R)}|j,R\sigma\rangle$ for $\alpha(\sigma,R)\in\mathbb{R}$, \textit{cf.} \cite{Perelomov_1972}.
	Thus, $Q_j$ intertwines between the left-action representation and $\tilde{D}^{(j)}$.
	Moreover, it is well-known that
	\begin{multline}\label{Eq: Berezin SDQ - scalar product intertwining property}
		\langle Q_j(a)|Q_j(a')\rangle_{\textsc{hs}}
		=\int_{\mathbb{S}^2}\int_{\mathbb{S}^2}
		\overline{a(\sigma)}a'(\sigma')|\langle j,\sigma|j,\sigma'\rangle|^2
		\mathrm{d}\mu_j(\sigma)\mathrm{d}\mu_0(\sigma')
		\\
		=\langle a|\check{a}_j'\rangle_{L^2(\mathbb{S}^2,\mu_0)}
		=\langle \check{a}_j|a'\rangle_{L^2(\mathbb{S}^2,\mu_0)}\,,
	\end{multline}
	where we used both Equations \eqref{Eq: Berezin quantization map - single site}-\eqref{Eq: Berezin SDQ - check function}.
	
	Equations \eqref{Eq: Berezin SDQ - left-action representation intertwining property}-\eqref{Eq: Berezin SDQ - scalar product intertwining property} have crucial consequences in the decomposition of $\tilde{D}^{(j)}$.
	In particular, Equation \eqref{Eq: Berezin SDQ - left-action representation intertwining property} implies that $\hat{R}\check{a}_j=\widecheck{[a\circ R^{-1}]}_j$, \textit{i.e.} the left-action representation and the check-operator commute.
	At an infinitesimal level this implies
	\begin{align}\label{Eq: Berezin SDQ - check commutates with Laplacian}
		\Delta_{\mathbb{S}^2}\check{a}_j=\widecheck{[\Delta_{\mathbb{S}^2}a]}_j\,,
	\end{align}
	where we observed that $\Delta_{\mathbb{S}^2}=\sum_{k=1}^3\hat{J}_k^2$.
	Together with Equation \eqref{Eq: Berezin SDQ - scalar product intertwining property}, Equation \eqref{Eq: Berezin SDQ - check commutates with Laplacian} implies that
	\begin{align}\label{Eq: Berezin SDQ - check of spherical harmonics}
		\widecheck{[Y_{\ell,m}]}_j=c_{j,\ell}Y_{\ell,m}\,,
	\end{align}
	where $c_{j,\ell}$ is explicitly computed using Equation \eqref{Eq: Berezin SDQ - quantization of spherical harmonics}, \textit{cf.} Example \ref{Ex: Berezin SDQ - quantization of spherical harmonics}:
	\begin{align*}
		c_{j,\ell}=\langle Y_{\ell,m}|\widecheck{[Y_{\ell,m}]}_j\rangle_{L^2(\mathbb{S}^2,\mu_\ell)}
		&=\frac{2\ell+1}{2j+1}\langle Q_j(Y_{\ell,m})|Q_j(Y_{\ell,m})\rangle_{\textsc{hs}}
		\\
		&=(\textsc{cg}_{\ell,0;j,j}^{j,j})^2
		\frac{2\ell+1}{2j+1}\sum_{m'=-j}^j
		(\textsc{cg}_{\ell,m;j,m'}^{j,m+m'})^2
		=(\textsc{cg}_{\ell,0;j,j}^{j,j})^2\,,
	\end{align*}
	where in the last line we used the symmetry property of the Clebsch-Gordan coefficients, \textit{cf.} \cite[\S 8]{Khersonskii_Moskalev_Varshalovich_1988}.
	Equations \eqref{Eq: Berezin SDQ - scalar product intertwining property}-\eqref{Eq: Berezin SDQ - check of spherical harmonics} imply that \begin{align*}
		\{Q_j(Y_{\ell,m})\,|\,
		\ell\in\{0,\ldots,2j\}\,,\,
		m\in[-\ell,\ell]\cap\mathbb{Z}
		\}\,,
	\end{align*}
	form a complete orthogonal system in $M_{2j+1}(\mathbb{C})$ with respect to the Hilbert-Schmidt scalar product \eqref{Eq: Hilbert-Schmidt scalar product}.
	(As an aside, we observe that this fact provides a quick proof of that $Q_j\colon\dot{B}_\infty\to B_j$ is surjective: Indeed, for any $A_j\in B_j$ we may consider $a_j:=\sum_{\ell=0}^{2j}\sum_{m=-\ell}^\ell A_{\ell,m}Y_{\ell,m}\in\dot{B}_\infty$ where $A_{\ell,m}:=\|Q_j(Y_{\ell,m})\|_{\textsc{hs}}^{-2}\langle Q_j(Y_{\ell,m})|A_j\rangle_{\textsc{hs}}$ so that $A_j=Q_j(a_j)$.)
	Moreover, Equation \eqref{Eq: Berezin SDQ - left-action representation intertwining property} ensures that $Q_j$ intertwines between the left-action representation $R\mapsto\hat{R}$ and $\tilde{D}^{(j)}$.
	Since $R\mapsto\hat{R}$ is unitary equivalent to $D^{(\ell)}$ when restricted to the vector space spanned by $\{Y_{\ell,m}\}_{m\in[-\ell,\ell]\cap\mathbb{Z}}$, it follows that $\tilde{D}^{(j)}$ is unitary equivalent to $D^{(\ell)}$ when restricted to the vector space spanned by $\{Q_j(Y_{\ell,m})\}_{m\in[-\ell,\ell]\cap\mathbb{Z}}$.
	Thus, $\tilde{D}^{(j)}$ is not irreducible and decomposes into direct sum of the irreducible representations of $SU(2)$ with total spin $\ell\in\{0,\ldots,2j\}$, each of which taken with multiplicity one.
\end{remark}

\subsection{Berezin SDQ on $\Gamma$}
\label{Subsec: Berezin SDQ on Gamma}

The goal of this section is to prove that the Berezin SDQ $Q_j\colon\dot{B}_\infty\to B_j$, \textit{cf.} Equation \eqref{Eq: Berezin quantization map - single site}, lifts to a SDQ $Q_j^\Gamma\colon\dot{B}_\infty^\Gamma\to\dot{B}_j^\Gamma$ between the corresponding algebras of quasi-local observables for the corresponding infinitely extended systems.

To this avail we recall that in \cite{Murro_vandeVen_2022} the Berezin SDQ $Q_j\colon\dot{B}_\infty\to B_j$ has been lifted to a SDQ $\boxed{Q_j^\Lambda}\colon\dot{B}_\infty^\Lambda\to B_j^\Lambda$ for any finite region $\Lambda\Subset\Gamma$.
In a nutshell, this boils down to define $Q_j^\Lambda$ by linear extension of the tensor product map $\bigotimes_{x\in\Lambda}Q_j$ and checking that the data $B_\infty^\Gamma,B_\ast^\Lambda,\{Q_j\}_{j\in\mathbb{Z}_+/2}$ fulfil the requirement of a SDQ ---here $\boxed{B_\ast^\Lambda}:=\bigotimes_{x\in\Lambda} B_\ast$.
A non-trivial task in this setting is to prove that $B_\ast^\Lambda$ is again a bundle of $C^*$-algebras over $\overline{\mathbb{Z}_+/2}$: This is addressed in full generality in \cite{Murro_vandeVen_2022} by using the results of \cite{Kirchberg_Wassermann_1995}.

To extend the results of \cite{Berezin_1975,Murro_vandeVen_2022} to the case of an infinitely extended system over $\Gamma$ we need to identify a suitable continuous bundle of $C^*$-algebras $B_\ast^\Gamma\subset\prod_{j\in\overline{\mathbb{Z}_+/2}}B_j^\Gamma$, where $B_j^\Gamma$ are the quasi-local algebras introduced in Section \ref{Subsec: classical and quantum lattice systems on Gamma}.
Eventually we will define suitable quantization maps $Q_j^\Gamma\colon \dot{B}_\infty^\Gamma\to B_j^\Gamma$ abiding by the requirements \ref{Item: quantization maps are Hermitian and define a continuous section}-\ref{Item: quantization maps are strict}-\ref{Item: quantization maps fulfils the DGR condition} of a SDQ.

\begin{remark}\label{Rmk: sufficient condition for continuous bundle of Cstar algebras}
	Let $\{\mathfrak{A}_j\}_{j\in\overline{\mathbb{Z}_+/2}}$ be a collection of $C^*$-algebras.
	For later convenience we recall the following sufficient condition which identifies a continuous bundle of $C^*$-algebras $\mathfrak{A}\subset\prod_{j\in\overline{\mathbb{Z}_+/2}}\mathfrak{A}_j$ by defining a dense set of (a posteriori) elements of $\mathfrak{A}$ ---\textit{cf.} \cite[Prop. 1.2.3]{Landsman_1998}, \cite[Prop. C.124]{Landsman_2017}.
	In more details, let $\dot{\mathfrak{A}}\subseteq\prod_{j\in\overline{\mathbb{Z}_+/2}}\mathfrak{A}_j$ be such that:
	\begin{enumerate}
		\item\label{Item: tildeA is pointwise dense}
		For all $j\in\overline{\mathbb{Z}_+/2}$ the set $\{\mathfrak{a}_j\,|\,(\mathfrak{a}_j)_{j\in\overline{\mathbb{Z}_+/2}}\in\dot{\mathfrak{A}}\}$ is dense in $\mathfrak{A}_j$;
		
		\item\label{Item: tildeA is a star-algebra}
		$\dot{\mathfrak{A}}$ is a $*$-algebra;
		
		\item\label{Item: tildeA fulfils the Rieffel condition}
		For all $(\mathfrak{a}_j)_{j\in\overline{\mathbb{Z}_+/2}}\in\dot{\mathfrak{A}}$, we have $(\|\mathfrak{a}_j\|_{\mathfrak{A}_j})_{j\in\overline{\mathbb{Z}_+/2}}\in C(\overline{\mathbb{Z}_+/2})$.
	\end{enumerate}
	Then
	\begin{align}\label{Eq: bundle defined from tildeA}
		\mathfrak{A}:=\bigg\lbrace \mathfrak{a}\in\prod_{j\in\overline{\mathbb{Z}_+/2}}\mathfrak{A}_j\,|\,
		\forall\varepsilon>0\;\exists j_{\varepsilon}\in\mathbb{Z}_+/2\,,\,
		\exists \mathfrak{a}'\in\dot{\mathfrak{A}}\colon
		\|\mathfrak{a}_j-\mathfrak{a}'_j\|_{\mathfrak{A}_j}<\varepsilon\;\forall j\geq j_\varepsilon
		\bigg\rbrace\,,
	\end{align}
	is the smallest continuous bundle of $C^*$-algebras over $\overline{\mathbb{Z}_+/2}$ which contains $\dot{\mathfrak{A}}$.
\end{remark}

The following proposition identifies a continuous bundle of $C^*$-algebras $B_\ast^\Gamma$ with fibers $\{B_j^\Gamma\}_{j\in\overline{\mathbb{Z}_+/2}}$.

\begin{proposition}\label{Prop: oo-site continuous bundle of Cstar algebras}
	Let $\dot{B}_\ast^\Gamma\subset\prod_{j\in\overline{\mathbb{Z}_+/2}}B_j^\Gamma$ be defined by
	\begin{multline}\label{Eq: oo-site Cstar bundle dense subset}
		\dot{B}_\ast^\Gamma
		:=\operatorname{Alg}(\dot{V})\,,
		\\
		\dot{V}:=\Big\lbrace
		(\mathfrak{a}_j)_{j\in\overline{\mathbb{Z}_+/2}}\in\prod_{j\in\overline{\mathbb{Z}_+/2}}B_j^\Gamma\,|\,
		\exists \Lambda\Subset\Gamma\,,\,
		a_\Lambda\in \dot{B}_\infty^\Lambda\colon
		\mathfrak{a}_j=
		\begin{cases}
			Q_j^\Lambda(a_\Lambda)
			& j\in\mathbb{Z}_+/2
			\\
			a_\Lambda
			&j=\infty
		\end{cases}
		\Big\rbrace\,,
	\end{multline}
	where $\operatorname{Alg}(\dot{V})$ denotes the algebra generated by the vector space $\dot{V}$.
	
	Then $\dot{\mathfrak{A}}$ fulfils conditions \ref{Item: tildeA is pointwise dense}-\ref{Item: tildeA is a star-algebra}-\ref{Item: tildeA fulfils the Rieffel condition} of Remark \ref{Rmk: sufficient condition for continuous bundle of Cstar algebras}, thus, it identifies a continuous bundle of $C^*$-algebras $B_\ast^\Gamma$ over $\overline{\mathbb{Z}_+/2}$ defined by
	\begin{align}\label{Eq: oo-site continuous bundle of Cstar algebras}
		B_\ast^\Gamma
		:=\{
		(\mathfrak{a}_j)_{j\in\overline{\mathbb{Z}_+/2}}\,|\,
		\forall\varepsilon>0\,,\,
		\exists j_\varepsilon\in \mathbb{Z}_+/2\,,\,
		\exists \mathfrak{a}'\in\dot{B}_\ast^\Gamma\colon
		\|\mathfrak{a}_j-\mathfrak{a}_j'\|_{B_j^\Gamma}<\varepsilon
		\;\forall j\geq j_\varepsilon
		\}\,.
	\end{align}
\end{proposition}

\begin{proof}
	It suffices to prove that $\dot{B}_\ast^\Gamma$ fulfils conditions \ref{Item: tildeA is pointwise dense}-\ref{Item: tildeA is a star-algebra}-\ref{Item: tildeA fulfils the Rieffel condition}.
	\begin{description}
		\item[\ref{Item: tildeA is pointwise dense}]
		Let $j\in\mathbb{Z}_+/2$ and consider $\{\mathfrak{a}_j\,|\,(\mathfrak{a}_j)_{j\in\overline{\mathbb{Z}_+/2}}\in\dot{B}_\ast^\Gamma\}\subset B_j^\Gamma$ ---a similar argument applies for $j=\infty$.
		Since $Q_j\colon\dot{B}_\infty\to B_j$ is surjective, \textit{cf.} Remark \ref{Rmk: Berezin quantization useful remark}, the same holds for $Q_j^\Lambda\colon\dot{B}_\infty^\Lambda\to B_j^\Lambda$ for all $\Lambda\Subset\Gamma$.
		This implies that any $A_\Lambda\in B_j^\Lambda\subset \dot{B}_j^\Gamma$ can be written as $Q_j^\Lambda(a_{j,\Lambda})$ for some $a_{j,\Lambda}\in \dot{B}_\infty^\Lambda$.
		Thus, $A_j\in\{\mathfrak{a}_j\,|\,(\mathfrak{a}_{j'})_{j'\in\overline{\mathbb{Z}_+/2}}\in \dot{B}_\ast^\Gamma\}$ because $A_j=\mathfrak{a}_j$ for $(\mathfrak{a}_{j'})_{j'\in\overline{\mathbb{Z}_+/2}}$ defined by $\mathfrak{a}_{j'}:=Q_{j'}^\Lambda(a_{j,\Lambda})$ for all $j'\in\overline{\mathbb{Z}_+/2}$.
		Condition \ref{Item: tildeA is pointwise dense} follows from the density of $\dot{B}_j^\Gamma$ in $B_j^\Gamma$.
		
		\item[\ref{Item: tildeA is a star-algebra}]
		Condition \ref{Item: tildeA is a star-algebra} holds because $\dot{B}_\ast^\Gamma=\operatorname{Alg}(\dot{V})$, moreover, $\dot{V}$ is closed under $*$-conjugation.
		Notice that $\dot{V}$ is not an algebra because $Q_j^\Lambda(a_\Lambda\tilde{a}_\Lambda)\neq Q_j^\Lambda(a_\Lambda)Q_j^\Lambda(\tilde{a}_\Lambda)$, although this is true in the limit $j\to\infty$.
		
		\item[\ref{Item: tildeA fulfils the Rieffel condition}]
		For any $(\mathfrak{a}_j)_{j\in\overline{\mathbb{Z}_+/2}}\in\dot{V}$ we have
		\begin{align*}
			\|\mathfrak{a}_j\|_{B_j^\Gamma}
			=\|Q_j^\Lambda(a_\Lambda)\|_{B_j^\Lambda}
			\underset{j\to\infty}{\longrightarrow}
			\|a_\Lambda\|_{B_\infty^\Lambda}
			=\|\mathfrak{a}_\infty\|_{B_\infty^\Gamma}\,,
		\end{align*}
		where we used that $(Q_j^\Lambda(a_\Lambda))_{j\in\overline{\mathbb{Z}_+/2}}\in B_\ast^\Lambda$, therefore, $\overline{\mathbb{Z}_+/2}\ni j\mapsto\|Q_j^\Lambda(a_\Lambda)\|_{B_j^\Lambda}$ is continuous.
        Condition \ref{Item: tildeA fulfils the Rieffel condition} follows from the latter observation together with the fact that
        $\|Q_j^\Lambda(a_\Lambda)Q_j^\Lambda(a_\Lambda')-Q_j^\Lambda(a_\Lambda a_\Lambda')\|_{B_j^\Lambda}\underset{j \to\infty}{\longrightarrow}0$
        and $\dot{B}_\ast^\Gamma=\operatorname{Alg}(\dot{V})$, 
	\end{description}
\end{proof}

We now move to the definition of a SDQ associated with $B_\infty^\Gamma$ and $B_\ast^\Gamma$.
For later convenience we observe that, by direct inspection,
\begin{align}\label{Eq: Q, varphi interplay}
	Q_j^\Lambda\circ\iota^{\Lambda_0}_\Lambda=\iota^{\Lambda_0}_\Lambda\circ Q_j^{\Lambda_0}
	\qquad\forall \Lambda_0\subset\Lambda\Subset\Gamma\,.
\end{align}

\begin{theorem}\label{Thm: Berezin SDQ on Gamma}
	For $j\in\overline{\mathbb{Z}_+/2}$ let $Q_j^\Gamma\colon\dot{B}_\infty^\Gamma\to\dot{B}_j^\Gamma$ be the map defined by
	\begin{align}\label{Eq: SDQ for oo spin system - quantization map}
		Q_j^\Gamma(a_\Lambda):=
		\begin{dcases}
			Q_j^\Lambda(a_\Lambda)
			&j\in\mathbb{Z}_+/2
			\\
			a_\Lambda
			&j=\infty
		\end{dcases}
	\end{align}
	where $\Lambda\Subset\Gamma$, $a_\Lambda\in\dot{B}_\infty^\Lambda\subset\dot{B}_\infty^\Gamma$.
	
	Then the data $B_\infty^\Gamma$, $B_\ast^\Gamma$ and $\{Q_j^\Gamma\}_{j\in\overline{\mathbb{Z}_+/2}}$ define a SDQ.
\end{theorem}

\begin{proof}
	We will prove properties \ref{Item: quantization maps are Hermitian and define a continuous section}-\ref{Item: quantization maps fulfils the DGR condition}-\ref{Item: quantization maps are strict}.
	\begin{description}
		\item[\ref{Item: quantization maps are Hermitian and define a continuous section}]
		By direct inspection $Q_j^\Gamma(a_\Lambda)^*=Q_j^\Gamma(a_\Lambda^*)$ for all $a_\Lambda\in B_\infty^\Lambda\subset\dot{B}_\infty^\Gamma$, moreover, $(Q_j^\Gamma(a_\Lambda))_{j\in\overline{\mathbb{Z}_+/2}}\in \dot{B}_\ast^\Gamma$ ---\textit{cf.} Equation \eqref{Eq: oo-site Cstar bundle dense subset}--- thus, it defines a continuous section of $B_\ast^\Gamma$.
		
		\item[\ref{Item: quantization maps fulfils the DGR condition}]
		For all $a_\Lambda,\tilde{a}_\Lambda\in\dot{B}_\infty^\Lambda$ we have
		\begin{align*}
			Q_j^\Gamma\big(
			\{ a_\Lambda,\tilde{a}_\Lambda\}_{B_\infty^\Gamma}\big)
			=Q_j^\Gamma\big(
			\{a_\Lambda,\tilde{a}_\Lambda\}_{B_\infty^\Lambda}
			\big)
			= Q_j^\Lambda\big(\{a_\Lambda,\tilde{a}_\Lambda\}_{B_\infty^\Lambda}\big)\,,
		\end{align*}
		therefore,
		\begin{multline*}
			\big\|Q_j^\Gamma\big(\{ a_\Lambda, \tilde{a}_\Lambda\}_{B_\infty^\Gamma}
			\big)
			-ih_j^{-1}\big[Q_j^\Gamma( a_\Lambda),Q_j^\Gamma( \tilde{a}_\Lambda)\big]\big\|_{B_j^\Gamma}
			\\
			=\big\|Q_j^\Lambda\big(\{a_\Lambda,\tilde{a}_\Lambda\}_{B_\infty^\Lambda}\big)
			-ih_j^{-1}[Q_j^\Lambda(a_\Lambda),
			Q_j^\Lambda\tilde{a}_\Lambda]\big\|_{B_j^\Lambda}
			\underset{j\to\infty}{\longrightarrow}0\,,
		\end{multline*}
		where in the last line we used property \ref{Item: quantization maps fulfils the DGR condition} for the quantization map $Q_j^\Lambda\colon\dot{B}_\infty^\Lambda\to B_j^\Lambda$.
		
		\item[\ref{Item: quantization maps are strict}]
		For $j=\infty$ property \ref{Item: quantization maps are strict} follows from the density of $\dot{B}_\infty^\Gamma$ in $B_\infty^\Gamma$.
		For $j\in\mathbb{Z}_+/2$ it suffices to observe that
		\begin{align*}
			Q_j^\Gamma(\dot{B}_\infty^\Gamma)
			=\bigcup_{\Lambda\Subset\Gamma}Q_j^\Gamma(\dot{B}_\infty^\Lambda)
			=\bigcup_{\Lambda\Subset\Gamma}Q_j^\Lambda(\dot{B}_\infty^\Lambda)
			=\bigcup_{\Lambda\Subset\Gamma}B_j^\Lambda
			=\dot{B}_j^\Gamma\,,
		\end{align*}
		where we used that $Q_j^\Lambda\colon\dot{B}_\infty^\Lambda\to B_j^\Lambda$ is surjective together with the definitions of $\dot{B}_\infty^\Gamma$ and $\dot{B}_j^\Gamma$.
		Since $\dot{B}_j^\Gamma$ is dense in $B_j^\Gamma$, the claim follows.
	\end{description}
\end{proof}

\section{The semiclassical limit of the quantum KMS condition}
\label{Sec: The semiclassical limit of the quantum KMS condition}

In this section we will study the notion of classical and quantum thermal equilibrium for the algebras $B_j^\Gamma$, $j\in\overline{\mathbb{Z}_+/2}$, of classical and quantum observables for the spin lattice systems over $\Gamma$ introduced in Section \ref{Subsec: classical and quantum lattice systems on Gamma}.
Thermal equilibrium is described by states $\omega_j^{\beta,\Gamma}\in S(B_j^\Gamma)$ fulfilling the KMS condition, \textit{cf.} Equations \eqref{Eq: KMS condition - quantum on abstract C* algebra}-\eqref{Eq: KMS condition - classical on abstract C* algebra}: These conditions cover the notion of both classical and quantum thermal equilibrium.

Our main interest concerns the connection between classical and quantum thermal equilibrium.
Specifically, while the physical justification of the quantum KMS condition has been the subject of many investigations, \textit{cf.} \cite{Haag_Hugenholtz_Winnik_1967,Pusz_Woronowicz_1978}, the classical KMS condition is usually justified with a formal semi-classical limit of the quantum KMS condition, \textit{cf.} \cite{Gallavotti_Verboven_1975}.
Our main goal is to prove that this derivation can be proved rigorously within the framework of the SDQ introduced in Section \ref{Sec: Berezin SDQ on a lattice system}.

In particular, using the quantization maps $Q_j^\Gamma\colon\dot{B}_\infty^\Gamma\to\dot{B}_j^\Gamma$ it is possible to analyse the semi-classical behaviour of sequences of quantum states $(\omega_j)_{j\in\mathbb{Z}_+/2}$, $\omega_j\in S(B_j^\Gamma)$, by studying the weak*-limit points of the sequence $(\omega_j\circ Q_j^\Gamma)_{j\in\mathbb{Z}_+/2}$ of classical states on $S(B_\infty^\Gamma)$ ---it is worth noticing that $Q_j^\Gamma$ preserves positivity because so does $Q_j$, \textit{cf.} Equation \eqref{Eq: Berezin quantization map - single site}.
At this stage the natural question is whether a sequence of quantum KMS states $(\omega_j)_{j\in\mathbb{Z}_+/2}$ leads to weak*-limit points $(\omega_j\circ Q_j^\Gamma)_{j\in\mathbb{Z}_+/2}$ which fulfil the classical KMS condition.
This is in fact what happens as we will prove in Theorem \ref{Thm: limit points of quantum KMS in oo volume are classical KMS in oo volume}.

In Section \ref{Subsec: classical and quantum KMS condition on BooGamma and BjGamma} we will recall the notion of classical and quantum KMS condition within the framework introduced in Section \ref{Subsec: classical and quantum lattice systems on Gamma}, \textit{cf.} \cite{Bratteli_Robinson_97,Drago_Van_de_Ven_2023}.
In passing, we will prove a characterization of the classical KMS condition, \textit{cf.} Lemma \ref{Lem: classical auto-correlation lower bound}, which is inspired by an analogous characterization of the quantum KMS condition known with the name of Roepstorff-Araki-Sewell auto-correlation lower bound, \textit{cf.} \cite[Thm. 5.3.15]{Bratteli_Robinson_97}.
Section \ref{Subsec: semi-classical of quantum KMS condition} is devoted to the proof of Theorem \ref{Thm: limit points of quantum KMS in oo volume are classical KMS in oo volume}.
The latter is based on Lemma \ref{Lem: classical auto-correlation lower bound} together with a result on the semiclassical limit of the quantum derivation associated with the quantum KMS condition, \textit{cf.} Lemma \ref{Lem: classical limit of local dynamics in oo-volume}.

\subsection{Classical and quantum KMS condition on $B_\infty^\Gamma$ and $B_j^\Gamma$}
\label{Subsec: classical and quantum KMS condition on BooGamma and BjGamma}

In this section we briefly recall the notion of classical and quantum KMS conditions, \textit{cf.} Equations \eqref{Eq: KMS condition - quantum on abstract C* algebra}-\eqref{Eq: KMS condition - classical on abstract C* algebra}, for the quasi-local algebras $B_j^\Gamma$, $B_\infty^\Gamma$.
Eventually we will move to the investigation of the semiclassical limit of quantum KMS states.

\bigskip

We will begin with the quantum KMS condition: Since the latter is very well-known, \textit{cf.} \cite{Bratteli_Robinson_97}, we will streamline its presentation by focusing on the main details which will be important for the forthcoming discussion.

The quantum KMS condition \eqref{Eq: KMS condition - quantum on abstract C* algebra} relies on the choice of a strongly one-parameter group on the $C^*$-algebra of interest: For the particular case of $B_j^\Lambda$ and $B_j^\Gamma$ suitable one-parameter groups of automorphisms $\tau^\Lambda$, $\tau^\Gamma$ are identified by considering a family $\boxed{\Phi_j}=\{\Phi_{j,\Lambda}\}_{\Lambda\Subset\Gamma}$ of self-adjoint elements $\Phi_{j,\Lambda}\in B_j^\Lambda\subset B_j^\Gamma$.
We will refer to $\Phi_{j,\Lambda}$ as the \textbf{potential} associated with $\Lambda\Subset\Gamma$: For fixed $\Lambda\Subset\Gamma$ the \textbf{quantum Hamiltonian} $\boxed{H_{j,\Lambda}}\in B_j^\Lambda$ associated to $\Phi_j$ is defined by $H_{j,\Lambda}:=\sum_{X\Subset\Lambda}\Phi_{j,X}$.
The latter induces a strongly one-parameter group $t\mapsto\boxed{\tau^\Lambda_t}$ on $B_j^\Lambda$ whose generator $\boxed{\delta_j^\Lambda}$ is given by $\delta_j^\Lambda:=i[H_{j,\Lambda},\;]$.
Within this setting it can be shown that there exists a unique $(\beta,\delta_j^\Lambda)$-KMS quantum state $\boxed{\omega_j^{\beta,\Lambda}}$ on $B_j^\Lambda$ called \textbf{quantum Gibbs state} and defined by
\begin{align}\label{Eq: KMS condition - quantum Gibbs state}
	\omega_j^{\beta,\Lambda}(A_\Lambda)
	:=\frac{\operatorname{Tr}\big(e^{-\beta H_{j,\Lambda}}A_\Lambda\big)}{\operatorname{Tr}(e^{-\beta H_{j,\Lambda}})}
	\qquad
	\forall A_\Lambda\in B_j^\Lambda\,.
\end{align}
Uniqueness of $(\beta,\delta_j^\Lambda)$-KMS quantum states is a manifestation of the relative simple nature of thermal equilibrium for systems of finite size: Needless to say, this does not hold any more for infinitely extended system.

Thermal equilibrium on $B_j^\Gamma$ is described by a suitable limit procedure $\Lambda\uparrow\Gamma$.
For that, further mild assumptions on $\Phi_j$ are required: In particular, we will assume that
\begin{align}\label{Eq: quantum potential - condition for existence of tau}
	\exists\lambda>0\colon
	\|\Phi\|_\lambda:=\sum_{m\geq 0}e^{\lambda m}
	\sup_{x\in\Gamma}\sum_{\substack{X\ni x\\|X|=m+1}}
	\|\Phi_{j,X}\|_{B_j^X}
	<\infty\,.
\end{align}
Within assumption \eqref{Eq: quantum potential - condition for existence of tau} it can be shown that $\delta_j^\Lambda$ converges strongly on $\dot{B}_j^\Gamma$ as $\Lambda\uparrow\Gamma$ to a $C^*$-derivation $\boxed{\delta_j^\Gamma}\colon\dot{B}_j^\Gamma\to B_j^\Gamma$ which generates a strongly continuous one-parameter group $t\mapsto\tau^\Gamma_t$ on $B_j^\Gamma$, \textit{cf.} \cite[Thm. 6.2.4]{Bratteli_Robinson_97}.
In particular, $\delta_j^\Gamma$ is explicitly given by
\begin{align*}
	\delta_j^\Gamma(A_\Lambda)
	=i\sum_{X\Subset\Gamma}[\Phi_X,A_\Lambda]
	=i\sum_{\substack{X\Subset\Gamma\\X\cap\Lambda\neq\varnothing}}[\Phi_X,A_\Lambda]
	\qquad
	\forall A_\Lambda\in B_j^\Lambda\subset\dot{B}_j^\Gamma\,.
\end{align*}
For later convenience it is worth to recall that for all $A_\Lambda\in B_j^\Gamma$, $\Lambda\Subset\Gamma$, we may compute
\begin{align}\label{Eq: analyticity of local observables}
	\tau^\Gamma_t(A)=\sum_{n\geq 0}\frac{t^n}{n!}(\delta_j^\Gamma)^n(A_\Lambda)\,,
\end{align}
where the series converges in $B_j^\Gamma$ for $|t|\leq\lambda/2\|\Phi_j\|_\lambda$, \textit{cf.} \cite[Thm. 6.2.4]{Bratteli_Robinson_97}.

Thus, assumption \eqref{Eq: quantum potential - condition for existence of tau} ensures the existence of a time evolution on $B_j^\Gamma$ and thermal equilibrium is then described by $(\beta,\delta_j^\Gamma)$-KMS quantum states on $B_j^\Gamma$.
Notably, the $(\beta,\delta_j^\Gamma)$-KMS condition does not select a unique state in general: Whenever uniqueness fails a \textbf{quantum phase transition} is said to occur.

\bigskip

Concerning the classical KMS condition \eqref{Eq: KMS condition - classical on abstract C* algebra} for the $C^*$-algebras $B_\infty^\Lambda$ and $B_\infty^\Gamma$, we will again consider a family $\varphi=\{\varphi_\Lambda\}_{\Lambda\Subset\Gamma}$ of self-adjoint potentials $\varphi_\Lambda\in\dot{B}_\infty^\Lambda$, \textit{cf.} \cite{Drago_vandeVen_2023}.
For the classical lattice system in a finite region $\Lambda\Subset\Gamma$ this suffices to identify a $*$-derivation $\boxed{\delta_\infty^\Lambda}\colon\dot{B}_\infty^\Lambda\to \dot{B}_\infty^\Lambda$ defined by $\delta_\infty^\Lambda:=\{\;,h_\Lambda\}$, where $\boxed{h_\Lambda}:=\sum_{X\subset\Lambda}\varphi_X$ is called \textbf{classical Hamiltonian}.
Similarly to the quantum case, the $(\beta,\delta_\infty^\Lambda)$-KMS condition select a unique state $\boxed{\omega_\infty^{\beta,\Lambda}}\in S(B_\infty^\Lambda)$ called \textbf{classical Gibbs state} and defined by
\begin{align}\label{Eq: KMS condition - classical Gibbs state}
	\omega_\infty^{\beta,\Lambda}(a_\Lambda)
	:=\frac{\int_{\mathbb{S}^2_\Lambda}a_\Lambda e^{-\beta h_\Lambda}\mathrm{d}\mu_0^\Lambda}
	{\int_{\mathbb{S}^2_\Lambda}e^{-\beta h_\Lambda}\mathrm{d}\mu_0^\Lambda}
	\qquad\forall a_\Lambda\in B_\infty^\Lambda\,.
\end{align}
The description of thermal equilibrium for $B_\infty^\Gamma$, \textit{i.e.} for a classical lattice system on the infinite region $\Gamma$, requires further assumptions on $\varphi$.
A sufficiently mild condition is provided by
\begin{align}\label{Eq: classical potential - condition for existence of delta on Gamma}
	\sup_{x\in\Gamma}\sum_{\substack{X\Subset\Gamma\\ X\ni x}}
	\|\varphi_X\|_{C^1(\mathbb{S}^2_X)}
	<\infty\,,
\end{align}
where $\|\;\|_{C^1(\mathbb{S}^2_\Lambda)}$ denote the $C^1$-norm on $C^1((\mathbb{S}^2)^{|\Lambda|})$.

For a family $\varphi=\{\varphi_\Lambda\}_{\Lambda\Subset\Gamma}$ of potentials fulfilling condition \eqref{Eq: classical potential - condition for existence of delta on Gamma} we may introduce a derivation $\boxed{\delta_\infty^\Gamma}\colon\dot{B}_\infty^\Gamma\to B_\infty^\Gamma$ defined by
\begin{align}\label{Eq: Poisson derivation on B0Gamma}
	\delta_\infty^\Gamma(a_\Lambda)
	:=\sum_{X\Subset\Gamma}
	\{a_\Lambda,\varphi_X\}_{B_\infty^\Gamma}
	=\sum_{\substack{X\Subset\Gamma\\X\cap\Lambda\neq\varnothing}}
	\{a_\Lambda,\varphi_X\}_{B_\infty^\Lambda}
	\qquad
	\forall a_\Lambda\in B_\infty^\Lambda\subset\dot{B}_\infty^\Lambda\,.
\end{align}
Notice that, on account of Equation \eqref{Eq: Poisson structures - reduction property}, the sum over $X\Subset\Gamma$ is restricted to $X\cap\Lambda\neq\varnothing$: This implies well-definiteness of $\delta_\infty^\Gamma$ because of condition \eqref{Eq: classical potential - condition for existence of delta on Gamma}.
Similarly to the quantum case, $\delta_\infty^\Gamma$ can be approximated by $\delta_\infty^\Lambda$, that is, $\delta_\infty^\Gamma (a_\Lambda)=\lim_{X\uparrow\Gamma}\delta_\infty^X(a_{\Lambda})$ for all $a_\Lambda\in B_\infty^\Lambda\subset\dot{B}_\infty^\Gamma$.
Thermal equilibrium on $B_\infty^\Gamma$ is described by considering $(\beta,\delta_\infty^\Gamma)$-KMS classical states.
Once again, the $(\beta,\delta_\infty^\Gamma)$-KMS classical condition does not identify a unique state on $B_\infty^\Gamma$ in general, leading to the notion of \textbf{classical phase transition}.

\begin{remark}\label{Rmk: KMS state on BooGamma seen as weak*-limit points of Gibbs states}
	\noindent
	\begin{enumerate}[(i)]
		\item
		We stress that conditions \eqref{Eq: quantum potential - condition for existence of tau}-\eqref{Eq: classical potential - condition for existence of delta on Gamma} ---see also \eqref{Eq: classical potential - condition for semiclassical limit}--- are minimal requirements for the discussion of this section.
		In the forthcoming Section \ref{Sec: Common absence of CPTs and of QPTs} we will to specialize further our assumptions, \textit{cf.} Theorems \ref{Thm: uniqueness result for classical KMS states}-\ref{Thm: uniqueness result for quantum KMS states}.
		In applications these conditions are usually met because the family of potentials $\Phi_j$, $\varphi$ turn out to be of \textbf{finite range}, namely there exists $m\in\mathbb{Z}_+$ and $d>0$ such that $\varphi_X=0$ (\textit{resp.} $\Phi_{j,X}=0$) if $|X|>m$ or $\operatorname{diam}(X)>d$.
		
		\item
		Remarkably, any weak*-limit point of the sequence $(\omega_j^{\beta,\Lambda})_{\Lambda\Subset\Gamma}$ of $(\beta,\delta_j^\Lambda)$-KMS quantum states leads to a $(\beta,\delta_j^\Gamma)$-KMS quantum state, \textit{cf.} \cite[Cor. 6.2.19]{Bratteli_Robinson_97}. The existence of these weak*-limit points follows from a standard Hahn-Banach and weak*-compactness argument.
		Similar considerations apply to the sequence $(\omega_\infty^{\beta,\Lambda})_{\Lambda\Subset\Gamma}$ of $(\beta,\delta_\infty^\Lambda)$-KMS classical states, ensuring the existence of $(\beta,\delta_\infty^\Gamma)$-KMS classical states.
	\end{enumerate}
\end{remark}

At this stage, we are in position to set our investigation of the semiclassical limit of quantum thermal states.
Specifically, we will consider a family $\varphi=\{\varphi_\Lambda\}_{\Lambda\Subset\Gamma}$ abiding by the condition
\begin{align}\label{Eq: classical potential - condition for semiclassical limit}
	\exists\lambda>0\colon
	\sum_{m\geq 0}e^{\lambda m}
	\sup_{x\in\Gamma}\sum_{\substack{X\ni x\\|X|=m+1}}
	\|\varphi_X\|_{C^1(\mathbb{S}^2_X)}
	<\infty\,.
\end{align}
The latter implies condition \eqref{Eq: classical potential - condition for existence of delta on Gamma}, furthermore, it ensures that the family $\Phi_j:=\{Q_j^\Gamma(\varphi_\Lambda)\}_{\Lambda\Subset\Gamma}$ of quantum potentials fulfils \eqref{Eq: quantum potential - condition for existence of tau}.
Thus, we may consider both $(\beta,\delta_\infty^\Gamma)$-KMS classical states as well as $(\beta,\delta_j^\Gamma)$-KMS quantum states for such families of classical and quantum potentials.

We now consider a sequence $(\omega_j^{\beta,\Gamma})_{j\in\mathbb{Z}_+/2}$ where for each $j\in\mathbb{Z}_+/2$ the state $\omega_j^{\beta,\Gamma}\in S(B_j^\Gamma)$ is a $(\beta,\delta_j^\Gamma)$-KMS quantum state.
Since $Q_j^\Gamma\colon\dot{B}_\infty^\Gamma\to B_j^\Gamma$ preserves positive elements, we find that $\omega_j^{\beta,\Gamma}\circ Q_j^\Gamma\in S(B_\infty^\Gamma)$ is well-defined ---here we implicitly extended $\omega_j^{\beta,\Gamma}\circ Q_j^\Gamma\colon\dot{B}_\infty^\Gamma\to\mathbb{C}$ using the continuity of $\omega_j^{\beta,\Gamma}$ and the density of $Q_j^\Gamma(\dot{B}_\infty^\Gamma)$, \textit{cf.} item \ref{Item: quantization maps are strict}.
By weak*-compactness of $S(B_\infty^\Gamma)$ the sequence $(\omega_j^{\beta,\Gamma}\circ Q_j^\Gamma)_{j\in\mathbb{Z}_+/2}$ has weak*-limit points: A natural question is whether these satisfy the $(\beta,\delta_\infty^\Gamma)$-KMS classical condition.

This problem has already been investigated in \cite{vandeVen_2024} for the case of a lattice system in a finite region.
Therein it can be shown that the sequence $(\omega_j^{\beta,\Lambda}\circ Q_j^\Lambda)_{j\in\mathbb{Z}_+/2}$ has a limit for $j\to\infty$, moreover, it holds
\begin{align}\label{Eq: limit quantum Gibbs states is classical Gibbs state}
	\lim_{j\to\infty}\omega_j^{\beta,\Lambda}\circ Q_j^\Lambda
	=\omega_\infty^{\beta,\Lambda}\,.
\end{align}
In other words the semiclassical limit of the sequence of quantum Gibbs state \eqref{Eq: KMS condition - quantum Gibbs state} is the classical Gibbs state \eqref{Eq: KMS condition - quantum Gibbs state}.
Theorem \ref{Thm: limit points of quantum KMS in oo volume are classical KMS in oo volume} generalizes this result to lattice systems on the infinite region $\Gamma$.

\subsection{Semi-classical limit of quantum KMS condition}
\label{Subsec: semi-classical of quantum KMS condition}

The goal of this section is to prove Theorem \ref{Thm: limit points of quantum KMS in oo volume are classical KMS in oo volume}, which shows that weak*-limit points of KMS quantum states are KMS classical states.

The proof requires two main ingredients.
The first one is a useful characterization of the $(\beta,\delta_\infty^\Gamma)$-KMS classical condition, \textit{cf.} Lemma \ref{Lem: classical auto-correlation lower bound}: The latter is inspired by an analogous characterization of $(\beta,\delta_j^\Gamma)$-KMS quantum states, \textit{cf.} Remark \ref{Rmk: quantum auto-correlation lower bound}.
The second piece of information is a control of the semiclassical limit of the derivation $\delta_j^\Gamma$.
Specifically, for any finite region $\Lambda\Subset\Gamma$, the DGR condition \eqref{Eq: Dirac-Groenewold-Rieffel condition} implies that $h_j^{-1}\delta_j^\Lambda \circ Q_j^\Lambda-Q_j^\Lambda\circ\delta_\infty^\Lambda\underset{j\to\infty}{\longrightarrow}0$ strongly on $\dot{B}_\infty^\Lambda$.
This is a key property which notably holds also for the lattice system on the entire $\Gamma$, \textit{cf.} Lemma \ref{Lem: classical limit of local dynamics in oo-volume}.

\bigskip

We begin with the characterization of the classical KMS condition.
The following result is essentially the classical version of the Roepstorff-Araki-Sewell auto-correlation lower bound, \textit{cf.} \cite[Thm. 5.3.15]{Bratteli_Robinson_97}.

\begin{lemma}\label{Lem: classical auto-correlation lower bound}
	Let $\mathfrak{A}$ be a commutative Poisson $C^*$-algebra, let $\delta\colon\dot{\mathfrak{A}}\to\mathfrak{A}$ be a $*$-derivation and consider $\beta\in[0,\infty)$.
	A state $\omega\in S(\mathfrak{A})$ is a $(\beta,\delta)$-KMS classical state if and only if
	\begin{align}\label{Eq: classical auto-correlation lower bound}
		-i\beta\omega(\mathfrak{a}^*\delta(\mathfrak{a}))
		\geq
		-i\omega(\{\mathfrak{a},\mathfrak{a}^*\})\,,
	\end{align}
	for all $\mathfrak{a}\in\dot{\mathfrak{A}}$.
\end{lemma}

We will call inequality \eqref{Eq: classical auto-correlation lower bound} the \textbf{classical auto-correlation lower bound}, \textit{cf.} Remark \ref{Rmk: quantum auto-correlation lower bound}.

\begin{proof}[Proof of Lemma \ref{Lem: classical auto-correlation lower bound}]
	If $\omega\in S(\mathfrak{A})$ is a $(\beta,\delta)$-KMS classical state then it also fulfils the classical auto-correlation lower bound \eqref{Eq: classical auto-correlation lower bound}: Indeed, by direct inspection
	\begin{align*}
		-i\omega(\{\mathfrak{a},\mathfrak{a}^*\})
		=-i\beta\omega(\mathfrak{a}^*\delta(\mathfrak{a}))\,.
	\end{align*}
	Conversely, let $\omega\in S(\mathfrak{A})$ be such that \eqref{Eq: classical auto-correlation lower bound} is fulfilled: We will prove that $\omega$ fulfils the $(\beta,\delta)$-KMS classical condition \eqref{Eq: KMS condition - classical on abstract C* algebra}.
	
	To begin with we observe that the classical auto-correlation lower bound implies in particular that $-i\omega(\mathfrak{a}^*\delta(\mathfrak{a}))\geq 0$ for all $\mathfrak{a}=\mathfrak{a}^*\in\dot{\mathfrak{A}}$: In particular $-i\omega(\mathfrak{a}^*\delta(\mathfrak{a}))\in\mathbb{R}$ for all $\mathfrak{a}=\mathfrak{a}^*\in\dot{\mathfrak{A}}$.
	This implies that $\omega(\delta(\mathfrak{a}))=0$ for all $\mathfrak{a}\in\dot{\mathfrak{A}}$: The proof is a classical counterpart of \cite[Lem. 5.3.16]{Bratteli_Robinson_97} and will be reviewed for the sake of clarity.
	In particular for $\mathfrak{a}=\mathfrak{a}^*\in\dot{\mathfrak{A}}$ we find
	\begin{align*}
		\omega(\delta(\mathfrak{a}^2))
		=\omega(\mathfrak{a}^*\delta(\mathfrak{a}))
		+\overline{\omega(\mathfrak{a}^*\delta(\mathfrak{a}))}
		=0\,.
	\end{align*}
	Thus $\omega(\delta(\mathfrak{a}))=0$ on all positive elements $\mathfrak{a}\in\dot{\mathfrak{A}}$: Since finite linear combinations of the latter elements generate $\dot{\mathfrak{A}}$ we conclude that $\omega\circ\delta=0$.
	
	Thus, $\omega(\delta(\mathfrak{a}))=0$ for all $\mathfrak{a}\in\dot{\mathfrak{A}}$.
	Evaluating condition \eqref{Eq: classical auto-correlation lower bound} for $\mathfrak{a}^*\in\dot{\mathfrak{A}}$ we find
	\begin{align*}
		-i\beta\omega(\mathfrak{a}\delta(\mathfrak{a}^*))
		\geq
		-i\omega(\{\mathfrak{a}^*,\mathfrak{a}\})
		=i\omega(\{\mathfrak{a},\mathfrak{a}^*\})
		\geq
		i\beta\omega(\mathfrak{a}^*\delta(\mathfrak{a}))
		=-i\beta\omega(\mathfrak{a}\delta(\mathfrak{a}^*))\,.
	\end{align*}
	where we used that $\omega(\delta(\mathfrak{a}\mathfrak{a}^*))=0$ together with the fact that $\mathfrak{A}$ is commutative.
	Thus, all inequalities must be equalities and we find
	\begin{align}\label{Eq: classical auto-correlation lower bound - equality}
		\omega(\{\mathfrak{a},\mathfrak{a}^*\})
		=\beta\omega(\mathfrak{a}^*\delta(\mathfrak{a}))\,.
	\end{align}
	Let now $\mathfrak{a},\mathfrak{b}\in\dot{\mathfrak{A}}$: Evaluation of \eqref{Eq: classical auto-correlation lower bound - equality} for $\mathfrak{a}+\mathfrak{b}$ leads to
	\begin{multline*}
		\cancel{\omega(\{\mathfrak{a},\mathfrak{a}^*\})}
		+\omega(\{\mathfrak{b},\mathfrak{a}^*\})
		+\omega(\{\mathfrak{a},\mathfrak{b}^*\})
		+\cancel{\omega(\{\mathfrak{b},\mathfrak{b}^*\})}
		\\
		=\omega(\{\mathfrak{a}+\mathfrak{b},\mathfrak{a}^*+\mathfrak{b}^*\})
		=\beta\omega((\mathfrak{a}+\mathfrak{b})^*\delta(\mathfrak{a}+\mathfrak{b}))
		\\
		=\cancel{\beta\omega((\mathfrak{a}^*\delta(\mathfrak{a}))}
		+\beta\omega(\mathfrak{b}^*\delta(\mathfrak{a}))
		+\beta\omega(\mathfrak{a}^*\delta(\mathfrak{b}))
		+\cancel{\beta\omega(\mathfrak{b}^*\delta(\mathfrak{b}))}\,.
	\end{multline*}
	Equating the terms linear in $\mathfrak{a}$ we find the $(\beta,\delta)$-KMS classical condition.
\end{proof}

\begin{remark}\label{Rmk: quantum auto-correlation lower bound}
	The classical auto-correlation lower bound \eqref{Eq: classical auto-correlation lower bound} is a classical analogue of the quantum auto-correlation lower bound, \textit{cf.} \cite[Thm. 5.3.15]{Bratteli_Robinson_97}.
	To state the latter let $\mathfrak{A}$ be a non-commutative $C^*$-algebra and let $\tau$ be a strongly continuous one-parameter group of $*$-automorphisms on $\mathfrak{A}$ with infinitesimal generator $\delta$.
	Then $\omega\in S(\mathfrak{A})$ is a $(\beta,\delta)$-KMS quantum state if and only if
	\begin{align}\label{Eq: quantum auto-correlation lower bound}
		-i\beta\omega\big(\mathfrak{a}^*\delta(\mathfrak{a})\big)
		\geq\omega(\mathfrak{a}^*\mathfrak{a})\log\bigg(
		\frac{\omega(\mathfrak{a}\mathfrak{a}^*)}{\omega(\mathfrak{a}^*\mathfrak{a})}
		\bigg)\,,
	\end{align}
	for all $\mathfrak{a}$ in the domain of $\delta$.
	In the latter inequality the function $u,v\mapsto u\log(u/v)$ is defined by
	\begin{align*}
		u\log(u/v):=
		\begin{dcases}
			u\log(u/v)
			& uv>0
			\\
			0
			& u=0\,,\,v>0
			\\
			+\infty
			& u>0\,,v=0
		\end{dcases}
	\end{align*}
	Notice that the quantum auto-correlation lower bound \eqref{Eq: quantum auto-correlation lower bound} trivialises to $\delta$-invariance ---\textit{i.e.} $\omega\circ\delta=0$--- if $\mathfrak{A}$ is commutative.
	As we will see, \textit{cf.} the proof of Theorem \ref{Thm: limit points of quantum KMS in oo volume are classical KMS in oo volume}, the quantum auto-correlation lower bound \eqref{Eq: quantum auto-correlation lower bound} reduces to the classical auto-correlation lower bond \eqref{Eq: classical auto-correlation lower bound} only for suitably scaled $*$-derivations.
\end{remark}

We now move to the discussion of the semiclassical limit of $\delta_j^\Gamma$.
To this avail we observe that the DGR condition \eqref{Eq: Dirac-Groenewold-Rieffel condition} implies
\begin{align}\label{Eq: classical limit of local dynamics}
	\lim_{j\to\infty}\Big\|h_j^{-1}\delta_j^\Lambda(Q_j^\Lambda(a_\Lambda))- Q_j^\Lambda(\delta_\infty^\Lambda(a_\Lambda))\Big\|_{B_j^\Lambda}
	=0
	\qquad
	\forall a_\Lambda\in\dot{B}_\infty^\Lambda\,.
\end{align}
The following lemma proves that the same property holds in the thermodynamical limit.

\begin{lemma}\label{Lem: classical limit of local dynamics in oo-volume}
	For all $a_\Lambda\in\dot{B}_\infty^\Lambda\subset\dot{B}_\infty^\Gamma$ it holds
	\begin{align}\label{Eq: classical limit of local dynamics in oo-volume}
		\lim_{j\to\infty}\Big\|h_j^{-1}\delta_j^\Gamma (Q_j^\Gamma(a_\Lambda))
		-Q_j^\Gamma(\delta_\infty^\Gamma(a_\Lambda))\Big\|_{B_j^\Gamma}
		=0\,.
	\end{align}
\end{lemma}
\begin{proof}
	Let $a_\Lambda\in\dot{B}_\infty^\Lambda$.
	According to Definition \eqref{Eq: SDQ for oo spin system - quantization map} we have
	\begin{align*}
		\delta_j^\Gamma\big(Q_j^\Gamma(a_\Lambda)\big)
		=\delta_j^\Gamma\big(Q_j^\Lambda(a_\Lambda)\big)
		=\frac{1}{i}\sum_{X\cap\Lambda\neq\varnothing}
		\big[Q_j^\Lambda(a_\Lambda),Q_j^X(\varphi_X)\big]\,.
	\end{align*}
	Similarly, the continuity of $Q_j^\Gamma\colon\dot{B}_\infty^\Gamma\to\dot{B}_j^\Gamma$ implies
	\begin{align*}
		Q_j^\Gamma[\delta_\infty^\Gamma(a_\Lambda)]
		=Q_j^\Gamma\Big(
		\sum_{X\cap\Lambda\neq\varnothing}\{a_\Lambda,\varphi_X\}_{B_\infty^\Lambda}
		\Big)
		=\sum_{X\cap\Lambda\neq\varnothing}
		Q_j^{\Lambda\cup X}\big(
		\{a_\Lambda,\varphi_X\}_{B_\infty^\Lambda}\big)\,.
	\end{align*}
	Overall we have
	\begin{multline*}
		\Big\|h_j^{-1}\delta_j^\Gamma\big(Q_j^\Gamma(a_\Lambda)\big)
		-Q_j^\Gamma\big(\delta_\infty^\Gamma(a_\Lambda)\big)
		\Big\|_{B_j^\Gamma}
		\\
		\leq\sum_{X\cap\Lambda\neq\varnothing}\Big\|
		\frac{1}{ih_j}[Q_j^\Lambda (a_\Lambda),Q_j^X(\varphi_X)]
		-Q_j^{\Lambda\cup X}\big(\{ a_\Lambda,\varphi_X\}_{B_\infty^\Lambda}\big)
		\Big\|_{B_j^{\Lambda\cup X}}\,.
	\end{multline*}
	For each $j\in\mathbb{Z}_+/2$ the series
	\begin{align*}
		\sum_{X\cap\Lambda\neq\varnothing}\Big\|
		\frac{1}{ih_j}[Q_j^\Lambda (a_\Lambda),Q_j^X(\varphi_X)]
		-Q_j^{\Lambda\cup X}\big(\{ a_\Lambda,\varphi_X\}_{B_\infty^\Lambda}\big)
		\Big\|_{B_j^{\Lambda\cup X}}\,,
	\end{align*}
	converges on account of condition \eqref{Eq: classical potential - condition for semiclassical limit}.
        We will now prove that it vanishes in the limit $j\to\infty$.
        Notice that each term of the series vanishes as $j\to\infty$ on account of the DGR condition \eqref{Eq: Dirac-Groenewold-Rieffel condition}.
        We will now prove that each term can be bounded as
        \begin{align*}
            \Big\|\frac{1}{ih_j}[Q_j^\Lambda (a_\Lambda),Q_j^X(\varphi_X)]
		-Q_j^{\Lambda\cup X}\big(\{ a_\Lambda,\varphi_X\}_{B_\infty^\Lambda}\big)
		\Big\|_{B_j^{\Lambda\cup X}}
            \leq c_\Lambda\|\varphi_X\|_{C^1(\mathbb{S}^2_X)}\,,
        \end{align*}
        for a $j$-independent constant $c_\Lambda>0$.
        This will allows to apply dominated convergence, concluding the proof.

        To this avail we consider the Banach space
        \begin{align*}
            C^0(\mathbb{S}^2_{\Lambda^c},C^1(\mathbb{S}^2_\Lambda))
            =C^1(\mathbb{S}^2_\Lambda)
            \otimes_\varepsilon
            B_\infty^{\Lambda^c}\,,
            \qquad
            \|f\|_{C^1(\mathbb{S}^2_\Lambda)
            \otimes_\varepsilon
            B_\infty^{\Lambda^c}}
            :=\sup_{\sigma_{\Lambda^c}\in\mathbb{S}^2_{\Lambda^c}}\|f(\sigma_{\Lambda^c})\|_{C^1(\mathbb{S}^2_\Lambda)}\,.
        \end{align*}
        where $\otimes_\varepsilon$ denotes the injective tensor product of Banach spaces, \textit{cf.} \cite[\S 3.2]{Ryan_2002}.
        Let $T_j\colon\dot{B}_\infty^\Gamma\to\dot{B}_j^\Gamma$ be the linear operator defined by
        \begin{align*}
            T_j(b_X)
            :=\frac{1}{ih_j}[Q_j^\Lambda (a_\Lambda),Q_j^X(b_X)]
		-Q_j^{\Lambda\cup X}\big(\{ a_\Lambda,b_X\}_{B_\infty^\Lambda}\big)\,,
        \end{align*}
        for all $b_X\in\dot{B}_\infty^X\subset\dot{B}_\infty^\Gamma$, $X\Subset\Gamma$.
        By direct inspection we find
        \begin{align*}
            \|T_j(b_X)\|_{B_j^\Gamma}
            \leq c_\Lambda\|b_X\|_{C^1(\mathbb{S}^2_\Lambda)\otimes_\varepsilon B_\infty^{\Lambda^c}}\,,
        \end{align*}
        which implies that $T_j$ has a unique extension $T_j\colon C^1(\mathbb{S}^2_\Lambda)\otimes_\varepsilon B_\infty^{\Lambda^c}\to B_j^\Gamma$.
        Notice that, the DGR condition \eqref{Eq: Dirac-Groenewold-Rieffel condition} entails
        $\|T_j(b_X)\|_{B_j^\Gamma}\underset{j\to\infty}{\longrightarrow}0$ for all $b_X\in\dot{B}_\infty^X$, $X\Subset\Gamma$.
        Moreover,
        \begin{align*}
            \|T_j(b_X)\|_{B_j^\Gamma}
            \leq\|T_j\|_{C^1(\mathbb{S}^2_\Lambda)\otimes_\varepsilon B_\infty^{\Lambda^c}\to B_j^\Gamma}
            \|b_X\|_{C^1(\mathbb{S}^2_\Lambda)\otimes_\varepsilon B_\infty^{\Lambda^c}}
            \leq\|T_j\|_{C^1(\mathbb{S}^2_\Lambda)\otimes_\varepsilon B_\infty^{\Lambda^c}\to B_j^\Gamma}
            \|b_X\|_{C^1(\mathbb{S}^2_X)}\,,
        \end{align*}
        where $\|T_j\|_{C^1(\mathbb{S}^2_\Lambda)\otimes_\varepsilon B_\infty^{\Lambda^c}\to B_j^\Gamma}$ denotes the operator norm of $T_j$.
        Thus, if we were able to prove that $\sup_j \|T_j\|_{C^1(\mathbb{S}^2_\Lambda)\otimes_\varepsilon B_\infty^{\Lambda^c}\to B_j^\Gamma}<\infty$ then condition \eqref{Eq: classical potential - condition for semiclassical limit} would entail that
        \begin{align*}
            \sum_{X\cap\Lambda\neq\varnothing}\Big\|
		\frac{1}{ih_j}[Q_j^\Lambda (a_\Lambda),Q_j^X(\varphi_X)]
		-Q_j^{\Lambda\cup X}\big(\{ a_\Lambda,\varphi_X\}_{B_\infty^\Lambda}\big)
		\Big\|_{B_j^{\Lambda\cup X}}
            =\sum_{X\cap\Lambda\neq\varnothing}\|T_j(\varphi_X)\|_{B_j^\Lambda}\,,
        \end{align*}
        converges, moreover, it vanishes as $j\to\infty$ by dominated convergence, concluding the proof.

        Thus, we prove that $\sup_j \|T_j\|_{C^1(\mathbb{S}^2_\Lambda)\otimes_\varepsilon B_\infty^{\Lambda^c}\to B_j^\Gamma}<\infty$.
        For that let us consider the linear map
        \begin{align*}
            T_j^\Lambda\colon C^1(\mathbb{S}^2_\Lambda)\to B_j^\Lambda
            \qquad
            T_j^\Lambda(b_\Lambda)
            :=\frac{1}{ih_j}[Q_j^\Lambda (a_\Lambda),Q_j^\Lambda(b_\Lambda)]
		  -Q_j^\Lambda\big(\{ a_\Lambda,b_\Lambda\}_{B_\infty^\Lambda}\big)\,.
        \end{align*}
        We then "lift" $T_j^\Lambda$ to
        \begin{align*}
            \hat{T}_j^\Lambda\colon C^1(\mathbb{S}^2_\Lambda)
            \to\prod_{p\in\mathbb{Z}_+/2}B_p^\Lambda\,,
            \qquad
            [\hat{T}_j^\Lambda(b_\Lambda)]_p:=
            \begin{dcases}
                T_j^\Lambda(b_\Lambda)
                & p=j
                \\
                0
                & p\neq j
            \end{dcases}\,,
        \end{align*}
 	where $\prod_{p\in\mathbb{Z}_+/2}B_p^\Lambda$ denotes the full $C^*$-direct product, \textit{cf.} \cite{Blackadar_2006}.
        By direct inspection $\hat{T}_j^\Lambda$ is linear and bounded, moreover, the DGR condition \eqref{Eq: Dirac-Groenewold-Rieffel condition} implies that
        \begin{align*}
    	\sup_{j\in\mathbb{Z}_+/2}
    	\big\|\hat{T}_j^\Lambda(b_\Lambda)\big\|_{\prod_{p\in\mathbb{Z}_+/2}B_p^\Lambda}
            =\sup_{j\in\mathbb{Z}_+/2}\sup_{p\in\mathbb{Z}_+/2}
            \|[\hat{T}_j^\Lambda(b_\Lambda)]_p\|_{B_p^\Lambda}
            =\sup_{j\in\mathbb{Z}_+/2}\|T_j^\Lambda(b_\Lambda)\|_{B_j^\Lambda}
            <\infty\,,
        \end{align*}
        for all $b_\Lambda\in C^1(\mathbb{S}^2_\Lambda)$.
        By Banach-Steinhaus Theorem it follows that
        \begin{align*}
    	\sup_{j\in\mathbb{Z}_+/2}\big\|\hat{T}_j^\Lambda\big\|_{C^1(\mathbb{S}^2_\Lambda)\to\prod_{p\in\mathbb{Z}_+/2}B_p^\Lambda}
       	<\infty\,.
        \end{align*}
        We then consider the inductive tensor product of $\hat{T}_j^\Lambda\colon C^1(\mathbb{S}^2_\Lambda)\to\prod_{p\in\mathbb{Z}_+/2}B_p^\Lambda$ with $Q_j^{\Lambda^c}\colon B_\infty^{\Lambda^c}\to B_j^{\Lambda^c}$, \textit{cf.} \cite[Prop. 3.2]{Ryan_2002}.
        The latter is the unique bounded linear map
        \begin{align*}
            \hat{T}_j^\Lambda\otimes_\varepsilon Q_j^{\Lambda^c}\colon
            C^1(\mathbb{S}^2_\Lambda)
            \otimes_\varepsilon B_\infty^{\Lambda^c}
            \to \prod_{p\in\mathbb{Z}_+/2}B_p^\Lambda
            \otimes
            B_j^{\Lambda^c}\,,
        \end{align*}
        which extends the algebraic tensor product $\hat{T}_j^\Lambda\otimes Q_j^{\Lambda^c}$, that is,
        \begin{align*}
            (\hat{T}_j^\Lambda\otimes_\varepsilon Q_j^{\Lambda^c})(b_\Lambda\otimes b_{\Lambda^c})
            =\hat{T}_j^\Lambda(b_\Lambda)\otimes Q_j^{\Lambda^c}(b_{\Lambda^c})\,,
        \end{align*}
        for all $b_\Lambda\in C^1(\mathbb{S}^2_\Lambda)$ and $b_{\Lambda^c}\in B_\infty^{\Lambda^c}$.
        Notice that $\prod_{p\in\mathbb{Z}_+/2}B_p^\Lambda\otimes B_j^{\Lambda^c}$ is a $C^*$-algebra once completed with its unique $C^*$-cross norm.
        The uniqueness of such norm is due to the fact that $B_j^{\Lambda^c}$ is the $C^*$-inductive limit of finite dimensional, hence nuclear, $C^*$-algebras, therefore, it is nuclear as well, \textit{cf.} \cite[\S II.9.4.5]{Blackadar_2006}.
        We recall that
        \begin{align*}
            \|\hat{T}_j^\Lambda\otimes_\varepsilon Q_j^{\Lambda^c}\|_{C^1(\mathbb{S}^2_\Lambda)
            \otimes_\varepsilon B_\infty^{\Lambda^c}
            \to \prod_{p\in\mathbb{Z}_+/2}B_p^\Lambda
            \otimes
            B_j^{\Lambda^c}}
            &=
            \|\hat{T}_j^\Lambda\|_{C^1(\mathbb{S}^2_\Lambda)\to \prod_{p\in\mathbb{Z}_+/2}B_p^\Lambda}
            \|Q_j^{\Lambda^c}\|_{B_\infty^{\Lambda^c}\to B_j^{\Lambda^c}}
            \\
            &=\|\hat{T}_j^\Lambda\|_{C^1(\mathbb{S}^2_\Lambda)\to \prod_{p\in\mathbb{Z}_+/2}B_p^\Lambda}\,,
        \end{align*}
        therefore, $\sup_j \|\hat{T}_j^\Lambda\otimes_\varepsilon Q_j^{\Lambda^c}\|_{C^1(\mathbb{S}^2_\Lambda)
        \otimes_\varepsilon B_\infty^{\Lambda^c}
        \to \prod_{p\in\mathbb{Z}_+/2}B_p^\Lambda
        \otimes
        B_j^{\Lambda^c}}<\infty$.
        Finally, let
        \begin{align*}
            \hat{T}_j\colon C^1(\mathbb{S}^2_\Lambda)
            \otimes_\varepsilon B_\infty^{\Lambda^c}
            \to\prod_{p\in\mathbb{Z}_+/2}B_p^\Lambda
            \otimes B_j^{\Lambda^c}
            \qquad
            [\hat{T}_j(f)]_p:=
            \begin{cases}
                T_j(f)
                & p=j
                \\
                0
                & p\neq j
            \end{cases}\,.
        \end{align*}
        Then $\hat{T}_j$ is linear and continuous, moreover, by direct inspection
        \begin{align*}
            \hat{T}_j(b_\Lambda\otimes b_{\Lambda^c})
            =\hat{T}_j^\Lambda(b_\Lambda)
            \otimes Q_j^{\Lambda^c}(b_{\Lambda^c})\,.
        \end{align*}
        It follows that $\hat{T}_j=\hat{T}_j^\Lambda\otimes_\varepsilon Q_j^{\Lambda^c}$ and therefore
        \begin{align*}
            \|T_j\|_{C^1(\mathbb{S}^2_\Lambda)
            \otimes_\varepsilon B_\infty^{\Lambda^c}
            \to B_j^\Gamma}
            &=\|\hat{T}_j\|_{C^1(\mathbb{S}^2_\Lambda)
            \otimes_\varepsilon B_\infty^{\Lambda^c}
            \to\prod_{p\in\mathbb{Z}_+/2}B_p^\Lambda
            \otimes B_j^{\Lambda^c}}
            \\
            &=\|\hat{T}_j^\Lambda\otimes_\varepsilon Q_j^{\Lambda^c}\|_{C^1(\mathbb{S}^2_\Lambda)
            \otimes_\varepsilon B_\infty^{\Lambda^c}
            \to \prod_{p\in\mathbb{Z}_+/2}B_p^\Lambda
            \otimes
            B_j^{\Lambda^c}}
            \\
            &=\|\hat{T}_j^\Lambda\|_{C^1(\mathbb{S}^2_\Lambda)\to \prod_{p\in\mathbb{Z}_+/2}B_p^\Lambda}\,,
        \end{align*}
        showing that $\sup_j\|T_j\|_{C^1(\mathbb{S}^2_\Lambda)\otimes_\varepsilon B_\infty^{\Lambda^c}}<\infty$.

\end{proof}

We are in position to prove Theorem \ref{Thm: limit points of quantum KMS in oo volume are classical KMS in oo volume}, which guarantees that, for any sequence $(\omega_j^{\beta,\Gamma})_{j\in\mathbb{Z}_+/2}$ of $(\beta,\delta_j^\Gamma)$-KMS quantum states on $B_j^\Gamma$, any weak*-limit point of the sequence $(\omega_j^{\beta,\Gamma}\circ Q_j^\Gamma)_{j\in\mathbb{Z}_+/2}\subset S(B_\infty^\Gamma)$ is a $(\beta,\delta_\infty^\Gamma)$-KMS classical states.
At its core, the proof is based on the observation that the quantum $(\beta,\delta_j^\Gamma)$-KMS condition coincides with the quantum $(\beta/h,h\delta_j)$-KMS condition, $h=2j+1$, together with the application of Lemma \ref{Lem: classical limit of local dynamics in oo-volume} and of the classical auto-correlation lower bound \eqref{Eq: classical auto-correlation lower bound}.

\begin{theorem}
	\label{Thm: limit points of quantum KMS in oo volume are classical KMS in oo volume}
	For all $j\in\mathbb{Z}_+/2$ let $\omega_j^{\beta,\Gamma}\in S(B_j^\Gamma)$ be a $(\beta,\delta_j^\Gamma)$-KMS quantum state on $B_j^\Gamma$.
	Then, any weak*-limit point $\omega_\infty^{\beta,\Gamma}\in S(B_\infty^\Gamma)$ of the sequence $(\omega_j^{\beta,\Gamma}\circ Q_j^\Gamma)_{j\in\mathbb{Z}_+/2}$ is a $(\beta,\delta_\infty^\Gamma)$-KMS classical state.
\end{theorem}

\begin{proof}
	Up to moving to a subsequence in $j$ we may assume
	\begin{align*}
		\omega_\infty^{\beta,\Gamma}(a_\Lambda)
		=\lim_{j\to\infty}
		\omega_j^{\beta,\Gamma}(Q_j^\Lambda(a_\Lambda))
		\qquad\forall a_\Lambda\in B_\infty^\Lambda\subset B_\infty^\Gamma\,.
	\end{align*}
	We will prove that $\omega_\infty^{\beta,\Gamma}$ fulfils the classical $(\beta,\delta_\infty^\Gamma)$-KMS conditions by proving condition \eqref{Eq: classical auto-correlation lower bound}, \textit{cf.} Lemma \ref{Lem: classical auto-correlation lower bound}.
	To this avail let $a_\Lambda\in\dot{B}_\infty^\Lambda$ and consider
	\begin{align*}
		-i\beta\omega_\infty^{\beta,\Gamma}\big(a_\Lambda^*\delta_\infty^\Gamma (a_\Lambda)\big)
		&=\lim_{j\to\infty}(-i\beta)\omega_j^{\beta,\Gamma}\big(Q_j^\Gamma(a_\Lambda^*\delta_\infty^\Gamma(a_\Lambda))\big)
		\\
		&=\lim_{j\to\infty}(-i\beta)\omega_j^{\beta,\Gamma}\big(
		Q_j^\Lambda(a_\Lambda^*)
		Q_j^\Gamma(\delta_\infty^\Gamma(a_\Lambda))\big)
		\\
		&=\lim_{j\to\infty}-i\beta h_j^{-1}\omega_j^{\beta,\Gamma}\big(
		Q_j^\Lambda(a_\Lambda^*)
		\delta_j^\Gamma (Q_j^\Lambda (a_\Lambda))\big)
		&\textrm{Lem. } \ref{Lem: classical limit of local dynamics in oo-volume}
		\,.
	\end{align*}
	At this stage we may apply the quantum auto-correlation lower bound \eqref{Eq: quantum auto-correlation lower bound} for $\omega_j^{\beta,\Gamma}$, \textit{cf.} Remark \ref{Rmk: quantum auto-correlation lower bound}:
	\begin{align*}
		-i\beta\omega_j^{\beta,\Gamma}\big(
		 Q_j^\Lambda(a_\Lambda^*)
		\delta_j^\Gamma ( Q_j^\Lambda(a_\Lambda))\big)
		\geq
		\omega_j^{\beta,\Gamma}\big(
		 Q_j^\Lambda(a_\Lambda^*)
		 Q_j^\Lambda(a_\Lambda)
		\big)
		\log\bigg(
		\frac{\omega_j^{\beta,\Gamma}\big(
			 Q_j^\Lambda(a_\Lambda^*)
			 Q_j^\Lambda(a_\Lambda)
			\big)}
		{\omega_j^{\beta,\Gamma}\big(
			 Q_j^\Lambda(a_\Lambda)
			 Q_j^\Lambda(a_\Lambda^*)
			\big)}
		\bigg)\,.
	\end{align*}
	Thus, the DGR condition \eqref{Eq: Dirac-Groenewold-Rieffel condition} leads to
	\begin{multline*}
		-i\beta\omega_\infty^{\beta,\Gamma}\big(
		 a_\Lambda^*\delta_\infty(a_\Lambda)
		\big)
		\geq-\lim_{j\to\infty}
		\omega_j^{\beta,\Gamma}\big(
		 Q_j^\Lambda(a_\Lambda^*)
		 Q_j^\Lambda(a_\Lambda)
		\big)
		h_j^{-1}\log\bigg(1+
		\frac{\omega_j^{\beta,\Gamma}\big(
			[ Q_j^\Lambda(a_\Lambda),
			 Q_j^\Lambda(a_\Lambda^*)]
			\big)}
		{\omega_j^{\beta,\Gamma}\big(
			 Q_j^\Lambda(a_\Lambda^*)
			 Q_j^\Lambda(a_\Lambda)
			\big)}
		\bigg)
		\\
		=-\lim_{j\to\infty}
		\omega_j^{\beta,\Gamma}\big(
		 Q_j^\Lambda(a_\Lambda^*)
		 Q_j^\Lambda(a_\Lambda)
		\big)
		h_j^{-1}\log\bigg(1
		+\frac{i}{2j+1}\frac{\omega_j^{\beta,\Gamma}\big(
		Q_j^\Lambda(\{a_\Lambda,a_\Lambda^*\}_{B_\infty^\Lambda})
		\big)}
		{\omega_j^{\beta,\Gamma}\big(
			 Q_j^\Lambda(a_\Lambda^*)
			 Q_j^\Lambda(a_\Lambda)
			\big)}
		\bigg)
		\\
		=-i\lim_{j\to\infty}
		\omega_j^{\beta,\Gamma}\big(
		Q_j^\Lambda(\{a_\Lambda,a_\Lambda^*\}_{B_\infty^\Lambda})
		\big)
		=-i\omega_\infty^{\beta,\Gamma}\Big(\{ a_\Lambda, a_\Lambda^*\}_{B_\infty^\Gamma}\Big)\,.
	\end{multline*}
\end{proof}

\begin{remark}\label{Rmk: classical limit of quantum KMS in absence of phase transition}
	Proposition \ref{Thm: limit points of quantum KMS in oo volume are classical KMS in oo volume} has the following important consequence.
	Let $\beta\in[0,\infty)$ be such that there exists a unique $(\beta,\delta_\infty^\Gamma)$-KMS classical state on $B_\infty^\Gamma$, \textit{i.e.}, there are no classical phase transitions at $\beta$.
	In particular, there exists a unique limit point $\omega_\infty^{\beta,\Gamma}\in S(B_j^\Gamma)$ of $(\omega_\infty^{\beta,\Lambda})_{\Lambda\Subset\Gamma}$, \textit{cf.} Remark \ref{Rmk: KMS state on BooGamma seen as weak*-limit points of Gibbs states}.
	Let $(\omega_j^{\beta,\Gamma})_{j\in\mathbb{Z}_+/2}$ be any sequence such that $\omega_j^{\beta,\Gamma}\in S(B_j^\Gamma)$ is a $(\beta,\delta_j^\Gamma)$-KMS quantum state on $S(B_j^\Gamma)$ for all $j\in\mathbb{Z}_+/2$.
	Then, any limit point of the sequence $(\omega_j^{\beta,\Gamma}\circ Q_j^\Gamma)_{j\in\mathbb{Z}_+/2}$ is a $(\beta,\delta_\infty^\Gamma)$-KMS classical state by Proposition \ref{Thm: limit points of quantum KMS in oo volume are classical KMS in oo volume}.
	However, by assumption there is only one of such states, thus, the whole sequence $(\omega_j^{\beta,\Gamma}\circ Q_j^\Gamma)_{j\in\mathbb{Z}_+/2}$ converges to the unique classical $(\beta,\delta_\infty^\Gamma)$-KMS classical states $\omega_\infty^{\beta,\Gamma}$.
	In other words the absence of classical phase transitions implies that any sequence of $(\beta,\delta_j^\Gamma)$-KMS quantum states has a semiclassical limit.
	
	In the particular case where there is a unique $(\beta,\delta_j^\Gamma)$-KMS quantum state $\omega_j^{\beta,\Gamma}$ on $B_j^\Gamma$ ---\textit{i.e.}, there are no quantum phase transitions for all $j\in\mathbb{Z}_+/2$ at $\beta$--- the previous observation boils down to the equality
	\begin{align*}
		\lim_{j\to\infty}\Big(\lim_{\Lambda\uparrow\Gamma}
		\omega_j^{\beta,\Lambda}\Big)\circ Q_j^\Gamma
		=\omega_\infty^{\beta,\Gamma}
		=\lim_{\Lambda\uparrow\Gamma}\lim_{j\to\infty}
		\omega_j^{\beta,\Lambda}\circ Q_j^\Gamma\,,
	\end{align*}
	where in the last equality we used the result of \cite{vandeVen_2024} together with Remark \ref{Rmk: KMS state on BooGamma seen as weak*-limit points of Gibbs states}.
\end{remark}

\section{Common absence of CPTs and of QPTs at high temperatures}
\label{Sec: Common absence of CPTs and of QPTs}

In this section we will explore the relation between classical and quantum phase transitions within the setting of the Berezin quantization on $\Gamma$, \textit{cf.} Section \ref{Sec: Berezin SDQ on a lattice system}.
Our goal is to prove that, under mild conditions, above a common threshold critical temperature there is absence of both classical and quantum phase transitions, that is, uniqueness of both $(\beta,\delta_\infty^\Gamma)$-KMS classical states on $B_\infty^\Gamma$ and $(\beta,\delta_j^\Gamma)$-KMS quantum states on $B_j^\Gamma$ for all $j\in\mathbb{Z}_+/2$ holds for $\beta\leq\beta_c$ for a given $\beta_c\in(0,+\infty)$.

To this avail we will consider \cite[Prop. 6.2.45]{Bratteli_Robinson_97}, which provides a mild sufficient condition on a family of quantum potential $\Phi_j=\{\Phi_{j,\Lambda}\}_{\Lambda\Subset\Gamma}$ which ensures uniqueness of the corresponding $(\beta,\delta_j^\Gamma)$-KMS quantum state on $B_j^\Gamma$ for all $\beta\leq\beta_j$, where $\beta_j$ is called \textbf{quantum critical inverse temperature}.
This result does not suffices to our purposes because:
(a) there is no clear classical counterpart for the assumption on the family $\Phi_j$;
(b) the estimate for the critical inverse temperature $\beta_j$ provided in \cite[Prop. 6.2.45]{Bratteli_Robinson_97} is not uniform in $j\in\mathbb{Z}_+/2$, spoiling a comparison with the classical critical inverse temperature ---\textit{cf.} Remark \ref{Rmk: NO CPT implies NO QPT; comparison with BR result}.

To encompass these issues we revised the proof of \cite[Prop. 6.2.45]{Bratteli_Robinson_97} in order to produce a new result which is uniform in $j\in\mathbb{Z}_+/2$.
In particular, in Section \ref{Subsec: uniqueness result for classical KMS state} we prove a classical analogous of \cite[Prop. 6.2.45]{Bratteli_Robinson_97}, \textit{cf.} Theorem \ref{Thm: uniqueness result for classical KMS states}.
The latter result provides a sufficient condition, \textit{cf.} \eqref{Eq: uniqueness result for classical KMS states - assumption on potential}, on the family of potential $\varphi=\{\varphi_\Lambda\}_{\Lambda\Subset\Gamma}$ which ensures uniqueness for the corresponding $(\beta,\delta_\infty^\Gamma)$-KMS classical state on $B_\infty^\Gamma$ provided $\beta\leq\beta_\infty$, $\beta_\infty$ being the \textbf{classical critical inverse temperature}.
Condition \eqref{Eq: uniqueness result for classical KMS states - assumption on potential} is essentially a classical version of the one required in \cite{Bratteli_Robinson_97}, \textit{cf.} \eqref{Eq: uniqueness result for quantum KMS states - BR assumption on potential}, and essentially boils down to an estimate on the sup-norm of the derivatives of the potentials $\varphi_\Lambda$, $\Lambda\Subset\Gamma$.

Next, we provide a uniqueness result, \textit{cf.} Theorem \ref{Thm: uniqueness result for quantum KMS states}, for $(\beta,\delta_j^\Gamma)$-KMS quantum states on $B_j^\Gamma$ associated with the family of quantum potentials $\Phi_{j,\Lambda}=Q_j^\Gamma(\varphi_\Lambda)$ obtained by Berezin quantization of the classical potential.
The novelty of this result relies on the bound on the critical inverse temperature, which is now uniform in $j\in\mathbb{Z}_+/2$, allowing for a better comparison with the classical critical inverse temperature.
Finally, it is shown that the assumptions of Theorem \ref{Thm: uniqueness result for quantum KMS states} implies those of Theorem \ref{Thm: uniqueness result for classical KMS states}, \textit{cf.} Remark \ref{Rmk: NO CPT implies NO QPT; comparison with BR result}.
For this class of models, this establishes a sufficient criterion for the absence of both classical and quantum phase transitions.

\subsection{Uniqueness result for classical KMS states}
\label{Subsec: uniqueness result for classical KMS state}

The goal of this section is to prove a uniqueness result for $(\beta,\delta_\infty^\Gamma)$-KMS classical states, \textit{cf.} Theorem \ref{Thm: uniqueness result for classical KMS states}.
The latter identifies a sufficient condition on the classical interaction potentials $\varphi=\{\varphi_\Lambda\}_{\Lambda\Subset\Gamma}$ which ensures uniqueness of $(\beta,\delta_\infty^\Gamma)$-KMS classical states on $B_\infty^\Gamma$ for sufficiently low $\beta$, that is, at sufficiently high temperature.
This result is inspired from \cite[Prop. 6.2.45]{Bratteli_Robinson_97} and  adapted to the classical setting.
In the forthcoming sections we will present an analogous uniqueness result in the quantum setting, \textit{cf.} Theorem \ref{Thm: uniqueness result for quantum KMS states}.

\bigskip

\begin{theorem}\label{Thm: uniqueness result for classical KMS states}
	Let $\varphi=\{\varphi_\Lambda\}_{\Lambda\Subset\Gamma}$ with $\varphi\in C^{2s}(\mathbb{S}^2_\Lambda)$, $s>7/4$, be such that
	\begin{align}\label{Eq: uniqueness result for classical KMS states - assumption on potential}
		\|\varphi\|_{0,s}
		:=\sum_{m\geq 0}
		(C_\Delta^sK_s)^m\sup_{x\in\Gamma}\sum_{\substack{|\Lambda|=m+1\\ x\in \Lambda}}
		\|\varphi_\Lambda\|_{C^{2s}(\mathbb{S}^2_\Lambda)}
		<+\infty\,,
	\end{align}
	where $\|\;\|_{C^{2s}(\mathbb{S}^2_\Lambda)}$ is the $C^{2s}$-norm on $C^{2s}((\mathbb{S}^2)^{|\Lambda|})$ while $K_s>1$ and $C_\Delta\geq 1$ are defined by
	\begin{align}\label{Eq: uniqueness result for classical KMS states - useful constants}
		K_s:=\sum_{\ell\in\mathbb{Z}_+}\frac{(2\ell+1)^{5/2}}{[1+\ell(\ell+1)]^s}\,,
		\qquad
		C_\Delta:=\|1-\Delta_{\mathbb{S}^2}\|_{C^2(\mathbb{S}^2)\to C^0(\mathbb{S}^2)}\,,
	\end{align}
	the latter being the operator norm of $1-\Delta_{\mathbb{S}^2}\colon C^2(\mathbb{S}^2)\to C^0(\mathbb{S}^2)$.
	Then there exists a unique $(\beta,\delta_\infty^\Gamma)$-KMS classical state on $B_\infty^\Gamma$ for all $\beta\in[0,\beta_{0,s})$ where
	\begin{align}\label{Eq: uniqueness result for classical KMS states - estimate on critical temperature}
		\beta_{0,s}
		:=\frac{\log 2}{2C_\Delta^sK_s\|\varphi\|_{0,s}}
		\,.
	\end{align}
\end{theorem}

The proof of Theorem \ref{Thm: uniqueness result for classical KMS states} is rather long and will take the whole section.
For the sake of clarity we start by collecting some observations in a series of Remarks \ref{Rmk: spherical harmonics - set-up, identities and useful bounds}-\ref{Rmk: Banach space for state - main definitions}-\ref{Rmk: uniqueness result for classical KMS states - beta=0 case}.
Together with Lemma \ref{Lem: classical KMS condition and rotations} these will develop all technical tools need for the proof of Theorem \ref{Thm: uniqueness result for classical KMS states}.

\begin{remark}\label{Rmk: spherical harmonics - set-up, identities and useful bounds}
	\noindent
	For later convenience we will recollect in this remark a few useful observations concerning the Fourier-Laplace expansion on $\mathbb{S}^2$ and $\mathbb{S}^2_\Lambda$, $\Lambda\Subset\Gamma$.
	
	With reference to Example \ref{Ex: Berezin SDQ - quantization of spherical harmonics} we denote by $Y_{\ell,m}\in\dot{B}_\infty$ the spherical harmonic with parameter $\ell\in\mathbb{Z}_+$, $m\in[-\ell,\ell]\cap\mathbb{Z}$.
	With the choice of normalization of Equation \eqref{Eq: spherical harmonics - convention} we have $\|Y_{\ell,m}\|_{B_\infty}\leq 1$, moreover, the set $\{Y_{\ell,m}\}_{\ell,m}$ is a complete orthogonal system for $L^2(\mathbb{S}^2,\mu_0)$ made by orthogonal but not $L^2$-normalized vectors.
	In particular for any $a\in L^2(\mathbb{S}^2,\mu_0)$ we have the following Fourier-Laplace expansion:
	\begin{align}\label{Eq: Fourier-Laplace expansion - single site}
		a=\sum_{\ell\in\mathbb{Z}_+}\sum_{m=-\ell}^{\ell}
		\hat{a}(\ell,m)Y_{\ell,m}\,,
		\qquad
		\hat{a}(\ell,m)
		:=(2\ell+1)\langle Y_{\ell,m}|a\rangle_{L^2(\mathbb{S}^2,\mu_0)}\,.
	\end{align}
	It is worth recalling that the series in Equation \eqref{Eq: Fourier-Laplace expansion - single site} converges in the $L^2$-norm, moreover, the series converges uniformly with all derivatives if $a\in\dot{B}_\infty$.
	Furthermore, for $\ell,s\in\mathbb{Z}_+$
	\begin{align*}
		|\hat{a}(\ell,m)|
		&=(2\ell+1)[1+\ell(\ell+1)]^{-s}|\langle Y_{\ell,m}|(1-\Delta_{\mathbb{S}^2})^s a\rangle_{L^2(\mathbb{S}^2,\mu_0)}|
		\\
		&\leq C_\Delta^s(2\ell+1)[1+\ell(\ell+1)]^{-s}
		\|a\|_{C^{2s}(\mathbb{S}^2)}\|Y_{\ell,m}\|_{L^1(\mathbb{S}^2,\mu_0)}
		\\
		&\leq C_\Delta^s(2\ell+1)^{1/2}[1+\ell(\ell+1)]^{-s}
		\|a\|_{C^{2s}(\mathbb{S}^2)}\,,
	\end{align*}
	where $C_\Delta$ has been defined in Equation \eqref{Eq: uniqueness result for classical KMS states - useful constants} while we used the eigenvalues property $-\Delta_{\mathbb{S}^2}Y_{\ell,m}=\ell(\ell+1)Y_{\ell,m}$ together with integration by parts and Cauchy-Schwarz inequality.
	The bound above ensures the uniform convergence of the series in Equation \eqref{Eq: Fourier-Laplace expansion - single site} provided that $1/2-2s+1<-1$, the additional $+1$ factor arising from the summation over $m\in[-\ell,\ell]$.
	
	The previous considerations generalize to the case of $B_\infty^\Lambda$, $\Lambda\Subset\Gamma$.
	In particular, if $\ell_\Lambda\in\mathbb{Z}_+^\Lambda$ and $m_\Lambda\in\mathbb{Z}^\Lambda$, with $m_x\in[-\ell_x,\ell_x]\cap\mathbb{Z}$ for all $x\in\Lambda$, we set
	\begin{align}
		Y_{\ell_\Lambda,m_\Lambda}
		:=\bigotimes_{x\in\Lambda}Y_{\ell_x,m_x}
		\in\dot{B}_\infty^\Lambda\,.
	\end{align}
	The Fourier-Laplace expansion of $a_\Lambda\in\dot{B}_\infty^\Lambda$ is given by
	\begin{align}\label{Eq: Fourier-Laplace expansion - finite region}
		a_\Lambda=\sum_{\ell_\Lambda\in\mathbb{Z}_+^\Lambda}\sum_{m_\Lambda}
		\hat{a}(\ell_\Lambda,m_\Lambda)Y_{\ell_\Lambda,m_\Lambda}\,,
		\qquad
		\hat{a}(\ell_\Lambda,m_\Lambda)
		:=\Big(\prod_{x\in\Lambda}(2\ell_x+1)\Big)\langle Y_{\ell_\Lambda,m_\Lambda}|a_\Lambda\rangle_{L^2(\mathbb{S}^2_\Lambda,\mu_0^\Lambda)}\,,
	\end{align}
	where we may estimate
	\begin{align*}
		|\hat{a}(\ell_\Lambda,m_\Lambda)|
		\leq C_\Delta^{s|\Lambda|}
		\prod_{x\in\Lambda}
		\bigg(\frac{(2\ell_x+1)^{1/2}}{[1+\ell_x(\ell_x+1)]^s}\bigg)
		\|a_\Lambda\|_{C^{2s}(\mathbb{S}^2_\Lambda)}\,.
	\end{align*}

\end{remark}

\begin{remark}\label{Rmk: Banach space for state - main definitions}
	Let $\omega_\infty^\Gamma\in S(B_\infty^\Gamma)$ be an arbitrary state on $B_\infty^\Gamma$.
	Then $\omega_\infty^\Gamma$ is uniquely determined by its values on $\dot{B}_\infty^\Lambda$ for all $\Lambda\Subset\Gamma$.
	On account of Remark \ref{Rmk: spherical harmonics - set-up, identities and useful bounds} ---\textit{cf.} Equation \eqref{Eq: Fourier-Laplace expansion - single site}--- it suffices to determine the values
	\begin{align}\label{Eq: classical state functions}
		\underline{\omega}_\infty^\Gamma(\ell_\Lambda,m_\Lambda)
		:=\omega_\infty^\Gamma( Y_{\ell_\Lambda,m_\Lambda})\,,
		\qquad
		\underline{\omega}_\infty^\Gamma(\ell_\varnothing,m_\varnothing)
		:=\omega_\infty^\Gamma(1)=1\,,
	\end{align}
	where $\ell_\Lambda\in\mathbb{N}^\Lambda$ while $m_\Lambda\in\mathbb{Z}^\Lambda$ are such that $m_x\in[-\ell_x,\ell_x]\cap\mathbb{Z}$ for all $x\in\Lambda$.
	Notice that we avoided to consider the case $\ell_x=0$ for some $x\in\Lambda$ because of the following compatibility condition: If $\ell_x=0$ then
	\begin{align*}
		\underline{\omega}_\infty^\Gamma(\ell_\Lambda,m_\Lambda)
		:=\omega_\infty^\Gamma( Y_{\ell_\Lambda,m_\Lambda})
		=\omega_\infty^\Gamma(Y_{\ell_{\Lambda\setminus\{x\}},m_{\Lambda\setminus\{x\}}})
		=\underline{\omega}_\infty^\Gamma(\ell_{\Lambda\setminus\{x\}},m_{\Lambda\setminus\{x\}})\,.
	\end{align*}

	At this stage we may regard $\underline{\omega}_\infty^\Gamma$ as an element of a suitable Banach space $\underline{\mathsf{X}}$.
	The latter is defined by
	\begin{multline}\label{Eq: Banach space for classical state functions}
		\underline{\mathsf{X}}:=
		\{\underline{f}=(f_\Lambda)_{\Lambda\Subset\Gamma}\,|\,
		f_\Lambda\colon\{
		(\ell_\Lambda,m_\Lambda)\in\mathbb{N}^\Lambda\times\mathbb{Z}^\Lambda\,|\,m_x\in[-\ell_x,\ell_x]\cap\mathbb{Z}\;\forall x\in\Lambda
		\}\to\mathbb{C}\colon
		\\
		\|\underline{f}\|_{\underline{\mathsf{X}}}:=\sup_{\Lambda\Subset\Gamma}\sup_{\ell_\Lambda,m_\Lambda}|f_\Lambda(\ell_\Lambda,m_\Lambda)|<\infty
		\}\,,
	\end{multline}
	where $f_\varnothing\in\mathbb{C}$.
	It is worth observing that $\underline{\omega}_\infty^\Gamma\in\underline{\mathsf{X}}$ because of our choice for the normalization of $Y_{\ell_\Lambda,m_\Lambda}$, \textit{cf.} Example \ref{Ex: Berezin SDQ - quantization of spherical harmonics}, which ensures that $\|\underline{\omega}_\infty^\Gamma\|_{\underline{\mathsf{X}}}\leq 1$.
	
	Summing up, a state $\omega_\infty^\Gamma\in S(B_\infty^\Gamma)$ is completely determined by the corresponding element $\underline{\omega}_\infty^\Gamma\in\underline{\mathsf{X}}$.
\end{remark}

The following Lemma is crucial and shows the relation between the classical $(\beta,\delta_\infty^\Gamma)$-KMS condition and the action of $SU(2)$ (in fact, of $SO(3)$) on $B_\infty^\Gamma$.
In particular for all $x\in\Gamma$ let $R_x\in SU(2)$.
Recalling Section \ref{Sec: Berezin SDQ on a lattice system} let $\boxed{\hat{R}_x}\in\mathcal{B}(L^2(\mathbb{S}^2,\mu_0))$ be the corresponding unitary operator on $L^2(\mathbb{S}^2,\mu_0)$ defined by $\hat{R}_xa_x:=a_x\circ R_x^{-1}$.
If $R_x=\exp[D_x]$, $D_x\in\mathfrak{su}(2)$, we will denote by $\hat{D}_x$ the vector field on $\mathbb{S}^2_x$ corresponding to the infinitesimal generator of $\hat{R}_x(t)$, where $R_x(t):=\exp[tD_x]$ ---for definiteness $\hat{R}_x(t)=\exp[-it\hat{D}_x]$.
The action of $\hat{R}_x$ on $B_\infty^x$ can be lifted to $B_\infty^\Gamma$ by setting, for all $a_\Lambda\in B_\infty^\Lambda$ and $\Lambda\Subset\Gamma$, $\hat{R}_xa_\Lambda(\sigma_\Lambda):=a_\Lambda(\sigma_\Lambda')$ with $\sigma_\Lambda=(\sigma_y)_{y\in\Lambda}$, $\sigma_\Lambda'=(\sigma_y')_{y\in\Lambda}$ being $\sigma_y'=\sigma_y$ for $y\neq x$ while $\sigma_x':=R_x^{-1}\sigma_x$.

\begin{lemma}\label{Lem: classical KMS condition and rotations}
	Let $\omega_\infty^{\beta,\Gamma}\in S(B_\infty^\Gamma)$ be a $(\beta,\delta_\infty^\Gamma)$-KMS classical states.
	Then, for any $x\in\Gamma$, $R_x\in SU(2)$ and $a_\Lambda\in\dot{B}_\infty^\Lambda$, $\Lambda\Subset\Gamma$, it holds
	\begin{align}\label{Eq: classical KMS condition and rotations}
		\omega_\infty^{\beta,\Gamma}( a_\Lambda)
		=\omega_\infty^{\beta,\Gamma}\Big(e^{\beta\sum_{X\ni x}(I-\hat{R}_x)\varphi_X}
		\hat{R}_x a_\Lambda
		\Big)
		\,.
	\end{align}
\end{lemma}

\begin{remark}
	\noindent
	\begin{enumerate}[(i)]
		\item
		Lemma \ref{Lem: classical KMS condition and rotations} is inspired by \cite[Prop. 5]{Drago_Waldmann_2024}, where Equation \eqref{Eq: classical KMS condition and rotations} has been proved in the context of finite dimensional Poisson manifolds.
		Equation \ref{Eq: classical KMS condition and rotations} can be understood as an adaptation to the present setting.
		
		\item
		Notice that the element $\sum_{X\ni x}(I-\hat{R}_x)\varphi_X\in B_\infty^\Gamma$ can be interpreted as the difference $h_{\Gamma}-\hat{R}_xh_{\Gamma}$ between the (ill-defined) classical Hamiltonian $h_{\Gamma}=\sum_{X\Subset\Gamma}\varphi_X$ and its $R_x$-rotated version: Loosely speaking, the locality of the involved rotation ensures that the above difference is finite.
	\end{enumerate}
\end{remark}

\begin{proof}[Proof of Lemma \ref{Lem: classical KMS condition and rotations}]
	Let $\hat{D}_x$ be the infinitesimal generator of $\hat{R}_x$.
	Since $\mathbb{S}^2_x$ is a symplectic manifold and $\hat{D}_x$ is a Poisson vector field on $\mathbb{S}^2_x$, there exists a closed 1-form $\alpha_x\in\Omega^1(\mathbb{S}^2_x)$ such that $\hat{D}_xa_x=\pi_x(\mathrm{d}a_x,\alpha_x)$ for all $a_x\in \dot{B}_\infty^x$, $\pi_x$ being the Poisson bitensor associated with the symplectic 2-form on $\mathbb{S}^2_x$.
	Let $\{U_x\}$ be an open cover of $\mathbb{S}^2_x$ such that $\alpha_x=\mathrm{d}g_{U_x}$, $g_{U_x}\in C^\infty(U_x)$, and let $\{\chi_{U_x}\}$ be a partition of unity associated with $\{U_x\}$.
	It follows that, for all $a_x\in\dot{B}_\infty^x$,
	\begin{align*}
		\hat{D}_xa_x
		=\sum_{U_x}\hat{D}_x(\chi_{U_x}a_x)
		=\sum_{U_x}\{\chi_{U_x}a_x,g_{U_x}\}_x\,.
	\end{align*}
	The latter identity and the $(\beta,\delta_\infty^\Gamma)$-KMS condition leads to
	\begin{align}\label{Eq: classical KMS condition and rotation - infinitesimal version}
		\omega_\infty^{\beta,\Gamma}(\hat{D}_xa_\Lambda)
		=-\beta\omega_\infty^{\beta,\Gamma}\Big(
		 a_\Lambda\sum_{U_x}\chi_{U_x}
		\sum_{X\ni x}\{g_{U_x},\varphi_X\}
		\Big)
		=\beta\omega_\infty^{\beta,\Gamma}\Big[\Big(\sum_{X\ni x}\hat{D}_x\varphi_X\Big) a_\Lambda\Big]\,.
	\end{align}
	Integration of the latter identity leads to Equation \eqref{Eq: classical KMS condition and rotations}.
	In more details we first observe that
	\begin{align*}
		\sum_{X\ni x}(I-\hat{R}_x)\varphi_X
		=\lim_{Y\uparrow\Gamma}(I-\hat{R}_x)h_Y\,,
		\qquad
		h_Y=\sum_{X\subset Y}\varphi_X
		\in\dot{B}_\infty^Y\,,
	\end{align*}
	where the limit converges in the $B_\infty^\Gamma$-norm.
	We then consider the function
	\begin{align*}
		\omega_Y(t):=\omega_\infty^{\beta,\Gamma}\bigg[
		e^{\beta\sum_{\substack{X\subset Y\\X\ni x}}(I-\hat{R}_x(t))\varphi_X
		}\hat{R}_x(t) a_\Lambda
		\bigg]\,,
	\end{align*}
	where $\hat{R}_x(t)=\exp[-it\hat{D}_x]$: Notice that
	\begin{align*}
		\omega_\infty^{\beta,\Gamma}\Big(e^{\beta\sum_{X\ni x}(I-\hat{R}_x)\varphi_X}
		 \hat{R}_xa_\Lambda
		\Big)
		=\lim_{Y\uparrow\Gamma}\omega_Y(1)\,.
	\end{align*}
	By dominated convergence and on account of Equation \eqref{Eq: classical KMS condition and rotation - infinitesimal version} we also have
	\begin{align*}
		-i\dot{\omega}_Y(t)
		&=\omega_\infty^{\beta,\Gamma}\bigg[
		e^{\beta\sum_{\substack{X\subset Y\\X\ni x}}(I-\hat{R}_x(t))\varphi_X
		}\Big(
		 [\hat{D}_x\hat{R}_x(t)a_\Lambda]
		-\beta(\hat{R}_x(t)a_\Lambda)\sum_{\substack{X\subset Y\\X\ni x}}
		\hat{D}_x\hat{R}_x(t)\varphi_X
		\Big)
		\bigg]
		\\
		&=-\beta\omega_\infty^{\beta,\Gamma}\bigg[
		e^{\beta\sum_{\substack{X\subset Y\\X\ni x}}(I-\hat{R}_x(t))\varphi_X
		}
		 [\hat{R}_x(t)a_\Lambda]\sum_{\substack{X\subset Y\\X\ni x}}\hat{D}_x\varphi_X\bigg]
		\\
		&+\omega_\infty^{\beta,\Gamma}\bigg[\hat{D}_x\bigg(
		e^{\beta\sum_{\substack{X\subset Y\\X\ni x}}(I-\hat{R}_x(t))\varphi_X}
		 [\hat{R}_x(t)a_\Lambda]
		\bigg)
		\bigg]
		\\
		&=-\beta\omega_\infty^{\beta,\Gamma}\bigg[
		e^{\beta\sum_{\substack{X\subset Y\\X\ni x}}(I-\hat{R}_x(t))\varphi_X
		}
		 [\hat{R}_x(t)a_\Lambda]\sum_{\substack{X\subset Y\\X\ni x}}\hat{D}_x\varphi_X
		\bigg]
		\\
		&+\beta\omega_\infty^{\beta,\Gamma}\bigg[
		e^{\beta\sum_{\substack{X\subset Y\\X\ni x}}(I-\hat{R}_x(t))\varphi_X}
		 [\hat{R}_x(t)a_\Lambda]
		\sum_{X\ni x} \hat{D}_x\varphi_X
		\bigg]
		\\
		&=\beta\omega_\infty^{\beta,\Gamma}\bigg[
		e^{\beta\sum_{\substack{X\subset Y\\X\ni x}}(I-\hat{R}_x(t))\varphi_X}
		 [\hat{R}_x(t)a_\Lambda]
		\sum_{\substack{X\cap Y^c\neq\varnothing\\X\ni x}} \hat{D}_x\varphi_X
		\bigg]\,.
	\end{align*}
	By direct inspection we also find
	\begin{align*}
		\lim_{Y\uparrow\Gamma}\sup_{t\in[0,1]}|\dot{\omega}_Y(t)|=0\,,
	\end{align*}
	thus, we have
	\begin{align*}
		\omega_\infty^{\beta,\Gamma}\Big(e^{\beta\sum_{X\ni x}(I-\hat{R}_x)\varphi_X}
		 \hat{R}_xa_\Lambda
		\Big)
		=\lim_{Y\uparrow\Gamma}\omega_Y(1)
		=\lim_{Y\uparrow\Gamma}\omega_Y(0)
		+\lim_{Y\uparrow\Gamma}\int_0^1\dot{\omega}_Y(t)\mathrm{d}t
		=\omega_\infty^{\beta,\Gamma}( a_\Lambda)\,.
	\end{align*}
\end{proof}

\begin{remark}\label{Rmk: uniqueness result for classical KMS states - beta=0 case}
	Lemma \ref{Lem: classical KMS condition and rotations} suffices to prove Theorem \ref{Thm: uniqueness result for classical KMS states} for the particular case $\beta=0$, where the assumption \eqref{Eq: uniqueness result for classical KMS states - assumption on potential} on the potential $\{\varphi_\Lambda\}_{\Lambda\Subset\Gamma}$ is not needed.
	In this case the unique $(0,\delta_\infty^\Gamma)$-KMS classical states $\omega_\infty^{0,\Gamma}$ coincides with the normalized Poisson trace
	\begin{align}\label{Eq: unique Poisson trace classical state}
		\omega_\infty^{0,\Gamma}( a_\Lambda)
		=\int_{\mathbb{S}^2_\Lambda}a_\Lambda(\sigma_\Lambda)\mathrm{d}\mu_0^\Lambda(\sigma_\Lambda)\,.
	\end{align}
	Although this can be shown with several methods, it is instructive to prove it with the help of Lemma \ref{Lem: classical KMS condition and rotations}: This leads to a first intuition on the strategy we will employ in the proof of Theorem \ref{Thm: uniqueness result for classical KMS states}.
	
	Let $\omega_\infty^{0,\Gamma}\in S(B_\infty^\Gamma)$ be a $(0,\delta_\infty^\Gamma)$-KMS classical states: We wish to prove that the associated element $\underline{\omega}_\infty^{0,\Gamma}\in\underline{\mathsf{X}}$ is necessarily of the form
	\begin{align}\label{Eq: unique Poisson trace classical state - Banach space function}
		\underline{\omega}_\infty^{0,\Gamma}(\ell_\Lambda,m_\Lambda)
		=\begin{dcases}
			1&\Lambda=\varnothing
			\\
			0&\Lambda\neq\varnothing
		\end{dcases}
		\,.
	\end{align}
	This implies that $\omega_\infty^{0,\Gamma}$ abides by Equation \eqref{Eq: unique Poisson trace classical state} on account of the Fourier-Laplace expansion \eqref{Eq: Fourier-Laplace expansion - finite region}.
	To prove Equation \eqref{Eq: unique Poisson trace classical state - Banach space function} it suffices to observe that Lemma \ref{Lem: classical KMS condition and rotations} implies $\omega_\infty^{0,\Gamma}\circ\hat{R}_x=\omega_\infty^{0,\Gamma}$.
	
	Let $\Lambda\Subset\Gamma$, $\ell_\Lambda\in\mathbb{N}^\Lambda$, $m_\Lambda\in\mathbb{Z}^\Lambda$, $m_x\in[-\ell_x,\ell_x]\cap\mathbb{Z}$ for all $x\in\Lambda$.
	For a fixed $x\in\Lambda$ we compute
	\begin{align*}
		\underline{\omega}_\infty^{0,\Gamma}(\ell_\Lambda,m_\Lambda)
		=\omega_\infty^{0,\Gamma}( Y_{\ell_\Lambda,m_\Lambda})
		=\int_{SU(2)}\omega_\infty^{0,\Gamma}(\hat{R}_x Y_{\ell_\Lambda,m_\Lambda})\mathrm{d}R_x
		=\omega_\infty^{0,\Gamma}\bigg(\int_{SU(2)}\hat{R}_x Y_{\ell_\Lambda,m_\Lambda}\mathrm{d}R_x\bigg)\,.
	\end{align*}
	where $\mathrm{d}R_x$ denotes the normalized Haar measure on $SU(2)$.
	At this stage we observe that the irreducibility of the left-representation of $SU(2)$ on the space generated by $\{Y_{\ell_x,m_x}\}_{m_x}$ entails
	\begin{align}\label{Eq: spherical harmonics - consequence of Shur Lemma}
		\int_{SU(2)}\hat{R}_xY_{\ell_x,m_x}\mathrm{d}R_x
		=\delta^0_{\ell_x}\delta^0_{m_x}\,.
	\end{align}
	This implies
	\begin{align*}
		\int_{SU(2)}\hat{R}_x Y_{\ell_\Lambda,m_\Lambda}\mathrm{d}R_x
		= \int_{SU(2)}\hat{R}_xY_{\ell_x,m_x}\mathrm{d}R_x
		\otimes Y_{\ell_{\Lambda\setminus\{x\}},m_{\Lambda\setminus\{x\}}}
		=\delta_{\ell_x}^0\delta_{m_x}^0 Y_{\ell_{\Lambda\setminus\{x\}},m_{\Lambda\setminus\{x\}}}
		=0\,,
	\end{align*}
	where in the second equality we used that $\ell_x\in\mathbb{N}$.
\end{remark}

\begin{proof}[Proof of Theorem \ref{Thm: uniqueness result for classical KMS states}]
	Let $\omega_\infty^{\beta,\Gamma}\in S(B_\infty^\Gamma)$ be a $(\beta,\delta_\infty^\Gamma)$-KMS classical states.
	We will consider the associated element $\underline{\omega}_\infty^{\beta,\Gamma}\in\underline{\mathsf{X}}$ as described in Remark \ref{Rmk: Banach space for state - main definitions}: Our goal is to prove that $\underline{\omega}_\infty^{\beta,\Gamma}$ is the solution to a linear equation in $\underline{\mathsf{X}}$ which is unique under assumption \eqref{Eq: uniqueness result for classical KMS states - assumption on potential}.
	
	To begin with, we choose an arbitrary but fixed bijection $\Gamma\simeq\mathbb{Z}$ and induce an ordering on $\Gamma$ based on such map.
	Thus, for any $\Lambda\Subset\Gamma$ we may set $x:=\min_{y\in\Lambda}y$ where the minimum is taken with respect to the chosen ordering.
        The choice of the ordering is only made to select a distinguished point $x\in\Lambda$ for any $\Lambda\Subset\Gamma$.
    
	Let now $\ell_\Lambda\in\mathbb{N}^\Lambda$ and $m_\Lambda\in\mathbb{Z}^\Lambda$ be such that $m_y\in[-\ell_y,\ell_y]\cap\mathbb{Z}$ for all $y\in\Lambda$.
	Proceeding as in Remark \ref{Rmk: uniqueness result for classical KMS states - beta=0 case} we compute
	\begin{align*}
		\underline{\omega}_\infty^{\beta,\Gamma}(\ell_\Lambda,m_\Lambda)
		&=\omega_\infty^{\beta,\Gamma}( Y_{\ell_\Lambda,m_\Lambda})
		\\
		&=\omega_\infty^{\beta,\Gamma}\bigg(\int_{SU(2)}(I-\hat{R}_x) Y_{\ell_\Lambda,m_\Lambda}\mathrm{d}R_x\bigg)
		\\
		&=\omega_\infty^{\beta,\Gamma}\bigg(\int_{SU(2)}
		\Big(I-e^{\beta\sum_{X\ni x}(I-\hat{R}_x)\varphi_X}\Big)
		 Y_{\ell_\Lambda,m_\Lambda}\mathrm{d}R_x\bigg)\,,
	\end{align*}
	where in the second equality we used that $R_x\mapsto\hat{R}_x$ is a unitary irreducible representation when restricted on the vector space generated by $\{Y_{\ell_x,m_x}\}_{m_x}$ while $\ell_x\in\mathbb{N}$, \textit{cf.} Remark \ref{Rmk: uniqueness result for classical KMS states - beta=0 case}.
	The third equality is nothing but Equation \eqref{Eq: classical KMS condition and rotations}.
	The exponential in the last term can be expanded in a series converging in $B_\infty^\Gamma$, thus,
	\begin{multline}\label{Eq: classical uniqueness - expansion of the state}
		\underline{\omega}_\infty^{\beta,\Gamma}(\ell_\Lambda,m_\Lambda)
		=-\sum_{n\geq 1}\frac{\beta^n}{n!}
		\sum_{\substack{X_1,\ldots,X_n\\x\in X_1\cap\ldots\cap X_n}}
		\omega_\infty^{\beta,\Gamma}\bigg(
		\int_{SU(2)}
		\prod_{k=1}^n(I-\hat{R}_x)\varphi_{X_k}
		 Y_{\ell_\Lambda,m_\Lambda}
		\mathrm{d}R_x
		\bigg)
		\\
		=-\sum_{n\geq 1}\frac{\beta^n}{n!}
		\sum_{\substack{X_1,\ldots,X_n\\x\in X_1\cap\ldots\cap X_n}}
		\sum_{\substack{\ell_{X_1},\ldots,\ell_{X_n}\\m_{X_1},\ldots,m_{X_n}}}
		\int_{SU(2)}
		\prod_{k=1}^n
		C_{X_k,R_x}(\ell_{X_k},m_{X_k})
		\mathrm{d}R_x
		\\
		\omega_\infty^{\beta,\Gamma}\Big(
		Y_{\ell_{X_1},m_{X_1}}
		\cdots Y_{\ell_{X_n},m_{X_n}}
		 Y_{\ell_\Lambda,m_\Lambda}
		\Big)
		\,,
	\end{multline}
	where we used the Fourier-Laplace expansion discussed in Remark \ref{Rmk: spherical harmonics - set-up, identities and useful bounds} and set
	\begin{multline*}
		C_{X_k,R_x}(\ell_{X_k},m_{X_k})
		:=\Big(\prod_{y\in X_k}(2\ell_{X_k,y}+1)\Big)
		\big\langle Y_{\ell_{X_k},m_{X_k}}\big|(I-\hat{R}_x)\varphi_{X_k}
		\big\rangle_{L^2(\mathbb{S}^2_{X_k},\mu_0^{X_k})}
		\\
		=\bigg(\prod_{y\in X_k}\frac{(2\ell_{X_k,y}+1)}{[1+\ell_{X_k,y}(\ell_{X_k,y}+1)]^s}\bigg)
		\big\langle Y_{\ell_{X_k},m_{X_k}}\big|(I-\hat{R}_x)(I-\Delta_{\mathbb{S}^2,X_k})^s\varphi_{X_k}
		\big\rangle_{L^2(\mathbb{S}^2_{X_k},\mu_0^{X_k})}\,,
	\end{multline*}
	where $\Delta_{\mathbb{S}^2,X_k}:=\bigotimes_{y\in X_k}\Delta_{\mathbb{S}^2,y}$ denotes the tensor product of the Laplacians $\Delta_{\mathbb{S}^2,y}$ acting on $\mathbb{S}^2_y$.
	
	At this stage it is important to carefully analyse the product of the spherical harmonics appearing in Equation \eqref{Eq: classical uniqueness - expansion of the state}:
	\begin{multline*}
		Y_{\ell_{X_1},m_{X_1}}
		\cdots Y_{\ell_{X_n},m_{X_n}}
		 Y_{\ell_\Lambda,m_\Lambda}
		\\
		=\Big(\prod_{y\in\Lambda\cap S_n^c} Y_{\ell_{\Lambda,y},m_{\Lambda,y}}\Big)
		\prod_{y\in S_n}
		Y_{\tilde{\ell}_{\Lambda,y},\tilde{m}_{\Lambda,y}}
		Y_{\tilde{\ell}_{X_1,y},\tilde{m}_{X_1,y}}
		\cdots
		Y_{\tilde{\ell}_{X_n,y},\tilde{m}_{X_n,y}}\,,
	\end{multline*}
	where we set $\boxed{S_n}:=X_1\cup\ldots\cup X_n$ while $\tilde{\ell}_{X_k}\in\mathbb{Z}_+^{S_n}$ denotes the extension of $\ell_{X_k}$ obtained by setting $\tilde{\ell}_{X_k,y}=0$ for $y\notin X_k$ ---$\tilde{m}_{X_k}$ is defined similarly.
	In particular, by an iterated use of Equation \eqref{Eq: spherical harmonics - expansion of pointwise product} we find, for all $y\in S_n$,
	\begin{align*}
		&Y_{\tilde{\ell}_{\Lambda,y},\tilde{m}_{\Lambda,y}}
		\prod_{k=1}^n
		Y_{\tilde{\ell}_{X_k,y},\tilde{m}_{X_k,y}}
		\\
		&=\sum_{s_{y,1}=|\tilde{\ell}_{\Lambda,y}-\tilde{\ell}_{X_1,y}|}^{\tilde{\ell}_{\Lambda,y}+\tilde{\ell}_{X_1,y}}
		c_{s_{y,1}}
		Y_{s_{y,1},\tilde{m}_{\Lambda,y}+\tilde{m}_{X_1,y}}
		\prod_{k=2}^n
		Y_{\tilde{\ell}_{X_k,y},\tilde{m}_{X_k,y}}
		 \\
		&=\sum_{\substack{s_{y,1},\ldots,s_{y,n}\\|s_{y,k-1}-\tilde{\ell}_{X_k,y}|\leq s_{y,k}\leq|s_{y,k-1}+\tilde{\ell}_{X_k,y}|}}
		\Big(
		\prod_{k=1}^n c_{s_{y,k}}
		\Big)
		Y_{s_{y,n},\tilde{m}_{\Lambda,y}+\tilde{m}_{X_1,y}+\ldots+\tilde{m}_{X_n,y}}\,,
	\end{align*}
	where $c_{s_{y,k}}$, $k\in\{1,\ldots,n\}$ are defined in Equation \eqref{Eq: spherical harmonics - expansion of pointwise product} ---we omitted the $m$-dependence since it will not play any role.
	The particular values of $c_{s_{y,k}}$ are not important, however, we crucially observe that $|c_{s_{y,k}}|\leq 1$.
	For later convenience we also observe that
	\begin{align*}
		N(\ell_{X_1},\ldots,\ell_{X_n})
		&:=\prod_{y\in S_n}
		\sum_{s_{y,1}=|\tilde{\ell}_{\Lambda,y}-\tilde{\ell}_{X_1,y}|}^{\tilde{\ell}_{\Lambda,y}+\tilde{\ell}_{X_1,y}}
		\ldots
		\sum_{s_{y,n}=|s_{n-1,y}-\tilde{\ell}_{X_n,y}|}^{s_{n-1,y}+\tilde{\ell}_{X_n,y}}
		\\
		&\leq\prod_{y\in S_n}\prod_{k=1}^n(2\tilde{\ell}_{X_k,y}+1)
		=\prod_{k=1}^n\prod_{y\in X_k}(2\ell_{X_k,y}+1)\,,
	\end{align*}
	because $\sum_{s=|a-b|}^{a+b}=2\min\{a,b\}+1$.
	This implies that, once considering the product over $y\in S_n$ and expanding the resulting sum, the product of spherical harmonics considered above can be written as a sum of at most $N(\ell_{X_1},\ldots,\ell_{X_n})$ terms of the form $Y_{\ell_{S_n}^k,m_{S_n}^k}$ where $\ell_{S_n}^k,m_{S_n}^k$, $k=1,\ldots,N(\ell_{X_1},\ldots,\ell_{X_n})$, are built out of $\ell_{\Lambda\cap S_n},\ell_{X_1},\ldots\ell_{X_n}$.
	Explicitly we have
	\begin{align*}
		&Y_{\ell_{X_1},m_{X_1}}
		\cdots Y_{\ell_{X_n},m_{X_n}}
		Y_{\ell_\Lambda,m_\Lambda}
		\\
		&=\Big(\prod_{y\in\Lambda\cap S_n^c} Y_{\ell_{\Lambda,y},m_{\Lambda,y}}\Big)
		\prod_{y\in S_n}
		\sum_{s_{y,1},\ldots,s_{y,n}}
		\Big(
		\prod_{k=1}^n c_{s_{y,k}}
		\Big)
		Y_{s_{y,n},\tilde{m}_{\Lambda,y}+\tilde{m}_{X_1,y}+\ldots+\tilde{m}_{X_n,y}}
		\\
		&=\Big(\prod_{y\in\Lambda\cap S_n^c} Y_{\ell_{\Lambda,y},m_{\Lambda,y}}\Big)
		\sum_{k=1}^{N(\ell_{X_1},\ldots,\ell_{X_n})}
		C'(\ell_{S_n}^k,m_{S_n}^k)
		Y_{\ell_{S_n}^k,m_{S_n}^k}\,,
	\end{align*}
	where the explicit expression of the coefficients $C'(\ell_{S_n}^k,m_{S_n}^k)$ will not matter in the forthcoming discussion, however, it will be important to observe that $|C'(\ell_{S_n}^k,m_{S_n}^k)|\leq 1$.
	Summing up, any $(\beta,\delta_\infty^\Gamma)$-KMS $\omega_\infty^{\beta,\Gamma}\in S(B_\infty^\Gamma)$ fulfils
	\begin{multline*}
		\underline{\omega}_\infty^{\beta,\Gamma}(\ell_\Lambda,m_\Lambda)
		=-\sum_{n\geq 1}\frac{\beta^n}{n!}
		\sum_{\substack{X_1,\ldots,X_n\\x\in X_1\cap\ldots\cap X_n}}
		\sum_{\substack{\ell_{X_1},\ldots,\ell_{X_n}\\m_{X_1},\ldots,m_{X_n}}}
		\int_{SU(2)}
		\prod_{k=1}^n
		C_{X_k,R_x}(\ell_{X_k},m_{X_k})\mathrm{d}R_x
		\\
		\sum_{k=1}^{N(\ell_{X_1},\ldots,\ell_{X_n})}
		C'(\ell_{S_n}^k,m_{S_n}^k)
		\omega_\infty^{\beta,\Gamma}\Big(
		Y_{\ell_{S_n}^k,m_{S_n}^k}
		\prod_{y\in\Lambda\cap S_n^c} Y_{\ell_{\Lambda,y},m_{\Lambda,y}}
		\Big)\,.
	\end{multline*}
	Let $\boxed{X(n,\Lambda)}:=\Lambda\cup S_n$ and set $\ell_{X(n,\Lambda)}^k:=\ell_{S_n}^k\ell_{\Lambda\cap S_n^c}$ (\textit{i.e.} $\ell_{X(n,\Lambda),y}^k=\ell_{\Lambda,y}$ if $y\in\Lambda\cap S_n^c$ and $\ell_{X(n,\Lambda),y}^k=\ell_{S_n,y}^k$ if $y\in S_n$) and similarly $m_{X(n,\Lambda)}=m_{S_n}m_{\Lambda\cap S_n^c}$.
	Then the above equality can be written as a linear equation in $\underline{\mathsf{X}}$, in particular
	\begin{align}\label{Eq: Banach space for state - linear equation for classical uniqueness}
		(I-L_\infty^\beta)\underline{\omega}_\infty^{\beta,\Gamma}=\underline{\delta}\,.
	\end{align}
	Here $\underline{\delta}\in\underline{\mathsf{X}}$ and $L_\infty^\beta\in\mathcal{B}(\underline{\mathsf{X}})$ are defined by
	\begin{align}\label{Eq: Banach space for state - non-homogeneous term}
		\underline{\delta}_\Lambda(\ell_\Lambda,m_\Lambda)
		:=\begin{dcases}
			1&\Lambda=\varnothing
			\\
			0&\Lambda\neq\varnothing
		\end{dcases}\,,
	\end{align}
	while for all $\underline{f}\in\underline{\mathsf{X}}$ we set $(L_\infty^\beta\underline{f})_\varnothing=0$ and
	\begin{multline}\label{Eq: Banach space for state - linear operator for classical uniqueness}
		(L_\infty^\beta\underline{f})_\Lambda(\ell_\Lambda,m_\Lambda)
		:=-\sum_{n\geq 1}\frac{\beta^n}{n!}
		\sum_{\substack{X_1,\ldots,X_n\\x\in X_1\cap\ldots\cap X_n}}
		\sum_{\substack{\ell_{X_1},\ldots,\ell_{X_n}\\m_{X_1},\ldots,m_{X_n}}}
		\int_{SU(2)}
		\prod_{k=1}^n
		C_{X_k,R_x}(\ell_{X_k},m_{X_k})\mathrm{d}R_x
		\\
		\sum_{k=1}^{N(\ell_{X_1},\ldots,\ell_{X_n})}
		C'(\ell_{S_n}^k,m_{S_n}^k)
		f_{X(n,\Lambda)}(\ell_{X(n,\Lambda)}^k,m_{X(n,\Lambda)}^k)\,,
	\end{multline}
	for all non-empty $\Lambda\Subset\Gamma$ ---we recall that we set $x:=\min_{y\in\Lambda}y$.
	
	At this stage we may bound $\|L_\infty^\beta\|_{\mathcal{B}(\underline{\mathsf{X}})}$ in such a way that $\|L_\infty^\beta\|_{\mathcal{B}(\underline{\mathsf{X}})}<1$ if $\beta$ is small enough.
	This will ensure that \eqref{Eq: Banach space for state - linear equation for classical uniqueness} has a unique solution $\underline{\omega}_\infty^{\beta,\Gamma}\in\underline{\mathsf{X}}$, therefore, its associated state $\omega_\infty^{\beta,\Gamma}$ will be the unique $(\beta,\delta_\infty^\Gamma)$-KMS classical states on $B_\infty^\Gamma$.
	(In passing, the forthcoming estimates will also prove that $L_\infty^\beta$ is bounded on $\underline{\mathsf{X}}$.)
	To this avail we observe that,
	\begin{multline*}
		\sup_{\ell_\Lambda,m_\Lambda}|(L_\infty^\beta\underline{f})_\Lambda(\ell_\Lambda,m_\Lambda)|
		\\
		\leq\|\underline{f}\|_{\underline{\mathsf{X}}}
		\sum_{n\geq 1}\frac{\beta^n}{n!}
		\sum_{\substack{X_1,\ldots,X_n\\x\in X_1\cap\ldots\cap X_n}}
		\sum_{\substack{\ell_{X_1},\ldots,\ell_{X_n}\\m_{X_1},\ldots,m_{X_n}}}
		\int_{SU(2)}
		\prod_{k=1}^n
		|C_{X_k,R_x}(\ell_{X_k},m_{X_k})|\prod_{y\in X_k}(2\ell_{X_k,y}+1)\mathrm{d}R_x\,,
	\end{multline*}
	where we used the bound on $|C'(\ell_{S_n}^k,m_{S_n}^k)|$ and on $N(\ell_{X_1},\ldots,\ell_{X_n})$.
	Moreover, proceeding as in Remark \ref{Rmk: spherical harmonics - set-up, identities and useful bounds} we have
	\ndnote{$2$-question here}
	\begin{align*}
		&\sum_{\substack{\ell_{X_1},\ldots,\ell_{X_n}\\m_{X_1},\ldots,m_{X_n}}}
		\int_{SU(2)}
		\prod_{k=1}^n
		|C_{X_k,R_x}(\ell_{X_k},m_{X_k})|\prod_{y\in X_k}(2\ell_{X_k,y}+1)\mathrm{d}R_x
		\\
		&\leq\sum_{\ell_{X_1},\ldots,\ell_{X_n}}
		2^n\prod_{k=1}^n
		\bigg(\prod_{y\in X_k}
		\frac{(\ell_{X_k,y}+1)^{5/2}}{[1+\ell_{X_k,y}(\ell_{X_k,y}+1)]^s}
		\bigg)
		(C_\Delta^s)^{|X_k|}\|\varphi_{X_k}\|_{C^{2s}(\mathbb{S}^2_{X_k})}
		\\
		&=2^n\prod_{k=1}^n
		(C_\Delta^sK_s)^{|X_k|}\|\varphi_{X_k}\|_{C^{2s}(\mathbb{S}^2_{X_k})}
		\,,
	\end{align*}
	where $K_s$ has been defined in Equation \eqref{Eq: uniqueness result for classical KMS states - useful constants}.
	It follows that
	\begin{align*}
		\sup_{\ell_\Lambda,m_\Lambda}|(L_\infty^\beta\underline{f})_\Lambda(\ell_\Lambda,m_\Lambda)|
		&\leq\|\underline{f}\|_{\underline{\mathsf{X}}}
		\sum_{n\geq 1}\frac{(2\beta)^n}{n!}
		\sum_{\substack{X_1,\ldots,X_n\\x\in X_1\cap\ldots\cap X_n}}
		\prod_{k=1}^n
		(C_\Delta^sK_s)^{|X_k|}\|\varphi_{X_k}\|_{C^{2s}(\mathbb{S}^2_{X_k})}
		\\
		&=\|\underline{f}\|_{\underline{\mathsf{X}}}
		\sum_{n\geq 1}\frac{1}{n!}
		\bigg(
		2\beta\sum_{\substack{X\Subset\Gamma\\x\in X}}
		(C_\Delta^sK_s)^{|X|}
		\|\varphi_X\|_{C^{2s}(\mathbb{S}^2_X)}\bigg)^n\,.
	\end{align*}
	Finally, we have
	\begin{align*}
		\sum_{\substack{X\Subset\Gamma\\x\in X}}
		(C_\Delta^sK_s)^{|X|}
		\|\varphi_X\|_{C^{2s}(\mathbb{S}^2_X)}
		&=\sum_{m\geq 0}(C_\Delta^sK_s)^{m+1}
		\sum_{\substack{|X|=m+1\\x\in X}}
		\|\varphi_X\|_{C^{2s}(\mathbb{S}^2_X)}
		\\
		&\leq\sum_{m\geq 0}(C_\Delta^sK_s)^{m+1}
		\sup_{x\in\Gamma}\sum_{\substack{|X|=m+1\\x\in X}}
		\|\varphi_X\|_{C^{2s}(\mathbb{S}^2_X)}
		=C_\Delta^sK_s\|\varphi\|_{0,s}\,,
	\end{align*}
	where $\|\varphi\|_{0,s}$ has been defined in Equation \eqref{Eq: uniqueness result for classical KMS states - assumption on potential}.
	Thus, $L_\infty^\beta\in\mathcal{B}(\underline{\mathsf{X}})$ with
	\begin{align*}
		\|L_\infty^\beta\|_{\mathcal{B}(\underline{\mathsf{X}})}
		\leq \exp\Big[2C_\Delta^sK_s\beta\|\varphi\|_{0,s}\Big]-1\,,
	\end{align*}
	which implies $\|L_\infty^\beta\|_{\mathcal{B}(\underline{\mathsf{X}})}<1$ provided $\beta<\beta_{0,s}$, where $\beta_{0,s}$ has been defined in Equation \eqref{Eq: uniqueness result for classical KMS states - estimate on critical temperature}.
\end{proof}

\subsection{Uniqueness result for quantum KMS state}
\label{Subsec: uniqueness result for quantum KMS state}

The goal of this section is to prove a uniqueness result for $(\beta,\delta_j^\Gamma)$-KMS quantum states on $B_j^\Gamma$, \textit{cf.} Theorem \ref{Thm: uniqueness result for quantum KMS states}
The latter applies under hypothesis very similar to those of Theorem \ref{Thm: uniqueness result for classical KMS states}.
In fact, Theorems \ref{Thm: uniqueness result for classical KMS states}, \ref{Thm: uniqueness result for quantum KMS states} will imply that, under suitably mild assumptions, for high enough temperatures there is absence of both classical and quantum phase transitions, \textit{cf.} Remark \ref{Rmk: NO CPT implies NO QPT; comparison with BR result}.
The proof of Theorem \ref{Thm: uniqueness result for quantum KMS states} is inspired by \cite[Prop. 6.2.45]{Bratteli_Robinson_97}, see also \cite{Frohlich_Ueltschi_2015}, although it requires a different argument to ensure a uniform bound on $j\in\mathbb{Z}_+/2$.

\begin{theorem}\label{Thm: uniqueness result for quantum KMS states}
	Let $\varphi:=\{\varphi_\Lambda\}_{\Lambda\Subset\Gamma}$ with $\varphi\in C^{2s}(\mathbb{S}^2_\Lambda)$, $s>7/4$.
	Let assume that there exists $\varepsilon>0$ such that
	\begin{align}\label{Eq: uniqueness result for quantum KMS states - assumption on potential}
		\|\varphi\|_{\varepsilon,s}
		:=\sum_{m\geq 0}(e^{\varepsilon}K_sC_\Delta^s)^m
		\sup_{y\in\Gamma}
		\sum_{\substack{|X|=m+1\\X\ni y}}
		\|\varphi_X\|_{C^{2s}(\mathbb{S}^2_X)}
		<+\infty\,,
	\end{align}
	where $C_s>1$ and $C_\Delta$ have been defined in Equation \eqref{Eq: uniqueness result for classical KMS states - useful constants}.
	Then there exists a unique $(\beta,\delta_j^\Gamma)$-KMS quantum states on $B_j^\Gamma$ for all $\beta\in[0,\beta_{\varepsilon,s})$ where
	\begin{align}\label{Eq: uniqueness result for quantum KMS states - estimate on critical temperature}
		\beta_{\varepsilon,s}
		:=\frac{\varepsilon}{1+e^\varepsilon}\frac{1}{2K_sC_\Delta^s\|\varphi\|_{\varepsilon,s}}\,.
	\end{align}
\end{theorem}

The proof of Theorem \ref{Thm: uniqueness result for quantum KMS states} is similar in spirit to the one of Theorem \ref{Thm: uniqueness result for classical KMS states}.
As such, it requires a few technical observations which we will recollect in the following remark.

\begin{remark}\label{Rmk: spherical harmonics - Hilbert-Schmidt orthornormal basis}
	\noindent
	\begin{enumerate}[(i)]
		\item
		Recalling Remark \ref{Rmk: Berezin quantization useful remark}, the set $\{Q_j(Y_{\ell,m})\,|\,\ell\in\mathbb{Z}_+\,,\,m\in[-\ell,\ell]\cap\mathbb{Z}\}$ is an orthogonal basis of $B_j=M_{2j+1}(\mathbb{C})$ with respect to the Hilbert-Schmidt scalar product \eqref{Eq: Hilbert-Schmidt scalar product} ---notice that $Q_j(Y_{\ell,m})=0$ if $\ell>2j$, \textit{cf.} Equation \eqref{Eq: Berezin SDQ - quantization of spherical harmonics}.
		In what follows we will normalize $Q_j(Y_{\ell,m})$ by setting
		\begin{align}\label{Eq: Berezin SDQ - HS normalized quantization of spherical harmonics}
			\mathcal{Y}_{j|\ell,m}:=\frac{1}{\sqrt{c_{j,\ell}}}Q_j(Y_{\ell,m})\,,
		\end{align}
		where $c_{j,\ell}>0$ has been defined in Equation \eqref{Eq: Berezin SDQ - check of spherical harmonics}: With this choice we find
		\begin{align}\label{Eq: Berezin SDQ - HS normalized quantization of spherical harmonics - properties}
			\|\mathcal{Y}_{j|\ell,m}\|_{\textsc{hs}}
			=\frac{1}{\sqrt{2\ell+1}}\,,
			\qquad
			\|\mathcal{Y}_{j|\ell,m}\|_{B_j}
			\leq 1\,.
		\end{align}
		Indeed by direct inspection we have
		\begin{align*}
			\|\mathcal{Y}_{j|\ell,m}\|_{\textsc{hs}}^2
			=\frac{1}{c_{j,\ell}}\langle Q_j(Y_{\ell,m})|Q_j(Y_{\ell,m})\rangle_{\textsc{hs}}
			=\frac{1}{c_{j,\ell}}\langle Y_{\ell,m}|\check{Y}_{\ell,m}\rangle_{L^2(\mathbb{S}^2,\mu_0)}
			=\frac{1}{2\ell+1}\,,
		\end{align*}
		where we used Equations \eqref{Eq: Berezin SDQ - scalar product intertwining property} and \eqref{Eq: spherical harmonics - convention}.
		Moreover, for all $\psi\in\mathbb{C}^{2j+1}$,
		\begin{align*}
			\|\mathcal{Y}_{j|\ell,m}\psi\|^2
			&=\frac{1}{c_{j,\ell}}
			\sum_{m'=-j}^j
			(\textsc{cg}_{\ell,m;j,m'}^{j,m+m'})^2
			(\textsc{cg}_{\ell,0;j,j}^{j,j})^2
			|\langle j,m'|\psi\rangle|^2
			\\
			&=\sum_{m'=-j}^j
			(\textsc{cg}_{\ell,m;j,m'}^{j,m+m'})^2
			|\langle j,m'|\psi\rangle|^2
			\leq\|\psi\|^2\,,
		\end{align*}
		where we used the explicit expression obtained in Example \ref{Ex: Berezin SDQ - quantization of spherical harmonics}.
        It is worth to mention that further properties of the matrices $\mathcal{Y}_{j|\ell,m}$ have been studied in \cite{Hoppe_1989}.
        
		For any finite region $\Lambda\Subset\Gamma$ we define $\boxed{\mathcal{Y}_{j|\ell_\Lambda,m_\Lambda}}\in B_j^\Lambda$ by setting
		\begin{align*}
			\mathcal{Y}_{j|\ell_\Lambda,m_\Lambda}
			:=\bigotimes_{x\in\Lambda}\mathcal{Y}_{j|\ell_x,m_x}
			\in B_j^\Lambda\subset B_j^\Gamma\,,
		\end{align*}
		where $\ell_\Lambda\in\mathbb{Z}_+^\Lambda$ and $m_\Lambda\in\mathbb{Z}^\Lambda$ are such that $m_x\in[-\ell_x,\ell_x]$ for all $x\in\Lambda$.
		Then $\{\mathcal{Y}_{j|\ell_\Lambda,m_\Lambda}\,|\,\ell_\Lambda\in\mathbb{Z}_+^\Lambda\,,\,m_\Lambda\in\mathbb{Z}^\Lambda\,,\,m_x\in[-\ell_x,\ell_x]\;\forall x\in\Lambda\}$ is a orthogonal basis of $B_j^\Lambda$ with respect to the Hilbert-Schmidt scalar product \eqref{Eq: Hilbert-Schmidt scalar product}.
		In particular, if $A_\Lambda\in B_j^\Lambda$ then
		\begin{align}\label{Eq: Hilbert-Schimdt expansion - finite region}
			A_\Lambda=\sum_{\ell_\Lambda,m_\Lambda}
			\Big(\prod_{x\in\Lambda}(2\ell_x+1)\Big)
			\langle\mathcal{Y}_{j|\ell_\Lambda,m_\Lambda}|A_\Lambda\rangle_{\textsc{hs}}
			\mathcal{Y}_{j|\ell_\Lambda,m_\Lambda}\,,
		\end{align}
		see Equation \eqref{Eq: Fourier-Laplace expansion - finite region} for comparison.

		\item
		A crucial step in the proof of Theorem \ref{Thm: uniqueness result for classical KMS states} is the use of Equation \eqref{Eq: spherical harmonics - consequence of Shur Lemma}, \textit{cf.} Remark \ref{Rmk: uniqueness result for classical KMS states - beta=0 case}.
		Thanks to Equation \eqref{Eq: Berezin SDQ - left-action representation intertwining property} an analogous property holds in the quantum setting.
		Specifically, by proceeding as in Remark \ref{Rmk: uniqueness result for classical KMS states - beta=0 case} we find
		\begin{align}\label{Eq: spherical harmonics - consequence of Shur Lemma for HS orthonormal basis}
			\int_{SU(2)}\tilde{D}^{(j)}(R)\mathcal{Y}_{j|\ell,m}\mathrm{d}R
			=\frac{1}{\sqrt{c_{j,\ell}}}Q_j\bigg(
			\int_{SU(2)}\hat{R}Y_{\ell,m}\mathrm{d}R
			\bigg)
			=0\,.
		\end{align}
	
		\item
		Similarly to the classical case, \textit{cf.} Remark \ref{Rmk: Banach space for state - main definitions}, any state $\omega_j^\Gamma\in S(B_j^\Gamma)$ is uniquely determined by its associated element $\underline{\omega}_j^\Gamma\in\underline{\mathsf{X}}$ defined by
		\begin{align*}
			\underline{\omega}_j^\Gamma(\ell_\Lambda,m_\Lambda)
			:=\begin{cases}
				1&\Lambda=\varnothing
				\\
				\omega_j^\Gamma(\mathcal{Y}_{j|\ell_\Lambda,m_\Lambda})
				&\Lambda\neq\varnothing
			\end{cases}\,.
		\end{align*}
		Notice that $\underline{\omega}_j$ is an element of $\textsf{X}$, \textit{cf.} Remark \ref{Rmk: Banach space for state - main definitions}, because of the bound
		\begin{align*}
			|\underline{\omega}_j(\ell_\Lambda,m_\Lambda)|
			=|\omega_j(\mathcal{Y}_{j|\ell_\Lambda,m_\Lambda})|
			\leq\|\mathcal{Y}_{j|\ell_\Lambda,m_\Lambda}\|_{B_j^\Lambda}
			\leq 1\,,
		\end{align*}
		where we used the second inequality in \eqref{Eq: Berezin SDQ - HS normalized quantization of spherical harmonics - properties}.
		
		\item
		For later convenience we also discuss a quantum version of Equation \eqref{Eq: spherical harmonics - expansion of pointwise product}.
		Indeed, once again thanks to Equation \eqref{Eq: Berezin SDQ - scalar product intertwining property}, an analogous identity holds for the $\mathcal{Y}$'s.
		This can be either argued by observing that the $\mathcal{Y}_{j|\ell,m}$ are spherical tensors of order $\ell$ with respect to the representation $\tilde{D}^{(j)}$, \textit{cf.} \cite[\S 3.11]{Napolitano_Sakuraki_2021}, or by direct inspection.
		In more details, let $\ell_1,\ell_2\in\mathbb{Z}_+$ and $m_1,m_2\in\mathbb{Z}$ with $m_k\in[-\ell_k,\ell_k]$, $k\in\{1,2\}$.
		Equation \eqref{Eq: Hilbert-Schimdt expansion - finite region} leads to
		\begin{align*}
			\mathcal{Y}_{j|\ell_1,m_1}\mathcal{Y}_{j|\ell_2,m_2}
			=\sum_\ell(2\ell+1)
			\langle\mathcal{Y}_{j|\ell,m_1+m_2}|\mathcal{Y}_{j|\ell_1,m_1}\mathcal{Y}_{j|\ell_2,m_2}\rangle_{\textsc{hs}}
			\mathcal{Y}_{j|\ell,m_1+m_2}\,,
		\end{align*}
		where the restriction to $m=m_1+m_2$ is obtained by acting on both side of the equality with $\tilde{D}^{(j)}(e^{iJ_z})$ or by computing the coefficients directly.
		
		Equation \eqref{Eq: Berezin SDQ - product HS orthogonal quantization of spherical harmonics} can be seen as a quantum version of Equation \eqref{Eq: spherical harmonics - expansion of pointwise product}.
		Moreover, the coefficient appearing in Equation \eqref{Eq: Berezin SDQ - product HS orthogonal quantization of spherical harmonics} can be computed in a fairly explicit fashion.
		In particular we find
		\begin{align*}
			&(2\ell+1)
			|\langle\mathcal{Y}_{j|\ell,m_1+m_2}|\mathcal{Y}_{j|\ell_1,m_1}\mathcal{Y}_{j|\ell_2,m_2}\rangle_{\textsc{hs}}|
			\\
			&=\frac{2\ell+1}{2j+1}
			\bigg|\sum_{M,M_1,M_2=-j}^j
			\textsc{cg}_{\ell,m;j,M}^{j,M_1}
			\textsc{cg}_{\ell_1,m_1;j,M_2}^{j,M_1}
			\textsc{cg}_{\ell_2,m_2;j,M}^{j,M_2}\bigg|
			\\
			&=\sqrt{(2\ell+1)(2j+1)}
			\bigg|\textsc{cg}_{\ell_1,m_1;\ell_2,m_2}^{\ell,m}
			\begin{Bmatrix}
				j&j&\ell_2
				\\
				\ell&\ell_1&j
			\end{Bmatrix}\bigg|\,,
		\end{align*}
		where we used a few properties of the Clebsch-Gordan coefficients and the definition of the $6j$-symbols, \textit{cf.} \cite[\S 8.7 and \S 9.1]{Khersonskii_Moskalev_Varshalovich_1988}.
		Notice that the appearance of the Clebsch-Gordan coefficient $\textsc{cg}_{\ell_1,m_1;\ell_2,m_2}^{\ell,m}$ ensures that the coefficients vanish unless $|\ell_1-\ell_2|\leq\ell\leq\ell_1+\ell_2$.
		Moreover, since it is known that the matrix
		\begin{align*}
			C_{pq}:=\sqrt{2p+1}\sqrt{2q+1}
			\begin{Bmatrix}
				j_1&j_2&p
				\\
				\ell_1&\ell_2&q
			\end{Bmatrix}\,,
		\end{align*}
		is orthogonal, \textit{cf.} \cite[9.1.1]{Khersonskii_Moskalev_Varshalovich_1988},we have
		\begin{align*}
			(2\ell+1)
			|\langle\mathcal{Y}_{j|\ell,m_1+m_2}|\mathcal{Y}_{j|\ell_1,m_1}\mathcal{Y}_{j|\ell_2,m_2}\rangle_{\textsc{hs}}|
			\leq 1\,.
		\end{align*}
		Overall we have
		\begin{align}\label{Eq: Berezin SDQ - product HS orthogonal quantization of spherical harmonics}
			\mathcal{Y}_{j|\ell_1,m_1}\mathcal{Y}_{j|\ell_2,m_2}
			&=\sum_{\ell=|\ell_1-\ell_2|}^{\ell_1+\ell_2}(2\ell+1)
			c_{j|\ell,m_1,m_2}
			\mathcal{Y}_{j|\ell,m_1+m_2}\,,
			&|c_{j|\ell,m}|\leq 1\,,
		\end{align}
		which will play the same role in the proof of Theorem \ref{Thm: uniqueness result for quantum KMS states} of Equation \eqref{Eq: spherical harmonics - expansion of pointwise product} in the proof of Theorem \ref{Thm: uniqueness result for classical KMS states}.
	\end{enumerate}
\end{remark}

\begin{proof}[Proof of Thm. \ref{Thm: uniqueness result for quantum KMS states}]
	Let $\omega_j^{\beta,\Gamma}\in S(B_j^\Gamma)$ be a $(\beta,\delta_j^\Gamma)$-KMS quantum states and let $\underline{\omega}_j^{\beta,\Gamma}\in\underline{\mathsf{X}}$ be its associated element of $\underline{\mathsf{X}}$.
	Proceeding as in the proof of Theorem \ref{Thm: uniqueness result for classical KMS states}, we will prove that $\underline{\omega}_j^{\beta,\Gamma}$ is a solution to a linear equation in $\underline{\mathsf{X}}$ which admits a unique solution under assumption \eqref{Eq: uniqueness result for quantum KMS states - assumption on potential}
	
	To this avail we consider the ordering on $\Gamma$ induced by an arbitrary but fixed bijection $\Gamma\simeq\mathbb{Z}$ and set $x:=\min_{y\in\Lambda}y$ for $\Lambda\Subset\Gamma$.
	Let $\ell_\Lambda\in\mathbb{N}^\Lambda$ and $m_\Lambda\in\mathbb{Z}^\Lambda$ be such that $m_y\in[-\ell_y,\ell_y]\cap\mathbb{Z}$ for all $y\in\Lambda$.
	Equation \eqref{Eq: spherical harmonics - consequence of Shur Lemma for HS orthonormal basis} leads to
	\begin{align*}
		\underline{\omega}_j^{\beta,\Gamma}(\ell_\Lambda,m_\Lambda)
		&=\omega_j^{\beta,\Gamma}(\mathcal{Y}_{j|\ell_\Lambda,m_\Lambda})
		\\
		&=
		\omega_j^{\beta,\Gamma}\bigg(\int_{SU(2)}\Big[\mathcal{Y}_{j|\ell_\Lambda,m_\Lambda}
		-D^{(j)}(R_x)(\mathcal{Y}_{j|\ell_\Lambda,m_\Lambda})D^{(j)}(R_x)^*\Big]\mathrm{d}R_x\bigg)
		\\
		&=
		\omega_j^{\beta,\Gamma}\bigg(\int_{SU(2)}\mathcal{Y}_{j|\ell_\Lambda,m_\Lambda}D^{(j)}(R_x)^*(I-\tau^\Gamma_{i\beta})D^{(j)}(R_x)\mathrm{d}R_x\bigg)\,,
	\end{align*}
	where we used Equation \eqref{Eq: spherical harmonics - consequence of Shur Lemma for HS orthonormal basis} and the $(\beta,\delta_j^\Gamma)$-KMS condition.
	We observe that $D^{(j)}(R_x)\in B_j^x$, thus, $\tau^\Gamma_{i\beta}(D^{(j)}(R_x))$ can be computed using Equation \eqref{Eq: analyticity of local observables} for $\beta<\lambda/2\|\Phi_j\|_\lambda$, where $\Phi_{j,\Lambda}=Q_j^\Lambda(\varphi_\Lambda)$.
	We find
	\begin{multline*}
		\underline{\omega}_j^{\beta,\Gamma}(\ell_\Lambda,m_\Lambda)
		=-\sum_{n\geq 1}\frac{(-\beta)^n}{n!}
		\sum_{\substack{X_1,\ldots,X_n\\X_q\cap S_{q-1}\neq\varnothing}}
		\\
		\omega_j^{\beta,\Gamma}\bigg(\int_{SU(2)}\mathcal{Y}_{j|\ell_\Lambda,m_\Lambda}
		D^{(j)}(R_x)^*
		(\operatorname{ad}_{Q_j^{X_n}(\varphi_{X_n})}
		\cdots\operatorname{ad}_{Q_j^{X_1}(\varphi_{X_1})})
		D^{(j)}(R_x)
		\mathrm{d}R_x\bigg)\,,
	\end{multline*}
	where $S_0:=\{x\}$ and $S_q:=S_{q-1}\cup X_q$ for $q\geq 1$, $\operatorname{ad}_{A}(A'):=[A,A']$.
	Denoting by
	\begin{align*}
		W_n:=
		D^{(j)}(R_x)^*
		\big(\operatorname{ad}_{Q_j^{X_n}(\varphi_{X_n})}
		\cdots\operatorname{ad}_{Q_j^{X_1}(\varphi_{X_1})}\big)
		D^{(j)}(R_x)\,,
	\end{align*}
	we find
	\begin{align}\label{Eq: Wn - W1 explicit formula}
		W_1=D^{(j)}(R_x)^*
		Q_j^{X_1}(\varphi_{X_1})
		 D^{(j)}(R_x)
		-Q_j^{X_1}(\varphi_{X_1})
		=Q_j^{X_1}(\hat{R}_x\varphi_{X_1})
		-Q_j^{X_1}(\varphi_{X_1})\,,
	\end{align}
	and by induction
	\begin{align}\label{Eq: Wn - recursion formula}
		W_q=Q_j^{X_q}(\hat{R}_x\varphi_{X_q})W_{q-1}
		-W_{q-1}Q_j^{X_q}(\varphi_{X_q})
		\qquad
		q=2,\ldots,n\,.
	\end{align}
	Out of Equations \eqref{Eq: Wn - W1 explicit formula}-\eqref{Eq: Wn - recursion formula} one may find a reasonably explicit expression for $W_n$.
	To this avail let
	\begin{align*}
		\psi_X^p:=
		\begin{dcases}
			\hat{R}_x\varphi_X
			& p=+
			\\
			-\varphi_X
			& p=-
		\end{dcases}\,.
	\end{align*}
	We consider the set $\Psi:=\{Q_j^{X_1}(\psi_{X_1}^\pm),\ldots,Q_j^{X_n}(\psi_{X_n}^\pm)\}$ with an order relation $\succ$ defined by
	\begin{align*}
		Q_j^{X_n}(\psi_{X_n}^+)
		\succ Q_j^{X_{n-1}}(\psi_{X_{n-1}}^+)
		\succ\ldots
		\succ Q_j^{X_1}(\psi_{X_1}^+)
		\succ Q_j^{X_1}(\psi_{X_1}^-)
		\succ Q_j^{X_2}(\psi_{X_2}^-)
		\succ\ldots
		\succ Q_j^{X_n}(\psi_{X_n}^-)\,.
	\end{align*}
	For two elements in $A,B\in\Psi$ we set $A\cdot_\succ B:=AB$ if $A\succ B$ and $A\cdot_\succ B:=BA$ if $B\succ A$.
	Then
	\begin{align*}
		W_n=\sum_{p\in\{\pm 1\}^n}
		Q_j^{X_1}(\psi_{X_1}^{p(1)})
		\cdot_\succ\ldots
		\cdot_\succ Q_j^{X_n}(\psi_{X_n}^{p(n)})\,.
	\end{align*}
	At this stage we proceed similarly to Theorem \ref{Thm: uniqueness result for classical KMS states} by expanding each $\psi_X$-term in its Fourier-Laplace expansion, \textit{cf.} Remark \ref{Rmk: spherical harmonics - set-up, identities and useful bounds}.
	In particular we have for all $k\in\{1,\ldots,n\}$ and $p\in\{+,-\}$,
	\begin{multline*}
		Q_j^{X_k}(\psi_{X_k}^p)
		=\sum_{\ell_{X_k,p},m_{X_k,p}}
		C_{j|X_k,p}(\ell_{X_k,p},m_{X_k,p})
		\mathcal{Y}_{j|\ell_{X_k,p},m_{X_k,p}}
		\\
		C_{j|X_k,p}(\ell_{X_k,p},m_{X_k,p})
		:=\Big(\prod_{y\in X_k}c_{j,\ell_{X_k,p,y}}^{1/2}(2\ell_{X_k,p,y}+1)\Big)
		\big\langle
		Y_{\ell_{X_k,p},m_{X_k,p}}\big|\psi_{X_k}^p\big\rangle_{L^2(\mathbb{S}^2_{X_j},\mu_0^{X_j})}\,.
	\end{multline*}
	Thus, we find
		\begin{multline*}
		\underline{\omega}_j^{\beta,\Gamma}(\ell_\Lambda,m_\Lambda)
		=-\sum_{n\geq 1}\frac{(-\beta)^n}{n!}
		\sum_{\substack{X_1,\ldots,X_n\\X_q\cap S_{q-1}\neq\varnothing}}
		\sum_{p\in\{\pm 1\}^n}
		\sum_{\substack{\ell_{X_1,p},\ldots,\ell_{X_n,p}\\m_{X_1,p},\ldots,m_{X_n,p}}}
		\\
		\int_{SU(2)}
		\prod_{k=1}^n C_{X_k,p}(\ell_{X_k,p},m_{X_k,p})
		\mathrm{d}R_x\;
		\omega_j^{\beta,\Gamma}\Big(
		\mathcal{Y}_{j|\ell_\Lambda,m_\Lambda}
		\mathcal{Y}_{j|\ell_{X_1,p(1)},m_{X_1,p(1)}}
		\cdot_\succ\ldots\cdot_\succ
		\mathcal{Y}_{j|\ell_{X_n,p(n)},m_{X_1,p(n)}}
		\
		\Big)\,.
	\end{multline*}
	We then expand the product of the $\mathcal{Y}$'s factors by means of Equation \eqref{Eq: Berezin SDQ - product HS orthogonal quantization of spherical harmonics} in exactly the same way we did for the classical spherical harmonics.
	Setting $S_n:=X_1\cup\ldots\cup X_n$ we find
	\begin{align*}
		&\mathcal{Y}_{j|\ell_\Lambda,m_\Lambda}
		\mathcal{Y}_{j|\ell_{X_1,p(1)},m_{X_1,p(1)}}
		\cdots_\succ
		\mathcal{Y}_{j|\ell_{X_n,p(n)},m_{X_1,p(n)}}
		\\
		&=\Big(\prod_{y\in\Lambda\cap S_n^c} \mathcal{Y}_{j|\ell_{\Lambda,y},m_{\Lambda,y}}\Big)
		\prod_{y\in S_n}
		\sum_{s_{y,1},\ldots,s_{y,n}}
		\Big(
		\prod_{k=1}^n c_{j,s_{y,k}}
		\Big)
		\mathcal{Y}_{s_{y,n},\tilde{m}_{\Lambda,y}+\tilde{m}_{X_1,p(1),y}+\ldots+\tilde{m}_{X_n,p(n),y}}
		\\
		&=\Big(\prod_{y\in\Lambda\cap S_n^c} \mathcal{Y}_{j|\ell_{\Lambda,y},m_{\Lambda,y}}\Big)
		\sum_{k=1}^{N(\ell_{X_1,p},\ldots,\ell_{X_n,p})}
		C_{j,p}'(\ell_{S_n}^k,m_{S_n}^k)
		\mathcal{Y}_{j|\ell_{S_n}^k,m_{S_n}^k}\,,
	\end{align*}
	where $|C_{j,p}'(\ell_{S_n}^k,m_{S_n}^k)|\leq 1$ while
	\begin{align*}
		N(\ell_{X_1,p},\ldots,\ell_{X_n,p})
		\leq\prod_{k=1}^n\prod_{y\in X_k}(2\ell_{X_k,p,y}+1)\,.
	\end{align*}
	Setting again $X(n,\Lambda):=\Lambda\cup S_n$ we have that for any $(\beta,\delta_j^\Gamma)$-KMS quantum states $\omega_j^{\beta,\Gamma}\in S(B_j^\Gamma)$ the corresponding element $\underline{\omega}_j^{\beta,\Gamma}\in\underline{\mathsf{X}}$ solves the linear equation
	\begin{align}\label{Eq: Banach space for state - linear equation for quantum uniqueness}
		(I-L_j^\beta)\underline{\omega}=\underline{\delta}\,.
	\end{align}
	where $\underline{\delta}\in\underline{\mathsf{X}}$ has been defined in Equation \eqref{Eq: Banach space for state - non-homogeneous term}.
	The operator $L_j^\beta\colon\underline{\mathsf{X}}\to\underline{\mathsf{X}}$ is defined by setting, for all $\underline{f}\in\underline{\mathsf{X}}$, $(L_j^\beta\underline{f})_\varnothing=0$ and
	\begin{multline}\label{Eq: Banach space for state - linear operator for quantum uniqueness}
		(L_j^\beta\underline{f})_\Lambda(\ell_\Lambda,m_\Lambda)
		\\
		=-\sum_{n\geq 1}\frac{(-\beta)^n}{n!}
		\sum_{\substack{X_1,\ldots,X_n\\X_q\cap S_{q-1}\neq\varnothing}}
		\sum_{p\in\{\pm 1\}^n}
		\sum_{\substack{\ell_{X_1,p},\ldots,\ell_{X_n,p}\\m_{X_1,p},\ldots,m_{X_n,p}}}
		\int_{SU(2)}
		\prod_{k=1}^n C_{j|X_k,p}(\ell_{X_k,p},m_{X_k,p})
		\mathrm{d}R_x
		\\
		\sum_{k=1}^{N(\ell_{X_1},\ldots,\ell_{X_n})}
		C_{j,p}'(\ell_{S_n}^k,m_{S_n}^k)
		f_{X(n,\Lambda)}(\ell_{X(n,\Lambda),p}^k,m_{X(n,\Lambda),p}^k)\,,
	\end{multline}
	for all non-empty $\Lambda\Subset\Gamma$ where $x:=\min_{y\in\Lambda}y$ while $\ell_{X(n,\Lambda),p,y}=\ell_{\Lambda,y}$ for all $y\in\Lambda\cap S_n^c$ and $\ell_{X(n,\Lambda),p,y}=\ell_{S_n,p,y}$ for $y\in S_n$.
	
	It remains to prove that Equation \eqref{Eq: Banach space for state - linear equation for quantum uniqueness} has a unique solution under assumption \eqref{Eq: uniqueness result for quantum KMS states - assumption on potential}.
	To this avail we observe that
	\begin{multline*}
		|(L_j^\beta\underline{f})_\Lambda(\ell_\Lambda,m_\Lambda)|
		\leq\|\underline{f}\|_{\underline{\mathsf{X}}}\sum_{n\geq 1}\frac{\beta^n}{n!}
		\sum_{\substack{X_1,\ldots,X_n\\X_q\cap S_{q-1}\neq\varnothing}}
		\sum_{p\in\{\pm 1\}^n}
		\sum_{\substack{\ell_{X_1,p},\ldots,\ell_{X_n,p}\\m_{X_1,p},\ldots,m_{X_n,p}}}
		\\
		\int_{SU(2)}
		\prod_{k=1}^n |C_{j|X_k,p}(\ell_{X_k,p},m_{X_k,p})|
		\prod_{y\in X_k}(2\ell_{X_k,p,y}+1)
		\mathrm{d}R_x\,,
	\end{multline*}
	where we used the bound on $C_{j,p}'$ and $N(\ell_{X_1,p},\ldots,\ell_{X_n,p})$.
	Moreover, again as in the proof of Theorem \ref{Thm: uniqueness result for classical KMS states}, we find
	\begin{align*}
		&\sum_{\substack{\ell_{X_1,p},\ldots,\ell_{X_n,p}\\m_{X_1,p},\ldots,m_{X_n,p}}}
		\int_{SU(2)}
		\prod_{k=1}^n |C_{j|X_k,p}(\ell_{X_k,p},m_{X_k,p})|
		\prod_{y\in X_k}(2\ell_{X_k,p,y}+1)
		\mathrm{d}R_x
		\\
		&\leq
		\sum_{\ell_{X_1,p},\ldots,\ell_{X_n,p}}
		\prod_{k=1}^n
		\Big(\prod_{y\in X_k}
		\frac{(2\ell_{X_k,p,y}+1)^{5/2}}{[1+\ell_{X_k,p,y}(\ell_{X_k,p,y}+1)]^s}
		\Big)
		C_\Delta^{s|X_k|}
		\|\varphi_{X_k}\|_{C^{2s}(\mathbb{S}^2_{X_k})}
		\\
		&=\prod_{k=1}^n
		(K_sC_\Delta^s)^{|X_k|}
		\|\varphi_{X_k}\|_{C^{2s}(\mathbb{S}^2_{X_k})}\,,
	\end{align*}
	where $C_\Delta$ and $K_s$ have been defined in Equation \eqref{Eq: uniqueness result for classical KMS states - useful constants}.
	The above estimate is uniform over $p\in\{\pm 1\}^n$, therefore,
	\begin{align*}
		|(L_j^\beta\underline{f})_\Lambda(\ell_\Lambda,m_\Lambda)|
		\leq\|\underline{f}\|_{\underline{\mathsf{X}}}\sum_{n\geq 1}\frac{(2\beta)^n}{n!}
		\sum_{\substack{X_1,\ldots,X_n\\X_q\cap S_{q-1}\neq\varnothing}}
		\prod_{k=1}^n
		(K_sC_\Delta^s)^{|X_k|}
		\|\varphi_{X_k}\|_{C^{2s}(\mathbb{S}^2_{X_k})}\,.
	\end{align*}
	Finally we apply the following estimate, \textit{cf.} \cite[Prop. 6.2.45]{Bratteli_Robinson_97}: If $\alpha_X\in\mathbb{R}_+$ for all $X\Subset\Gamma$ then for all $S\Subset\Gamma$
	\begin{align*}
		\sum_{X\cap S\neq\varnothing}\alpha_X
		\leq\sum_{x\in S}\sum_{X\ni x}\alpha_X
		=\sum_{x\in S}\sum_{m\geq 0}\sum_{\substack{|X|=m+1\\X\ni x}}\alpha_X
		\leq|S|\sum_{m\geq 0}\sup_{x\in\Gamma}
		\sum_{\substack{|X|=m+1\\ X\ni x}}\alpha_X\,,
	\end{align*}
	and by iteration, for $S_1=\{x\}$ and $S_q:=X_q\cup S_{q-1}$,
	\begin{align*}
		\sum_{\substack{X_1,\ldots,X_n\\ X_q\cap S_{q-1}\neq \varnothing}}\alpha_{X_1}\cdots\alpha_{X_n}
		&\leq\sum_{m_1\geq 0}\sup_{y\in\Gamma}
		\sum_{\substack{|X_1|=m_1+1\\X_1\ni y}}
		\alpha_{X_1}
		\sum_{\substack{X_2,\ldots,X_n\\ X_q\cap S_{q-1}\neq \varnothing}}\alpha_{X_2}\cdots\alpha_{X_n}
		\\
		&\leq\sum_{m_1,\ldots,m_n\geq 0}
		\prod_{k=1}^n(1+m_1+\ldots+m_{k-1})
		\sup_{y\in\Gamma}
		\sum_{\substack{|X_k|=m_k+1\\X_k\ni y}}
		\alpha_{X_k}
		\\
		&\leq n!\varepsilon^{-n}e^\varepsilon
		\bigg(\sum_{m\geq 0}
		e^{\varepsilon m}
		\sup_{y\in\Gamma}
		\sum_{\substack{|X|=m+1\\X\ni y}}
		\alpha_X\bigg)^n\,,
	\end{align*}
	where in the last line we observe that, for all $\lambda>0$,
	\begin{align*}
		\prod_{k=1}^n(1+m_1+\ldots+m_{k-1})
		\leq(1+m_1+\ldots+m_n)^n
		\leq n!\varepsilon^{-n}e^{\varepsilon(1+m_1+\ldots+m_n)}\,.
	\end{align*}
	Therefore, setting $\alpha_X:=(K_sC_\Delta^s)^{|X|}\|\varphi_X\|_{C^{2s}(\mathbb{S}^2_X)}$ we find
	\begin{align*}
		|(L_j^\beta\underline{f})_\Lambda(\ell_\Lambda,m_\Lambda)|
		&\leq\|\underline{f}\|_{\underline{\mathsf{X}}}e^\varepsilon\sum_{n\geq 1}(2\varepsilon^{-1}\beta K_sC_\Delta^s)^n
		\bigg(
		\sum_{m\geq 0}(e^\varepsilon K_sC_\Delta^s)^m
		\sup_{y\in\Gamma}
		\sum_{\substack{|X|=m+1\\X\ni y}}
		\|\varphi_X\|_{C^{2s}(\mathbb{S}^2_X)}
		\bigg)^n
		\\
		&=\|\underline{f}\|_{\underline{\mathsf{X}}}e^\varepsilon\sum_{n\geq 1}(2\varepsilon^{-1}\beta K_sC_\Delta^s\|\varphi\|_{\varepsilon,s})^n
		\\
		&=\|\underline{f}\|_{\underline{\mathsf{X}}}e^\varepsilon
		\frac{2\varepsilon^{-1}\beta K_sC_\Delta^s\|\varphi\|_{\varepsilon,s}}{1-2\varepsilon^{-1}\beta K_sC_\Delta^s\|\varphi\|_{\varepsilon,s}}\,,
	\end{align*}
	where we used assumption \eqref{Eq: uniqueness result for quantum KMS states - assumption on potential} and considered $\beta<\varepsilon/2K_sC_\Delta^s\|\varphi\|_{\varepsilon,s}$ ---notice that the latter value is lower than the previous bound $\beta<\varepsilon/2\|\Phi_j\|_\varepsilon$ necessary to ensure the expansion of $\tau_{i\beta}^\Gamma(D^{(j)}(R_x))$ according to Equation \eqref{Eq: analyticity of local observables}.
	It follows that $L_j^\beta\in\mathcal{B}(\underline{\mathsf{X}})$, moreover,
	\begin{align*}
		\|L_j^\beta\|_{\underline{\mathsf{X}}}
		\leq e^\varepsilon\frac{2\varepsilon^{-1}\beta K_sC_\Delta^s\|\varphi\|_{\varepsilon,s}}{1-2\varepsilon^{-1}\beta K_sC_\Delta^s\|\varphi\|_{\varepsilon,s}}
		<1\,,
	\end{align*}
	provided that $\beta<\beta_{\varepsilon,s}$, $\beta_{\varepsilon,s}$ being defined by Equation \eqref{Eq: uniqueness result for quantum KMS states - estimate on critical temperature}.
\end{proof}

\begin{remark}\label{Rmk: NO CPT implies NO QPT; comparison with BR result}
	\noindent
	\begin{enumerate}[(i)]
		\item
		By direct inspection assumption \eqref{Eq: uniqueness result for quantum KMS states - assumption on potential} implies \eqref{Eq: uniqueness result for classical KMS states - useful constants}.
        Thus, condition \eqref{Eq: uniqueness result for quantum KMS states - assumption on potential} is a sufficient condition which guarantees uniqueness of both $(\beta,\delta_\infty^\Gamma)$-KMS classical states and $(\beta,\delta_j^\Gamma)$-KMS quantum states for all $j\in\mathbb{Z}_+/2$ and for $\beta\leq\beta_{\varepsilon,s}$.
        In other words, \eqref{Eq: uniqueness result for quantum KMS states - assumption on potential} ensures the absence of both classical and quantum phase transitions at sufficiently high temperature.
        Furthermore, since $\|\varphi\|_{0,s}\leq\|\varphi\|_{\varepsilon,s}$ and $\varepsilon(1+e^\varepsilon)^{-1}<\log 2$, it follows that $\beta_{\varepsilon,s}<\beta_{0,s}$, that is, the corresponding quantum inverse critical temperature is slightly lower with respect to the corresponding classical one ---This ensures absence of phase transition starting from a common critical inverse temperature.
        
		Moreover, on account of Proposition \ref{Thm: limit points of quantum KMS in oo volume are classical KMS in oo volume}, \textit{cf.} Remark \ref{Rmk: classical limit of quantum KMS in absence of phase transition}, in this situation the classical limit $\lim_{j\to\infty}\omega_j^{\beta,\Gamma}\circ Q_j^\Gamma$ of the unique $(\beta,\delta_j^\Gamma)$-KMS quantum states $\omega_j^\Gamma\in S(B_j^\Gamma)$ coincides with the unique $(\beta,\delta_\infty^\Gamma)$-KMS classical states $\omega_\infty^{\beta,\Gamma}\in S(B_\infty^\Gamma)$.
        
		\item
		It is worth to compare Theorem \ref{Thm: uniqueness result for quantum KMS states} with \cite[Prop. 6.2.45]{Bratteli_Robinson_97}.
		The latter provide a sufficient condition for the uniqueness of $(\beta,\delta_j^\Gamma)$-KMS for fixed $j\in\mathbb{Z}_+/2$ which is similar in spirit to \eqref{Eq: uniqueness result for quantum KMS states - assumption on potential} ---in fact, Theorem \ref{Thm: uniqueness result for quantum KMS states} has been inspired by this latter result--- namely
		\begin{align}\label{Eq: uniqueness result for quantum KMS states - BR assumption on potential}
			\|\varphi\|_{\textsc{br},\varepsilon,j}
			:=\sum_{m\geq 0}e^{\varepsilon m}(2j+1)^m
			\sum_{\substack{|X|=m+1\\X\ni x}}\|Q_j^X(\varphi_X)\|_{B_j^X}<\infty\,.
		\end{align}
		Condition \eqref{Eq: uniqueness result for quantum KMS states - BR assumption on potential} is stronger than \eqref{Eq: uniqueness result for quantum KMS states - assumption on potential} because it only uses the $B_j^X$-norm of $Q_j^X(\varphi_X)$.
		However, it is not uniform of $j\in\mathbb{Z}_+/2$, in particular, it requires a faster and faster decay behaviour of the potential $\{\varphi_\Lambda\}_{\Lambda\Subset\mathbb{Z}^d}$ as $j\to\infty$.
		Moreover, the critical quantum inverse temperature $\beta_{\textsc{br}}(j,\lambda)$ predicted by \cite[Prop. 6.2.45]{Bratteli_Robinson_97} vanishes as $j\to\infty$.
		Finally, there is no classical version of condition \eqref{Eq: uniqueness result for quantum KMS states - BR assumption on potential}.
		For all these reasons \cite[Prop. 6.2.45]{Bratteli_Robinson_97} is not suitable for the comparison with the classical setting we are interested in.
		
		Instead, Theorem \ref{Thm: uniqueness result for quantum KMS states} leads to a result which is uniform in $j$, allowing for a simpler comparison with Theorem \ref{Thm: uniqueness result for classical KMS states}.
		The latter theorem can be understood as a classical counterpart of the uniqueness result presented in \cite{Bratteli_Robinson_97}.
		From a technical point of view the uniform behaviour in $j$ is obtained by trading the $B_\infty$-norm with the $C^{2s}$-norm for a suitably high $s$.
	\end{enumerate}
\end{remark}

\appendix

\end{document}